\documentclass{CSML}

\def\dOi{13(1:3)2017}
\lmcsheading%
{\dOi}
{1--51}
{}
{}
{Jul.~11, 2014}
{Feb.~\phantom02, 2017}
{}


\usepackage{hyperref}

\usepackage{amsmath,amsfonts,latexsym,amssymb,url}
\usepackage[all,curve]{xy}
\usepackage{comment}
\usepackage{xcolor}

\newcommand{\cat}[1]{\mathsf{#1}}
\newcommand{\bb}[1]{\mathbb{#1}}
\newcommand{\dsty}{\displaystyle}

\newcommand{\blue}[1]{\textcolor{blue}{#1}}

\newcommand{\Set}{\cat{Set}}   

\newcommand{\Pow}{\mathcal{P}}

\newcommand{\Id}{\mathrm{Id}}
\newcommand{\id}{\mathrm{id}}

\newcommand{\T}{\mathcal{T}}

\newcommand{\M}{\mathcal{M}} 

\newcommand{\bbN}{\mathbb{N}}
\newcommand{\bbZ}{\mathbb{Z}}

\newcommand{\bbQodd}{\mathbb{Q}_\text{odd}}
\newcommand{\bbR}{\mathbb{R}}

\newcommand{\Mat}{\mathsf{M}}
\newcommand{\pmat}[1]{\begin{pmatrix}#1\end{pmatrix}}

\renewcommand{\to}{\rightarrow}		
\newcommand{\To}{\Longrightarrow}		
\newcommand{\Ra}{\Rightarrow}

\newcommand{\injR}{\hookrightarrow}

\newcommand{\abstractgoes}[2]{%
\setbox0=\hbox{\ ${\scriptstyle#2}$\ }
\ifdim\wd0<12pt\wd0=12pt\fi
\mathrel{\stackrel{#2}{\rule[2.2pt]{\wd0}{0.6pt}}\mkern-16mu{#1}}
}

\newcommand{\sgoes}[1]{\stackrel{#1}{\longrightarrow}}

\newcommand{\sse}{\subseteq}

\newcommand{\restrict}{\!\!\upharpoonright}  

\newcommand{\tup}[1]{\langle #1 \rangle}
\newcommand{\ol}[1]{\overline{#1}}  

\newcommand{\strms}[1]{{#1}^\omega}
\newcommand{\sig}{\sigma} 
\newcommand{\X}{\mathsf{X}} 
\newcommand{\cns}[1]{[{#1}]} 
\DeclareMathOperator{\hd}{\mathit{hd}}		
\DeclareMathOperator{\tl}{\mathit{tl}}		
\newcommand{\beh}[1]{[\![#1]\!]}		

\newcommand{\nm}[1]{\mathsf{#1}}

\newcommand{\ones}{\nm{ones}}
\newcommand{\nats}{\nm{nats}}

\newcommand{\tail}{\mathrm{tail}}
\newcommand{\bbin}{\mathsf{bbin}}

\newcommand{\Sig}{\Sigma} 
\newcommand{\TSig}{T_\Sigma} 
\newcommand{\monTSig}{{\T}_\Sigma} 

\newcommand{\es}{e} 

\newcommand{\V}{\mathcal{V}} 

\newcommand{\emptyword}{\epsilon}

\newcommand{\zero}[1]{#1(0)}

\newcommand{\term}{\mathit{T}}

\newcommand{\coloneqq}{\mathrel{:=}}
\newcommand{\sdefi}{\mathcal{D}}
\newcommand{\osdefi}{o_{\sdefi}}
\newcommand{\dsdefi}{d_{\sdefi}}

\newcommand{\sem}[1]{[\![ #1 ]\!]}
\newcommand{\causal}{\mathsf{C}}

\newcommand{\syn}[1]{\underline{#1}}

\newcommand{\ext}[1]{#1^*}  
\newcommand{\interp}[1]{\overline{#1}}  

\newcommand{\even}{\mathsf{even}}
\newcommand{\odd}{\mathsf{odd}}
\newcommand{\zip}{\mathsf{zip}}

\theoremstyle{plain}\newtheorem{rema}[thm]{Remark}

\title[Stream Differential Equations]{Stream Differential Equations:\\ Specification Formats and Solution Methods}

\author[H.~H.~Hansen]{Helle Hvid Hansen\rsuper a}	
\address{{\lsuper a}Delft University of Technology and Centrum Wiskunde \& Informatica, Amsterdam}	
\email{h.h.hansen@tudelft.nl}  
\thanks{{\lsuper a}Supported by NWO-Veni grant 639.021.231.}	

\author[C.~Kupke]{Clemens Kupke\rsuper b}	
\address{{\lsuper b}University of Strathclyde}	
\email{clemens.kupke@strath.ac.uk}  
\thanks{{\lsuper b}Supported by EPSRC grant EP/N015843/1.}	

\author[J.~Rutten]{Jan Rutten\rsuper c}	
\address{{\lsuper c}Centrum Wiskunde \& Informatica, Amsterdam, and Radboud University Nijmegen}	
\email{jjmmrutten@gmail.com}  


\begin{document}

\begin{abstract}
Streams, or infinite sequences, are infinite objects of a very simple type,  
yet they have a rich theory partly due to their ubiquity in mathematics and 
computer science.
Stream differential equations are a coinductive method for specifying streams
and stream operations, and
their theory has been developed in many papers over the past two decades.
In this paper we present a survey of the many results in this area.
Our focus is on the classification of different formats of stream differential equations,
their solution methods, and the classes of streams they can define.
Moreover, we describe in detail the connection between the so-called syntactic solution method
and abstract GSOS.
\end{abstract}

\maketitle


\section{Introduction}

Streams, or infinite sequences, are infinite objects of
a very simple type, yet they
have a rich theory partly due to their ubiquity.
Streams occur as numerical expansions, data sequences, formal power series,
limit sequences, dynamic system behaviour, formal languages,
ongoing computations, and much more.

Defining the \emph{stream derivative} of a stream
$\sig = (\sig(0),\sig(1),\sig(2), \ldots)$ by
\[
\sig' =  (\sig(1),\sig(2), \sig(3), \ldots)
\]
and the \emph{initial value} of $\sig$ by $\sig(0)$,  one can 
develop a \emph{calculus of streams} in close analogy with classical
calculus in mathematical analysis. Notably,
using the notions of stream derivative and initial value,
we can specify streams by means of \emph{stream differential equations}.

For instance,
the stream differential equation $\sig(0) = 1,\; \sig' = \sig$
has the stream $\sig = (1,1,1,\ldots)$ as its unique solution,
and $\sig(0) = 1,\; \sig' = \sig + \sig$,
where $+$ is the elementwise addition of two streams, defines 
the stream $(2^0,2^1,2^2, \ldots)$ .
Similarly, we can specify stream functions. For example,
$f(\sig)(0) = \sig(0),\; f(\sig)' = f(\sig'')$
(this time using second-order derivatives) defines the function
$f(\sig) = (\sig(0), \sig(2), \sig(4) ,\ldots)$.
In these examples, it is easy to see that the stream differential equations
have a unique solution.
But how about $\tau(0) = 0,\; \tau' = f(\tau)$? A moment's thought reveals that
this equation has several solutions, e.g. $\tau = (0,0,0,\ldots)$
and $\tau = (0,0,1,1,1,\ldots)$.
But what is the difference between this equation and the previous ones?
How can we ensure the existence of unique solutions?
Which classes of streams can be defined using a finite amount of information?
These questions have been studied by several authors in many different
contexts in recent years,
and have led to notions such as rational streams,
context-free streams, and new insights into automatic and regular sequences.

In this paper, we present an overview of the current state-of-the-art 
in formats and solution methods for stream differential equations.
The theoretical basis for stream differential equations is given by 
coalgebra~\cite{Rut:univ-coalg}, but our aim is to give an elementary and
self-contained overview. 
We consider our contribution to be a unified and uniform presentation
of results which are collected from many different sources.
As modest new insights, we mention the results on the expressiveness of 
non-standard formats in Section~\ref{sec:non-standard}.
Another contribution which can be considered new, 
is the detailed analysis of the connection
between the syntactic method and abstract GSOS (in Section~\ref{sec:stream-gsos})
This connection is rather obvious to readers familiar with abstract GSOS, 
but probably less so to the uninitiated reader.

\emph{Overview:}
We start by giving an informal introduction to stream differential equations
in Section~\ref{sec:sde-intro}, and in Section~\ref{sec:prelims} we provide
some basic definitions regarding automata and stream calculus.
Next, in Sections~\ref{sec:simple-specs} through \ref{sec:non-standard},
we shall study in more detail various  
types of stream differential equations, each corresponding to a 
specification format. 
We describe 
solution methods for each of these formats, and
characterise the automata and the classes of streams 
that these families of stream differential equations can define.
The following little table contains some representative examples,
corresponding to Sections \ref{sec:simple-specs} through \ref{sec:non-standard}:
\[
\begin{array}{l|l|l|l}
  \textcolor[rgb]{0.00,0.00,0.00}{\mbox{initial value:}} & \textcolor[rgb]{0.00,0.00,0.00}{\mbox{derivative:}} &
  \textcolor[rgb]{0.00,0.00,0.00}{\mbox{solution:}}  &
    \textcolor[rgb]{0.00,0.00,0.00}{\mbox{type of equation:}}
  \\
  \hline
  \sigma(0) = 1 & \sigma' = \sigma & (1,1,1, \ldots ) &
  \textcolor[rgb]{0.00,0.00,0.00}{\mbox{simple}}
  \\
  \sigma(0) = 1 & \sigma' = \sigma + \sigma  & (2^0,2^1,2^2, \ldots ) &
  \textcolor[rgb]{0.00,0.00,0.00}{\mbox{linear}}
  \\
  \sigma(0) = 1 & \sigma' = \sigma \times \sigma \;\;\; &
  \textcolor[rgb]{0.00,0.00,0.00}{\mbox{Catalan numbers}} &
  \textcolor[rgb]{0.00,0.00,0.00}{\mbox{context-free/algebraic}}
  \\
  \sigma(0) = 1 & \frac{d }{dX}(\sigma) =  \sigma \;\;\; &
  \left(\frac{1}{0!}, \frac{1}{1!}, \frac{1}{2!}, \frac{1}{3!}, \frac{1}{4!},\dots \right)&
  \textcolor[rgb]{0.00,0.00,0.00}{\mbox{non-standard}}
\end{array}
\]
(For the definition of the convolution product $\times$ see
(\ref{eq:convprod}) in Section \ref{sec:sde-intro};  the non-standard derivative
$\frac{d}{dX}$ is defined in Example \ref{ex:representations}.)

In Section~\ref{sec:syntactic-method}, we describe a
concrete syntactic solution method 
for a large class of well-formed stream differential equations, including all those that
we discussed in Sections \ref{sec:simple-specs}-\ref{sec:context-free-specs}.
Finally, in Section \ref{sec:stream-gsos}, a more general, categorical perspective on the theory
of stream differential equations is presented.
In particular, this section places streams and automata in a more general context of algebras, coalgebras and 
so-called distributive laws. 
Note that this is the only section that requires
some basic knowledge of category theory. 
In Section~\ref{sec:rel-work}, we briefly discuss connections with other methods 
for representing streams, 
such as recurrence relations, generating functions and so on.

\emph{Section Interdependency}:
Sections \ref{sec:sde-intro}-\ref{sec:prelims} provide the reader with important prerequisites for the remainder of the article.
Sections~\ref{sec:simple-specs}-\ref{sec:non-standard} can be read independently of each other.
Section 8 can, in principle, be read without Sections \ref{sec:simple-specs}-\ref{sec:non-standard}, but it refers back to earlier sections for examples and motivations.
Section \ref{sec:stream-gsos} can be skipped by readers who are mainly interested in concrete specification formats.
Section 10 relies on Sections \ref{sec:sde-intro}-\ref{sec:non-standard}.

\emph{Related work}:
Here we mention the most important origins of the results in this paper.
A more extensive discussion of related work is found in 
Section~\ref{sec:rel-work}.
Stream differential equations \cite{Rut05:MSCS-stream-calc} came about as a special instance of behavioural
differential equations, for formal power series, which were introduced in \cite{Rut03:TCS-bde}.
Motivation came from the coalgebraic perspective on infinite data structures,
in which streams are a canonical example, but also  from work on language dervatives in classical automata theory,
notably \cite{Brz64} and \cite{Con71}. The idea of developing a calculus of streams in close analogy
to analysis was further inspired by the work on classical calculus in coinductive form in \cite{PavlovicE98}.
The classification of stream differential equations into the families of simple, linear and context-free
systems stems from our joint work with Marcello Bonsangue and Joost Winter, in \cite{BRW:CF-pow} and \cite{WBR:CF-LMCS},
on classifications of behavioural differential equations for streams, languages and formal power series.
The results on non-standard stream calculus come from \cite{KR:cocoop,KNR:SDE},
and the examples on automatic and regular sequences from \cite{KR:autseq}
and \cite{HKRW:k-regular}, respectively.

\emph{Acknowledgements}: 
It should be clear 
from the many references to the literature that our paper builds on the work of many
others. We are, in particular, much indebted to Marcello Bonsangue and Joost Winter for
many years of fruitful  collaboration on 
stream differential equations. A large part of the work presented here
was developed in joint work with them. 
We are also grateful to many other colleagues, with whom we have
worked together in different ways on ideas relating to streams, including: 
Henning Basold, Filippo Bonchi, 
J\"org Endrullis, Herman Geuvers, Clemens Grabmayer, Dimitri Hendriks,
Bart Jacobs, Bartek Klin, Jan-Willem Klop, Dorel Lucanu, Larry Moss, Milad Niqui, 
Grigore Rosu,
Jurriaan Rot, Alexandra Silva, Hans Zantema.

\tableofcontents

\section{Stream Differential Equations}
\label{sec:sde-intro}

In this section, we present several examples of
{\em stream differential equations (SDEs)} and their solutions.
For now the purpose is to get familiarised with the notation of SDEs.
Detailed proofs and solution methods are presented later.

We start by introducing notation and basic definitions on
streams.

\subsection{Basic definitions}
\label{ssec:stream-basics}

A stream over a given set $A$ is a function $\sigma \colon \bbN \to A$
from the natural numbers to $A$, which we will sometimes write as
\[ \sigma = (\sigma(0), \sigma(1), \sigma(2), \ldots) \]
The set of all {streams over $A$} is denoted by
\[
A^\omega = \{ \sigma \mid \sigma\colon \bbN \to A\}
\]
Given a stream  $\sig \in A^\omega$,
we define the {\it initial value} of $\sig$
as $\sig(0)$,
and the {\it derivative} of $\sig$ as the stream
$\sig' = (\sig(1), \sig(2), \sig(3), \ldots)$.
For $a \in A$ and $\sig\in\strms{A}$,
we define $a\!:\!\sig = (a,\sig(0),\sig(1),\sig(2),\ldots)$.
Higher order derivatives $\sig^{(n)}$
are defined inductively for all $n \in \bbN$ by:
\[ \sig^{(0)} = \sigma \quad  \quad \sig^{(n+1)} = (\sig^{(n)})'
\]
Initial value and derivative are also known as
{\it head} and {\it tail}, respectively.

\subsection{Simple examples}

Stream differential equations define a stream in terms of its initial value
and its derivative(s).
As a first elementary example, consider
\begin{equation}\label{eq:sde-ones}
\sig(0)=1, \quad \sig'=\sig.
\end{equation}
which has the stream $\ones = (1,1,1,\ldots)$
as the unique solution.

For a slightly more interesting example consisting of two SDEs
over two stream variables, consider
\begin{equation}\label{eq:sde-alt}
\begin{array}[t]{lcl}
\sig(0)=1, &\quad& \sig'=\tau\\
\tau(0)=0, &\quad& \tau'=\sig
\end{array}
\end{equation}
whose solution is $\sig=(1,0,1,0,\ldots)$ and
$\tau= (0,1,0,1,\ldots)$.

In the above SDEs, derivatives given by a stream variable. Such SDEs are called simple,
and in Section~\ref{sec:simple-specs}, we will characterise the class
of streams that can be specified by finite systems of simple equations.
More generally, we will consider SDEs involving not only variables,
but also operations on streams.

\subsection{Stream operations}
\label{ssec:stream-ops}

We illustrate how to define stream operations, and at the same time
introduce a bit of stream calculus.
Stream calculus is usually defined for streams over the real numbers $\bbR$,
but most definitions hold for more general data domains $A$.

Consider the set $A^\omega$ of streams over a ring $(A,+,-,\cdot,0,1)$.
The stream differential equation:
\begin{equation}\label{eq:sde-plus}
\begin{array}{lclclcl}
(\sig+\tau)(0) &=& \sig(0)+\tau(0),
&\quad &
(\sig+\tau)' &=& \sig'+\tau'
\end{array}
\end{equation}
defines the element-wise addition of two streams,
that is, for all $\sig,\tau\in \strms{A}$,
\begin{equation}\label{eq:ind-plus}
(\sig+\tau)(n) = \sig(n)+\tau(n) \qquad \text{ for all } n\in \bbN.
\end{equation}
(Note that we use the same symbol to denote addition in $A$ and
addition of streams. The typing should be clear from the context.)

Similarly, one can define the element-wise multiplication with a
scalar $a \in A$ with the SDE:
\begin{equation}\label{eq:sde-scalar-mult}
(a\cdot\sig)(0) = a \cdot \sig(0), \qquad
(a\cdot\sig)' = a\cdot\sig'.
\end{equation}
where $ a \cdot \sig(0)$ denotes multiplication in $A$.
It follows that,  for all $\sig\in \strms{A}$,
\[
 (a \cdot \sig)(n) = a\cdot\sig(n) \qquad \text{ for all } n\in \bbN.
\]
Clearly, any element-wise operation on streams can be defined
in a similar manner.
An example of a non-element-wise operation is the {\it convolution product}
of streams given explicitly by:
\begin{equation}\label{eq:convprod}
 (\sig\times\tau)(n) \;\;=\;\; \sum_{k=0}^n \sig(k)\cdot\tau(n-k)
  \qquad \text{ for all } n\in \bbN
\end{equation}
which is defined by the SDE
\begin{equation}\label{eq:sde-convprod}
\begin{array}{lcl}
(\sig\times\tau)(0) =\sig(0)\cdot\tau(0),
&\quad &
(\sig\times\tau)' = (\sig'\times\tau) + (\cns{\sig(0)}\times\tau')
\end{array}
\end{equation}
where for $a \in A$,
\begin{equation}\label{eq:sde-constant}
\cns{a}(0) = a,\qquad\cns{a}'= \cns{0}
\end{equation}
Note that the stream $[1] = (1,0,0,0, \ldots)$ is the identity for the
convolution product, that is, $\sigma \times [1] = [1] \times \sigma = \sigma$.

One can show that the convolution product is commutative if and only if the multiplication in $A$ is commutative.

Using the convolution product, and taking $A = \mathbb{Z}$ we can form the following
SDE: 
\begin{equation}\label{eq:catalan}
 \sig(0)=1, \qquad \sig' = \sig\times\sig
\end{equation}
It defines the stream  
$\sig = (1,1,2,5,14,42,132,429,1430,\ldots)$
of Catalan numbers, cf.~\cite{BRW:CF-pow}.

When $A$ is a field such as the reals $\bbR$,
some streams $\sig$ have an inverse $\sig^{-1}$
with respect to convolution product,
that is, $\sig\times\sig^{-1} = \sig^{-1}\times\sig= [1]$.
By taking initial value on both sides we find that the inverse should satisfy
 $(\sig\times\sig^{-1})(0) = 1$ and hence by the definition
of convolution product, $\sig^{-1}(0) = 1/\sig(0)$ which exists only
if $\sig(0) \neq 0$ in $A$.
Similarly, taking derivatives on both sides and rearranging, we find the following SDE:
\begin{equation}\label{eq:sde-conv-inverse}
 \sig^{-1}(0)=1/\sig(0), \qquad (\sig^{-1})' = [-1/\sig(0)]\times\sig'\times\sig^{-1}
\end{equation}

\subsection{Higher-order examples}

Just as with classical differential equations, SDEs can also be higher-order.
For instance, the second-order SDE 
\begin{equation}\label{eq:fib}
\sig(0)=0,\; \sig'(0)=1, \qquad \sig'' = \sig' + \sig
\end{equation}
(with $+$ as defined above) defines the stream of Fibonacci numbers
$\sig=(0,1,1,2,3,5,8,\ldots)$.
An $n$th order SDE can always be represented as a system of $n$
first-order SDEs. For example, the Fibonacci stream is equivalently defined as the solution for $\sig$ in
\begin{equation}\label{eq:fib-flat}
\begin{array}{ll}
\sig(0)=0, \qquad & \sig' = \tau\\
\tau(0)=1, &\tau' = \tau+\sig
\end{array}
\end{equation}
A similar example is given by
\begin{equation}\label{eq:sde-nats}
\begin{array}{lcllcl}
  \sig(0) &=& 1, \quad & \sig' &=& \sig\\
\tau(0) &=& 0, \quad & \tau'   &=& \tau + \sigma
\end{array}
\end{equation}
We know already that $\ones$ is a solution for $\sig$
(cf.~equation \eqref{eq:sde-ones}).
Hence \eqref{eq:sde-nats} is equivalent to
\begin{equation}\label{eq:sde-nats-ones}
\tau(0)=0 , \qquad \tau' =\tau+\ones
\end{equation}
which has $\tau=(0,1,2,3,4,\ldots)$
as its unique solution.

A slightly more involved example is given by 
the stream $\gamma$ of Hamming numbers (or regular numbers)
which consists of natural numbers of the form
$2^i3^j5^k$ for $i,j,k \geq 0$ in increasing order
(cf.~\cite{Dijkstra:Hamming,Yuen:Hamming}).
The first part of $\gamma$ looks like
$(1, 2, 3, 4, 5, 6, 8, 9, 10, 12, 15, 16, \ldots)$.
The stream $\gamma$ can be defined by the following SDE:
\begin{equation}\label{eq:sde-hamming}
\gamma(0)=1, \qquad
\gamma' = (2\cdot\gamma) \;\|\; ((3\cdot\gamma) \;\|\; (5\cdot\gamma))
\end{equation}
where $n\cdot\gamma$, $n \in\bbN$,
is the scalar multiplication defined as in \eqref{eq:sde-scalar-mult}
and $\|$ is the \emph{merge} operator defined by
\begin{equation}\label{eq:sde-merge}
(\sig\|\tau)(0) =
\left\{
 \begin{array}{ll}
 \sig(0) & \text{if } \sig(0) < \tau(0)\\
 \tau(0) & \text{if } \sig(0) \geq \tau(0)
 \end{array} \right.
\quad
(\sig\|\tau)' =
\left\{
 \begin{array}{ll}
 \sig'\|\tau & \text{if } \sig(0) < \tau(0)\\
 \sig'\|\tau' & \text{if } \sig(0) = \tau(0)\\
 \phantom{'}\sig\|\tau' & \text{if } \sig(0) > \tau(0)
 \end{array} \right.
\end{equation}
That is, $\|$ merges two streams into one
by taking smallest initial values first,
and removing duplicates.

Many more examples of stream differential equations inspired by
analysis, arithmetic and combinatorics can be found in
\cite{Rut01:MFPS-stream-calc,Rut03:TCS-bde,Rut05:MSCS-stream-calc,Winter:PhD}.


\section{Stream Automata and Stream Calculus}
\label{sec:prelims}

We will show how one can prove the existence of unique solutions to SDEs by
using the notions of stream automata and coinduction. 

\subsection{Stream Automata and coinduction}
\label{ssec:stream-aut}

Streams can be represented by so-called stream automata.
A {\it stream automaton} (with output in $A$)
is a pair $\tup{X,s}$
where $X$ is a set (called the state space, or the carrier) and
$s = \tup{o,d} \colon X \to A \times X$ is a function that maps
each $x \in X$ to a pair consisting of an output value $o(x) \in A$ and
a unique next state $d(x) \in X$ (corresponding to the derivative).
We will write $x \sgoes{a} y$ when $o(x)=a$ and $d(x)=y$.
A small example of a stream automaton is given in
Figure~\ref{fig:stream-aut}.
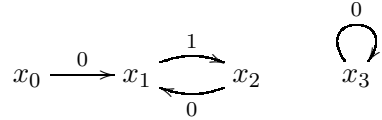
\begin{figure}[!h]
\[\xymatrix{
x_0 \ar[r]^-{0} & x_1 \ar@/^.6pc/[r]^-{1} & x_2 \ar@/^.6pc/[l]^-{0} & x_3 \ar@(ul,ur)^-{0}
}\]
\caption{Stream automaton with $X=\{x_0,x_1,x_2,x_3\}$ and $A = \{0,1\}$.}\label{fig:stream-aut}
\end{figure}
(In categorical terms, a stream automaton is a \emph{coalgebra}
for the functor $F$ on $\Set$ defined by $FX = A \times X$,
cf. Section \ref{sec:stream-gsos}.)

Intuitively, a state $x$ in a stream automaton $\tup{X,\tup{o,d}}$
represents the stream of outputs
that can be observed by following the transitions starting in $x$:
\[
(o(x), o(d(x)), o(d(d(x)),  o(d(d(d(x)))), \ldots )
\]
This stream is called the \emph{(observable) behaviour of $x$}.
For example, the behaviour of the state $x_0$ in Figure~\ref{fig:stream-aut}
is the stream $(0,1,0,1,\ldots)$. We will now characterise behaviour using the 
notion of homomorphism and finality.

A {\it homomorphism of stream automata} is a function between state spaces
that preserves outputs and transitions. Formally,
a function $f\colon X_1 \to X_2$ is a homomorphism from
$\tup{X_1,\tup{o_1,d_1}}$ to $\tup{X_2,\tup{o_2,d_2}}$
if and only if,
for all $x \in X_1$,
\[ o_1(x) = o_2(f(x)) \text{   and   }  f(d_1(x)) = d_2(f(x)),
\]
or equivalently, if and only if,
the following diagram commutes:
\[\xymatrix{
X_1 \ar[d]_-{\tup{o_1,d_1}} \ar[r]^-{f}
& X_2 \ar[d]^-{\tup{o_2,d_2}} \\
A \times X_1 \ar[r]^-{\id_A\times f}
& A \times X_2
}\]
where $\id_A$ denotes the identity map on $A$.

The set of streams $A^\omega$ is itself a
stream automaton under the map
$\zeta\colon \sigma \mapsto \tup{\sig(0),\sig'}$,
and it is moreover {\it final} which means that for any
stream automaton $\tup{o,d} \colon X \to A \times X$
there is a unique stream homomorphism $\beh{-}\colon X \to A^\omega$
(called the {\it final map}) into $\tup{\strms{A},\zeta}$:
\[\xymatrix@C+=5pc{
X \ar[d]_-{\forall\tup{o,d}} \ar@{.>}[r]^-{\exists !\beh{-}}
& A^\omega \ar[d]^-{\zeta} \\
A \times X \ar@{.>}[r]^-{\id_A\times \beh{-}}
& A \times A^\omega
}
\]
By the commutativity of the above diagram, we find that
\[
\beh{x} = (o(x), o(d(x)), o(d(d(x)),  o(d(d(d(x)))), \ldots )
\]
is  indeed the observable behaviour of $x \in X$.
The final map $\beh{-}$ is therefore often referred to as the
{\it behaviour map}.
Note that by uniqueness, the final map from $\tup{\strms{A},\zeta}$
to itself must be the identity homomorphism,
that is, for all $\sig \in \strms{A}$:
\begin{equation}\label{eq:sem-is-identity-in-final}
 \sem{\sig} = \sig.
\end{equation}

It can now easily be verified that for the stream automaton in
Figure~\ref{fig:stream-aut}, the behaviour map is:
\[ \begin{array}{rcl}
\beh{x_0} = \beh{x_2} &=& (0,1,0,1,\ldots)\\
\beh{x_1} &=& (1,0,1,0,\ldots)\\
\beh{x_3} &=& (0,0,0,0,\ldots)\\
\end{array}
\]

The universal property of the final stream automaton yields
a coinductive definition principle and a coinductive proof principle, both are often referred to as \emph{coinduction}.
In this paper we will make extensive use of both.

A map $f\colon X \to \strms{A}$
is said to be {\it defined by coinduction}, if it is obtained
as the unique homomorphism into the final stream automaton.
In practice, such an $f$ is obtained by equipping $X$
with a suitable stream automaton structure and using the finality
of $\tup{\strms{A},\zeta}$.

A \emph{proof by coinduction} is based on the notion of bisimulation.
Let
$\tup{X_1,\tup{o_1,d_1}}$ and $\tup{X_2,\tup{o_2,d_2}}$ be stream automata.
A relation $R \sse X_1\times X_2$ is a {\it stream bisimulation}
if for all $\tup{x_1,x_2} \in R$:
\begin{equation}\label{eq:def-stream-bis}
o_1(x_1) = o_2(x_2) \quad\text{ and }\quad  \tup{d_1(x_1) , d_2(x_2)} \in \, R .
\end{equation}
Two states $x_1$ and $x_2$ are said to be {\it bisimilar},
written $x_1 \sim x_2$,
if they are related by  some stream bisimulation.
We list a few well known facts about stream bisimulations,
cf.~\cite{Rut03:TCS-bde}.

\begin{lem}
\label{lem:bisim-facts}
Let
$\tup{X_1,s_1}$ and $\tup{X_2,s_2}$
be stream automata.
\begin{enumerate}
\item If $R_i \sse X_1 \times X_2$, $i \in I$, are stream bisimulations
then $\bigcup_{i\in I}R_i$ is a bisimulation.
\item The bisimilarity relation $\sim\; \sse X_1 \times X_2$ is the largest
bisimulation between $\tup{X_1,s_1}$ to $\tup{X_2,s_2}$.
\item If $f \colon X_1 \to X_2$ is a stream homomorphism, then
its graph $R=\{\tup{x,f(x)} \mid x \in X_1\}$ is a stream bisimulation.
\end{enumerate}
\end{lem}
\noindent The main result regarding bisimilarity is stated in the following theorem.
\begin{thm}\label{thm:coind}
For all stream automata $\tup{X_i,s_i}$ and all $x_i \in X_i$, $i=1,2$,
\[  x_1 \sim x_2 \quad\iff\quad \sem{x_1} = \sem{x_2}
\]
Consequently, by \eqref{eq:sem-is-identity-in-final},
for all streams $\sig,\tau \in \strms{A}$,
\[ \sig \sim \tau \quad\iff\quad \sig = \tau
\]
\end{thm}
\noindent From Theorem~\ref{thm:coind} we get the coinductive proof principle:
to prove that two streams are equal it suffices to show that they are related
by a bisimulation relation.

Finally, we also need the notions of subautomaton and minimal automaton.
A stream automaton $\tup{Y,t}$ is a \emph{subautomaton} of $\tup{X,s}$
if $Y \sse X$ and the inclusion map $\iota\colon Y \to X$ is a 
stream homomorphism, which means that $t = s\restrict_Y$. 
Given a stream automaton $\tup{X,s}$, 
the \emph{subautomaton generated by $Y_0 \sse X$} is the subautomaton
$\tup{Y,t}$ obtained by closing $Y_0$ under transitions.
A stream automaton $\tup{X,s}$ is \emph{minimal} if the behaviour map 
$\sem{-}\colon X \to \strms{A}$ is injective. Note that due to 
Theorem~\ref{thm:coind}, every subautomaton of the final stream automaton is
minimal.

Referring to the automaton in  Figure~\ref{fig:stream-aut}, an example of a stream
bisimulation is given by $\{(x_0,x_2), (x_1,x_1), (x_2,x_2) \}$.
It follows from Theorem~\ref{thm:coind} that $\beh{x_0} = \beh{x_2}$.

\subsection{Stream Calculus}
\label{ssec:stream-calc}

In this short section, we introduce some further preliminaries 
on stream calculus that we will be using in the remainder of the paper.
As we have seen in Section~ \ref{ssec:stream-ops},
any operation on $A$ can be lifted element-wise
to an operation on $\strms{A}$.
In fact, any algebraic structure on $A$
lifts element-wise to $\strms{A}$.
(We show this in a more abstract setting in Section~\ref{ssec:linear-revisited}.)
But we are not only interested in element-wise operations.
We will use that if $A$ is a commutative ring, then also
\[
(A^\omega,+,-,\times,\cns{0},\cns{1})
\]
is a commutative ring, cf.~\cite[Thm.4.1]{Rut03:TCS-bde}.
Similarly, if $A$ is a field, then $\strms{A}$ is a vector space over $A$
with the operations of scalar multiplication and addition.

Table~\ref{tbl:stream-calc} 
summarises the SDEs defining the stream calculus operations
on $\strms{A}$ most of which were already introduced in 
Section~\ref{ssec:stream-ops}.
The fact that these SDEs have unique solutions will follow
from the results in Section~\ref{sec:syntactic-method}.

\begin{table}[ht!]
\begin{center}
\begin{tabular}{|l|l|l|}
\hline
derivative: & initial value: & name:
\\
\hline
$[a]' = [0]$
&
$[a](0) = a$
&
constant $[a], a \in A$
\\
$(\sigma + \tau )' = \sigma' +  \tau'$
&
$(\sigma + \tau )(0) = \sigma(0) + \tau(0)$
&
sum
\\
$(a\cdot\sigma )' = a\cdot (\sigma' )$
&
$(a\cdot\sigma )(0) = a\cdot \sigma(0)$
&
scalar multiplication
\\
$(-\sigma )' = - (\sigma' )$
&
$(-\sigma )(0) = -  \sigma(0)$
&
minus
\\
$(\sigma \times \tau )' = (\sigma' \times  \tau)  +
([\sigma(0)]  \times \tau')$
&
$(\sigma \times \tau )(0) = \sigma(0) \cdot \tau(0)$
&
convolution product
\\
$( \sigma^{-1} )' = - [\sigma(0)^{-1}] \times  \sigma' \times  \sigma^{-1} $
&
$(\sigma^{-1})(0) =  \sigma(0)^{-1} $
&
convolution inverse
\\
\hline
\end{tabular}
\end{center}
\caption{Operations of stream calculus. Recall (from Section~\ref{ssec:stream-ops}) 
that convolution inverse is defined only
if $\sig(0)$ is an invertible element of $A$. In the case $A$ is a field
this is equivalent with $\sig(0) \neq 0$, and
as usual, we will often write 
$\frac{\sigma}{\tau}$ for $\sigma \times \tau^{-1}$.}
\label{tbl:stream-calc}
\end{table}

\noindent We further add to our stream calculus the constant stream
\[\X = (0,1,0,0,0, \ldots) \quad\text{ defined by }\quad
\X(0) = 0, \qquad \X' = [1].
\]
Multiplication by $\X$ acts as ``stream integration'' 
(seen as an inverse to stream derivative) since 
\[
(\X \times \sig)' = \sig
\]
This follows from the fact that,
for all $\sig \in \strms{A}$,
\[ \X \times \sig \;\;=\;\; \sig\times\X \;\;=\;\; (0,\sig(0),\sig(1),\sig(2),\ldots)
\]
This leads to the very useful \textit{fundamental theorem of stream calculus}
\cite[Thm.~5.1]{Rut03:TCS-bde}.
\begin{thm}
\label{fundamental theorem}
For every $\sigma \in A^\omega$,
$
\sigma \;\;=\;\; [\sigma(0)] + (\X \times \sigma ').
$
\end{thm}

\proof
For all $\sigma$, we have:\qquad
$\begin{array}[t]{rcl}
\sigma & = &
(\sig(0),\sig(1),\sig(2),\ldots)
\\
& = & (\sigma(0) , 0,0,0, \ldots ) + (0, \sig(1),\sig(2),\sig(3),\ldots)
\\
& = & 
      [\sig(0)] + (\X \times \sigma') \vspace{-12 pt}
\end{array}$

\qed

\noindent We conclude this section by an enhancement of the bisimulation
proof method. The general result behind the soundness of this method is
described in Section~\ref{ssec:bialgebras}.

\begin{defi}[bisimulation-up-to]\label{def:up-to}
Let $\Sig$ denote a collection of stream operations.
A relation $R \subseteq A^\omega \times A^\omega$
is a \emph{(stream) bisimulation-up-to-$\Sig$} if
for all $(\sigma, \tau) \in R$:
\[
\sigma(0)  = \tau(0) \qquad \text{and}\qquad
(\sigma',\tau') \in \bar{R},\]
where
$\bar{R} \subseteq A^\omega \times A^\omega$
is the smallest relation such that
\begin{enumerate}
\item $R \subseteq \bar{R}$
\item $\{ \, \tup{\sigma,\sigma} \mid \sigma \in A^\omega \, \}  \subseteq \bar{R}$
\item $\bar{R}$ is closed under the (element-wise application of)
operations in $\Sig$.
(For instance, if $\Sig$ contains addition and
$\tup{\alpha,\beta} , \tup{\gamma,\delta} \in \bar{R}$ then
$\tup{\alpha + \gamma , \, \beta + \delta} \in \bar{R}$.)
\end{enumerate}
We write $\sigma \sim_\Sig \tau$ if there exists a
bisimulation-up-to-$\Sig$ containing $\tup{\sigma , \tau}$.
\end{defi}

\begin{thm}[coinduction-up-to]\label{thm:up-to}
Let $\Sig$ be a subset of the stream calculus operations
from Table~\ref{tbl:stream-calc}.
We have:
\begin{equation}
\label{coinduction-up-to}
\sigma \sim_\Sig \tau \quad \Rightarrow \quad \sigma = \tau
\end{equation}
\end{thm}

\proof
If $R$ is a bisimulation-up-to-$\Sig$, then $\bar{R}$ can be shown to be
a bisimulation relation by structural induction on its definition.
The theorem then follows by Theorem~\ref{thm:coind}.
\qed

\section{Simple Specifications}
\label{sec:simple-specs}

In Sections~\ref{sec:simple-specs}-\ref{sec:context-free-specs}, we will
characterise the classes of streams, i.e.\ subsets of $A^\omega$,
that arise as the solutions to finite systems of SDEs over varying algebraic
structures/signatures.

We start by defining the most simple type of systems of SDEs.
Let $A$ be an arbitrary set.
A {\it simple equation system}
over a set (of variables) $X = \{x_i \mid i \in I\}$
is a collection of SDEs, one for each $x_i \in X$,
of the form
\[ \begin{array}{lclclcl}
x_i(0) &=& a_i, & \qquad x_i' &=& y_i;\\
\end{array}\]
where $a_i \in A$ and $y_i \in X$ for all $i \in I$.
We call a simple equation system over $X$ finite,
if $X$ is finite.
The SDEs in equations \eqref{eq:sde-ones} and \eqref{eq:sde-alt}
are examples of finite simple equation systems.
Note that any stream $\sigma$ is the
solution of the \emph{infinite} simple equation system
over $X = \{ x_n \mid n \in \bbN\}$ defined by:
$x_n(0) = \sig(n)$ and $x_n' = x_{n+1}$, for all $n\in\bbN$.

A simple equation system corresponds to a map
$e \colon X \to A \times X$,
i.e., to a stream automaton with state space $X$.
A {\it solution of $\es\colon X \to A \times X$} is an assignment
$h\colon X \to A^\omega$ of variables to streams that preserves the
equations: $h(x_i) = a_i$ and $h(x_i)' = h(y_i)$ for all $i \in I$.
This holds exactly when the following diagram commutes:
\[\xymatrix{
X \ar[d]_-{\es} \ar[r]^-{h} & A^\omega \ar[d]^-{\zeta}\\
A \times X \ar[r]^-{\id_A\times h} & A \times A^\omega
}\]
In other words, solutions are stream homorphisms
from $\tup{X,\es}$ to the final stream automaton.
By coinduction, solutions to simple equation systems exist and are unique.
We will also say that a stream $\sigma \in A^\omega$
is a solution of $\es$ if $\sigma = h(x)$ for some $x \in X$,
in which case we call $\es$ a {\it specification of $\sigma$}.

The solutions of finite simple equation systems
are exactly the behaviours of finite stream automata,
which are precisely the {\it eventually periodic streams}.
This is easy to prove. We state the result
explicitly to make clear the analogue with the results on
linear and context-free specifications that will be discussed
in Sections 
\ref{sec:linear-specs}  and \ref{sec:context-free-specs}.

\begin{prop}\label{prop:class-simple}
The following are equivalent for all streams $\sigma\in A^\omega$.
\begin{enumerate}
\item $\sigma$ is the solution of a finite simple equation system.
\item $\sigma$ generates a finite subautomaton of the final stream automaton.
\item $\sigma$ is eventually periodic, i.e., $\sig^{(k)} = \sig^{(n)}$
for some $k , n \in \bbN$ with $k<n$.
\end{enumerate}
\end{prop}

\proof
$1 \Ra 2$: Let $h\colon X \to A^\omega$ be a solution of the finite
$\es \colon X \to A \times X$ and $\sigma = h(x)$ for some $x \in X$.
The subautomaton generated by $\sigma$
is contained in the image $h(X)$ which is finite, since $X$ is finite.

$2 \Ra 3$: The subautomaton generated by $\sigma$ has as its state set
$\{\sig^{(k)} \mid k \in \bbN \}$ which is finite by assumption.
Consequently, there are $k,n\in\bbN$ such that $\sig^{(k)} = \sig^{(n)}$
and $k < n$.

$3 \Ra 1$: Assume that $\sig^{(k)} = \sig^{(n)}$ for $k < n \in \bbN$.
Let $X = \{x_0,\ldots,x_{n-1}\}$ and define $\es\colon X \to A \times X$,
for all $i = 0, \ldots , n-2$, by
\[
\es(x_i) = \tup{\sigma(i), x_{i+1}}
\]
and by $\es(x_{n-1}) = \tup{\sigma(n-1), x_{k}}$.
Now $\sigma = h(x_0)$ where $h$ is the unique solution of $\es$.
\qed

Eventually periodic streams constitute some of the simplest
infinite objects that have a finite representation.
Such finite representations make it possible to compute with and reason
about infinite objects.
We provide a couple of examples.

\begin{exa}[Rational numbers in binary]
Let $2 = \{0,1\}$ denote the set of bits.
Rational numbers with odd denominator,
that is, elements of $\bbQodd = \{q = \frac{n}{2m+1} \mid n,m \in \bbZ\}$,
can be represented as eventually periodic bitstreams.
The representation $B\colon \bbQodd \to 2^\omega$ is obtained
by coinduction via the following stream automaton structure on
$\bbQodd$:
\[
o(q) = n \text{ mod }2, \qquad
d(q) = (q - o(q))/2
\]
For example, the finitary representation of the number $\frac{17}{5}$
can be found by computing output and derivatives leading to
the following stream automaton:
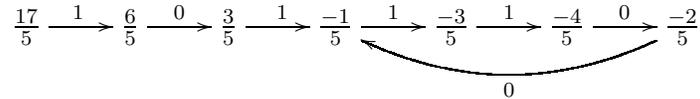
\[\xymatrix{
\frac{17}{5} \ar[r]^{1} &
\frac{6}{5} \ar[r]^{0} &
\frac{3}{5} \ar[r]^{1} &
\frac{-1}{5} \ar[r]^{1} &
\frac{-3}{5} \ar[r]^{1} &
\frac{-4}{5} \ar[r]^{0} &
\frac{-2}{5} \ar@/^1.5pc/[lll]^{0}
}\]
Hence, $B({\frac{17}{5}}) = 101(1100)^\omega$.
Such base 2 expansions allow for efficient implementations of arithmetic
operations, cf.~\cite{Hehner:siam}.
\end{exa}

\begin{exa}[Regular languages over one-letter alphabet]
A bitstream $\sig \in \strms{2}$ corresponds to a language
$L \sse \Pow(A^*)$ over a one-letter alphabet $A = \{a\}$ via:
\[ a^n \in L \quad\iff\quad \sig(n) = 1, \qquad \text{ for all } n \in \bbN.
\]
For example, the language $L = \{ a^n \mid n = 1 + 3k, k \in \bbN\}$ is
represented by the state 0 in the stream automaton:
\[\xymatrix{
0 \ar[r]^-{0} & 1 \ar[r]^-{1} & 2 \ar@/^1.2pc/[ll]^-{0} 
}
\]
\end{exa}


\section{Linear Specifications}
\label{sec:linear-specs}

Equations \eqref{eq:fib-flat} and \eqref{eq:sde-nats}
in Section~\ref{sec:sde-intro} are 
examples of \emph{linear equation systems}:
the righthand side of each SDE is a linear combination of
the variables on the left.
We will now
study this type of systems in more detail.

Throughout this section we assume $A$ is a field.
The set $A^\omega$ then becomes
a vector space over $A$ by defining scalar multiplication and
vector addition pointwise, as in Table~\ref{tbl:stream-calc}.
We denote by $\V(X)$
the set of all formal linear combinations over $X$,
i.e.,
\[\V(X) =
\{a_1x_1 + \ldots + a_nx_n \mid a_i \in A, x_i\in X, \;\forall i: 1\leq i\leq n\}\]
or equivalently,
$\V(X)$ is the set of all functions from $X$ to $A$ with
finite support.
In fact, $\V(X)$ is itself a vector space over $A$
by element-wise scalar multiplication and sum,
and it is freely generated by $X$. That means $X$ is a
basis for $\V(X)$, and hence every linear map from $\V(X)$ to a vector space $W$
is determined by its action on $X$. 
More precisely,
for every function $f\colon X \to W$ there is a unique linear map
$f^\sharp\colon \V(X) \to W$ extending $f$, which is
defined by: 
\[
 f^\sharp(a_1x_1 + \ldots + a_nx_n) = a_1f(x_1) + \ldots + a_nf(x_n).
\]
We note that the linear extension 
$\id^\sharp_{\strms{A}}\colon \V(A^\omega) \to A^\omega$
of the identity map $\id : A^\omega \to A^\omega$ 
gives the evaluation of formal linear combinations in the vector space $A^\omega$.

\subsection{Linear equation systems}
\label{ssec:linear-equation-systems}

A {\it linear equation system}
over a set $X = \{x_i \mid i \in I\}$
is a collection of SDEs, one for each $x_i \in X$,
of the form
\[ \begin{array}{lclclcl}
x_i(0) &=& a_i, & \qquad x_i' &=& y_i;\\
\end{array}\]
where $a_i \in A$ and $y_i \in \V(X)$ for all $i \in I$.
In other words,
a linear equation system is a map 
\[
\es=\tup{o,d} \colon X \to A \times \V(X).
\]
Again, we say that $e$ is finite, if $X$ is finite.
A {\it solution of $\es$} is an assignment $h\colon X \to A^\omega$
that preserves the equations, that is, for all $x_i \in X = \{x_1, \ldots, x_n\}$,
if $d(x_i) = a_1x_1 + \ldots + a_nx_n$, then
\[
   h(x_i)(0) = o(x_i)  \quad\text{ and }\quad
   h(x_i)' = a_1 \cdot h(x_1) + \cdots + a_n\cdot h(x_n)
\]

In the remainder of this
section, we give two ways of solving finite linear equation systems
and characterise their solutions.
The first uses coinduction for automata over vector spaces
--- here we will 
see that any linear equation system has a unique solution,
and solutions to finite linear equation systems are exactly the streams
that generate a finite-dimensional subspace.
The second uses stream calculus and yields a matrix solution method
which in turn shows that solutions to finite linear equation systems are
exactly the rational streams.

\subsection{Linear stream automata}
\label{ssec:linear-aut}

In this subsection, we first show how to solve 
linear equation systems by viewing them as
stream automata over vector spaces.
A linear equation system $\tup{o,d} \colon X \to A \times \V(X)$
can be seen as a (specification of a)
{\it weighted stream automaton},
cf.~\cite{BBBRS:lwa,Rut05:MSCS-stream-calc,Rut03:TCS-bde}.
In this view, the first component $o$
assigns output weights to states, and
the second component $d$ defines an $A$-weighted transition structure in
which state $x$ goes to state $y$ with weight $a \in A$ iff $d(x)(y)=a$.
(Recall that $d(x)$
is a function from $X$ to $A$ with finite support.)

To illustrate the construction, consider the two-dimensional
linear equation system from \eqref{eq:sde-nats}, repeated here:
\begin{equation}\label{eq:sde-nats-again}
\begin{array}{lcllcl}
x_1(0) &=& 1, \quad & x_1' &=& x_1\\
x_2(0) &=& 0, \quad & x_2'   &=& x_1+x_2
\end{array}
\end{equation}
It corresponds to the following weighted stream automaton 
(where a state is underlined if the output is 1, otherwise the output is 0):
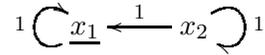
\[\xymatrix@R=1pc{
&&
\\
\underline{x_1} \ar@(dl,ul)^-{1}
&
x_2 \ar[l]_-{1} \ar@(ur,dr)^-{1}
}\]
Let us try to construct a stream automaton for the solution
of $x_2$ by inductively applying \eqref{eq:sde-nats-again} and the
definition of $+$ (cf.~\eqref{eq:sde-plus}):
\[\begin{array}{lcll}
x_2 &\sgoes{0}& x_1+x_2 \\
     &\sgoes{1}& x_1+ (x_1 + x_2) & = 2\cdot x_1 + x_2 \\
     &\sgoes{2}& x_1+(x_1 + (x_1 +x_2)) & = 3\cdot x_1 +x_2\\
     &\sgoes{3}& \ldots
\end{array}\]
We notice two things: First,
the stream behaviour of $x_2$ indeed consists of the sequence
of natural numbers $\nats = (0,1,2,3,4,5,\ldots)$.
Second, the states of this stream automaton are not
stream variables, but linear combinations of the stream variables
$x_1$ and $x_2$.

\begin{rema}
The above example also shows that for streams over the field
$A=\bb{R}$, if the coefficients of the linear system are integers,
then the solutions will be streams of integers, since all initial
values will be computed using only multiplication and addition of integers.
\end{rema}

The above example motivates the following definition.

\begin{defi}
A {\it linear stream automaton} is a stream automaton over vector spaces,
i.e., it is a pair of maps $\tup{o,d} \colon V \to A \times V$
where $V$ is a vector space over $A$,
and $o \colon V \to A$ and $d \colon V \to V$
are linear maps. Note that the pairing $\tup{o,d}$ is again linear.
A {\it homomorphism of linear stream automata} is a map between the
state vector spaces which is both linear and a homomorphism
of stream automata.
\end{defi}

Solutions to a linear equation system
will now be obtained by coinduction,
for linear stream automata, using the following lemma.

\begin{lem}\label{lem:final-lin}
We have:
\begin{enumerate}
\item A linear equation system $\es \colon X \to A \times \V(X)$
  corresponds to a 
  linear stream automaton $\es^\sharp \colon \V(X) \to A \times \V(X)$
\item The final stream automaton is also a final linear stream automaton.
\end{enumerate}
\end{lem}

\proof
(1): 
Since $A$ is a vector space over itself, $A \times \V(X)$ is a 
(product) vector space, and we obtain 
$\es^\sharp \colon \V(X) \to A \times \V(X)$ as the linear extension of $\es$.
Note that
$\es^\sharp = \tup{o^\sharp,d^\sharp}$. 

(2): 
The initial value and derivative maps are linear:
\[ \begin{array}{rcl}
(a \cdot \sigma + b\cdot \tau)(0) &=& a\cdot\sig(0) + b \cdot \tau(0)\\
(a \cdot \sigma + b\cdot \tau)' &=& a\cdot\sig' + b \cdot \tau'\\
\end{array}\]
Hence $\tup{A^\omega,\zeta}$ is a linear stream automaton.
Moreover, for any linear stream automaton
$\tup{o,d} \colon V \to A \times V$, the final map $\beh{-}$
of the underlying (set-based) stream automata is linear, since for
all $v, w \in V$, $a,b \in A$ and $n \in \bbN$,
\[ \begin{array}{rcl}
\beh{a\cdot v + b\cdot w}(n) &=&
  o(d^n(a\cdot v + b\cdot w)) \\
&=&  a\cdot o(d^n(v)) + b\cdot o(d^n(w)) \quad (\text{\small by linearity of $o$ and $d$})\\
&=&  a\cdot \beh{v}(n) + b\cdot \beh{w}(n)
\end{array}
\]
Hence $\beh{-}$ is also the unique homomorphism of linear stream automata
into $\tup{A^\omega,\zeta}$.
\qed

\begin{prop}\label{prop:linear-solutions}
Every linear equation system has a unique solution.
\end{prop}
\proof
Applying Lemma~\ref{lem:final-lin} and the coinduction principle
for linear stream automata,
we obtain for each linear equation system
$\tup{o,d}\colon X \to A \times \V(X)$ a unique linear stream homomorphism
$g\colon \V(X) \to \strms{A}$, as shown in the following picture
where $\eta_X\colon X \to \V(X)$ denotes the inclusion of the basis vectors
into $\V(X)$:
\begin{equation}\label{eq:linear-solutions}
\xymatrix@R=+10mm{
X \ar[d]_-{\tup{o,d}} \ar[r]^-{\eta_X}  
  & \V(X) \ar[dl]^<<{\tup{o^\sharp,d^\sharp}} \ar[r]^-{g} 
  & A^\omega \ar[d]^-{\zeta}\\
A \times \V(X) \ar[rr]^-{\id_A\times g}
  &
  & A \times A^\omega
}\end{equation}
The composition $g \circ\eta_X = g\restrict_X \colon X \to \strms{A}$ is a solution
of $\tup{o,d}$ by the linearity of $g$. To see this, suppose that
$d(x_i) = a_1x_1 + \ldots + a_nx_n$
for $x_i \in X$.
We then have:
\[ \begin{array}{rcl}
 g\restrict_X\!(x_i)' &=& g(d(x_i)) = g(a_1x_1 + \ldots + a_nx_n)\\
 &=& a_1g\restrict_X\!\!(x_1) + \ldots + a_ng\restrict_X\!\!(x_n)
\end{array}\]

We note that for finite $X$, the linear homomorphism $g\colon \V(X) \to \strms{A}$ can be represented
by a finite dimensional matrix with rational streams as entries, similar to the one in
\eqref{eq:lin-solution} of the next subsection; see \cite{Rut:rational} or \cite{BBBRS:lwa} for details.
\qed

We can now state the first characterisation of the solutions to finite linear equation systems.

\begin{prop}\label{prop:class-linear}
The following are equivalent for all streams $\sigma\in A^\omega$:
\begin{enumerate}
\item $\sigma$ is the solution of a finite linear equation system.
\item $\sigma$ generates a finite-dimensional subautomaton of the
 final linear stream automaton.
\end{enumerate}
\end{prop}

\proof
For the direction $1 \Ra 2$: Let $\sig$ be a solution to a
finite $\es\colon X \to A \times \V(X)$.
Let $\tup{Z_\sigma,\zeta_\sigma}$ be the linear subautomaton
generated by $\sigma$ in the final linear automaton,
i.e., the state space $Z_\sig$ is the subspace
generated by the derivatives of $\sig$.
Since $\V(X)$ is finite-dimensional, so is its final image $g(\V(X))$,
and since $Z_\sigma$ is a subspace of $g(\V(X))$,
also  $Z_\sigma$ is finite-dimensional.

The direction $2 \Ra 1$ follows by constructing a linear equation system
using a similar argument as the one given in
the proof of Proposition~\ref{prop:rational-linear-char} below.
A detailed proof can be found in \cite[section 5, Thm.5.4]{Rut:rational}.
\qed\enlargethispage{\baselineskip}

\subsection{Matrix solution method}
\label{ssec:matrix-solution-method}

In this section we will provide an algebraic characterisation of 
solutions of finite linear equation systems.
We will show that solutions of such 
systems are rational streams, and give a matrix-based method
for computing these solutions.
Recall (from \cite{Rut03:TCS-bde}) that a stream $\sigma \in A^\omega$
is \emph{rational} if it is of the form
\[
\sigma = \, \frac{a_0 + (a_1 \times \X) + (a_2 \times \X^2) + \cdots + (a_n \times \X^n)}
{b_0 + (b_1 \times \X) + (b_2 \times \X^2) + \cdots + (b_m \times \X^m)}
\]
for $n, m \in \bbN$ and $a_i, b_j \in A$ 
and with $b_0 \neq 0$. (The operations of sum, product and inverse
were all defined in Section~\ref{ssec:stream-calc}.)

First, we will identify the relevant algebraic structure in which we
can do matrix manipulations.
As mentioned in Section~\ref{ssec:stream-calc}, 
when $A$ is a commutative ring (so, in particular, when $A$ is a field),
the stream calculus operations turn $\strms{A}$ into a
commutative ring. 
For any ring $R$, the set $\Mat_n(R)$ of
$n$-by-$n$ matrices over $R$ is again a ring under matrix addition and
matrix multiplication. When $R$ is commutative then $\Mat_n(R)$
is an \emph{associative $R$-algebra}, which means that it
also has a scalar multiplication (with elements from $R$) which is
compatible with the ring structure, that is, for all $r \in R$
and $M, N\in \Mat_n(R)$, $r\cdot(MN) = (r\cdot M)N = M (r\cdot N)$.
This scalar multiplication
$r \cdot M$ is defined by multiplying each entry of $M$ by $r$, that is,
$(r\cdot M)_{i,j} = r\times M_{i,j}$.
We refer to \cite{Lang} for further results on matrix rings.

For a linear equation system with $n$ variables,
we will consider the associative $\strms{A}$-algebra $\Mat_n(\strms{A})$,
and we will denote both matrix multiplication and scalar
multiplication by $\cdot$. The context should make clear which
operation is intended. The $\cdot$ notation is used to distinguish
the operations from the multiplication in the underlying ring of
stream calculus.
In order to keep notation simple,
we describe the matrix solution method for two variables, but it is
straightforward to generalise it to $n$ variables.

A linear equation system with two variables
\begin{equation}\label{eq:lin-example-2dim}
\begin{array}{ccc}
x_1'= m_{11}x_1 + m_{12}x_2 & \qquad & x_1(0) = n_1\\[2mm]
x_2'= m_{21}x_1 + m_{22}x_2 & \qquad & x_2(0) = n_2
\end{array}
\end{equation}
can be written in matrix form as
\[
\pmat{x_1\\x_2}' = M \cdot \pmat{x_1\\x_2} \qquad \pmat{x_1\\x_2}(0) = N
\]
where derivative and initial value are taken element-wise, and where
$M$ and $N$ are matrices over $\strms{A}$ given by
\[
M = \pmat{\cns{m_{11}} & \cns{m_{12}}\\[1mm] \cns{m_{21}} & \cns{m_{22}}}
\qquad
N = \pmat{\cns{n_1}\\[1mm] \cns{n_2}}
\]
By applying the fundamental theorem of stream calculus to both
stream variables, we find that
\[\begin{array}{rcl}
\pmat{x_1\\x_2} &=& \pmat{x_1\\x_2}(0) + \X\cdot\pmat{x_1\\x_2}'
\\[.8em]
 &=& N + \X\cdot M\cdot\pmat{x_1\\x_2}
\end{array}\]
(Note that $\X = (0,1,0,0,0, \ldots)$ 
is a scalar stream.)
This is in $\Mat_n(\strms{A})$
equivalent to
\[
(I - (\X \cdot M))\cdot \pmat{x_1\\x_2} = N
\]
where $I$ is the identity matrix.
The solution to \eqref{eq:lin-example-2dim} can now be obtained
as:
\begin{equation}\label{eq:lin-solution}
\pmat{x_1\\x_2} = (I - (\X \cdot M))^{-1}\cdot N
\end{equation}
We should, of course, first convince ourselves that
the inverse of the matrix $I - (\X \cdot M)$ always exists.
In general, an element of a matrix ring $\Mat_n(R)$ (over a commutative
ring $R$) is invertible if its determinant has a multiplicative inverse
in $R$. Hence $I - (\X \cdot M)$ has an inverse in $\Mat_2(\strms{A})$
if its determinant is a stream whose initial value is non-zero.
The matrix $I - (\X \cdot M)$ looks as follows
\[
I - (\X \cdot M) =
\pmat{\cns{1} -(\X\times \cns{m_{11}}) & \cns{0}-(\X \times \cns{m_{12}})
\\[1mm]
\cns{0}-(\X\times \cns{m_{21}}) & \cns{1}-(\X \times \cns{m_{22}})
}
\]
From the definitions of sum and convolution product it follows that
the initial value of the determinant equals the determinant of the
matrix of initial values:
\[
\det(I - (\X \cdot M))(0)
\quad = \quad
\det\pmat{
1 -(0\cdot m_{11}) & 0-(0\cdot m_{12})
\\[1mm]
0-(0\cdot m_{21}) & 1 - (0\cdot m_{22})}
\quad = \quad
\det\pmat{1 & 0\\0&1}
\quad = \quad
1.
\]
Hence the determinant of $I - (\X \cdot M)$ will always have initial
value equal to 1, and consequently $(I - (\X \cdot M))^{-1}$ exists
and can be computed using the standard linear algebra technique
by performing elementary row operations on the identity matrix.
These row operations consist of multiplying or
dividing by a rational stream, and adding rows, hence
if an invertible matrix has rational streams as entries,
then so does its inverse. 
(Alternatively, this also follows from
Cramer's rule.)
It is easy to see that this argument carries over to
higher dimensions.
We have proved one direction of the second characterisation result.

\begin{prop}\label{prop:rational-linear-char}
The following are equivalent for all streams $\sigma\in A^\omega$:
\begin{enumerate}
\item $\sigma$ is the solution of a finite linear equation system.
\item $\sigma$ is rational.
\end{enumerate}
\end{prop}

\proof
If $\sigma$ is a solution to a finite linear equation system,
then by the argument above this proposition, we find that $\sigma$ is
a linear combination of rational streams, hence itself rational.
For the converse direction,
if $\sig \in \strms{A}$ is rational,
there exists a $d \in \bbN$ such that the $d$-th derivative
$\sig^{(d)}$ is a linear combination of $\sig^{(0)},\ldots,\sig^{(d-1)}$.
(The value $d$ is bounded in terms of the degree of $\rho$ and $\tau$
where $\sig = \rho/\tau$.)
Hence $\sig^{(d)} = \sum_{i=0}^{d-1} a_i\cdot\sig^{(i)}$
for some $a_i \in A$, $i < d$. It follows that $\sig$ is the solution for $x_0$
in the following $d$-dimensional linear equation system:
\[ \begin{array}{lclrcl}
x_0' &=& x_1 & x_0(0) &=& \sig(0)\\
x_1' &=& x_2 & x_1(0) &=& \sig(1)\\
\;\;\vdots &&\;\vdots & \;\vdots\;\;\;\; && \;\vdots\\
x_{d-2}' &=& x_{d-1} & x_{d-2}(0) &=& \sig(d-2)\\
x_{d-1}' &=&  a_0x_0 + \cdots + a_{d-1}x_{d-1} \qquad & x_{d-1}(0)&=& \sig(d-1)\\
\end{array}\]
See also \cite[Thm.5.3, Thm.5.4]{Rut:rational} for a more general proof
using the vector space structure of $\strms{A}$.
\qed

We illustrate the matrix solution method with an example.

\begin{exa}
The Fibonacci example from \eqref{eq:fib}
\[
\sig(0)=0, \;\; \sig'(0)=1, \qquad \sig'' = \sig' +\sig
\]
corresponds to the linear equation system (with $x_1 = \sig, x_2=\sig'$)
\[
\pmat{x_1\\x_2}' = \pmat{0 & 1\\1 & 1}\cdot\pmat{x_1\\x_2} \qquad
\pmat{x_1\\x_2}(0) = \pmat{0\\1}
\]
whose solution is given by instantiating \eqref{eq:lin-solution}:
\[ \begin{array}{rcl}
\pmat{x_1\\x_2} &=& \pmat{1 & -\X \\ -\X & 1 - \X}^{-1}\cdot\pmat{0\\1}
\\[1.5em]
&=& \pmat{\frac{1-\X}{1-\X-\X^2} & \frac{\X}{1-\X-\X^2} \\[.3em]
      \frac{\X}{1-\X-\X^2} & \frac{1}{1-\X-\X^2}}\cdot\pmat{0\\1}
\\[1.5em]
&=& \pmat{\frac{\X}{1-\X-\X^2}\\[.3em]\frac{1}{1-\X-\X^2}}
\end{array}\]
Hence the solution for $\sigma \;(=x_1)$ is the rational stream
\begin{equation}
\label{rational expression for Fibonacci}
\sigma = \frac{\X}{1-\X-\X^2}
\end{equation}
By computing successive initial value and derivatives
using the rational expression for $\sig$,
we find again the Fibonacci sequence:
\[
\sigma = (0,1,1,2,3,5,8,13,\ldots)
\]
\end{exa}

Here are some further examples of linear equation systems
that define some more and some less familiar rational streams.

\begin{exa}[Naturals]
Take $A = \bbR$.
The solution for $\sigma$ in
the following linear equation system
is the stream of natural numbers
$\sigma = \nats= (1,2,3,4, \ldots)$:
\[\begin{array}{rcrcrcl}
\sig(0) &=& 1,
& \quad&
\sigma' &=& \sig + \tau
\\
\tau(0) &=& 1, & \quad & \tau' &=&  \tau
\end{array}\]
Applying the matrix solution method, we find the rational expression
\[ \sigma = \frac{1}{(1 - \X)^2}
\]
\end{exa}

\begin{exa}[Powers]
Take $A = \bbR$.
For any $a \in \bbR$, the linear equation
\[\sig(0)=1, \quad \sigma' = a\cdot\sig
\]
has as its solution $\sigma = (1,a,a^2,a^3,a^4, \ldots)$
with rational expression
\[ \sigma = \frac{1}{1-(a\times\X)}
\]
\end{exa}

\begin{exa}[Alternating]
The second-order stream differential equation
\[ \sig(0)=0,\; \sig'(0)=1, \quad\sigma'' = -\sig
\]
can be written as a linear equation system
\[\begin{array}{rcrcrcl}
\sig(0) &=& 0, & \quad&  \sigma' &=& \tau
\\
\tau(0) &=& 1, & \quad & \tau' &=&  -\sig
\end{array}\]
The solution for $\sigma$ is $\sigma = (0,1,0,-1,0,1,0,-1,\ldots)$
with rational expression
\[\sig = \frac{\X}{1+\X^2}
\]
Note that $\sigma$ is actually eventually periodic,
and could also be defined by a simple equation system with four variables.
\end{exa}

\begin{exa}[$n$th powers]
For $n \in \bbN$,
consider the stream $\nats^{\tup{n}}=(1,2^n,3^n,4^n, \ldots)$
of $n$-th powers of the naturals.
Inspecting the derivatives, we find that
\[\begin{array}{lcl}
(\nats^{\tup{n}})' &=& (\; (1+1)^n, (1+2)^n, (1+3)^n, \ldots\;)
\\[.8em]
&=& \left(\; \sum_{k=0}^n \binom{n}{k}1^k,\; \sum_{k=0}^n \binom{n}{k}2^k, \;\sum_{k=0}^n \binom{n}{k}3^k, \ldots \;\right)
\\[.8em]
&=& \sum_{k=0}^n \binom{n}{k} \nats^{\tup{k}}
\end{array}\]
This shows that $\nats^{\tup{0}},\ldots,\nats^{\tup{n}}$ can be defined by a linear equation system in $n+1$ variables.
A rational expression for $\nats^{\tup{n}}$ can be computed using the
fundamental theorem (Theorem~\ref{fundamental theorem}).
We show here the expressions for $n \leq 3$:
\[\begin{array}{rclclcl}
\nats^{\tup{0}} &=& 1 + X \times \nats^{\tup{0}} & \Ra & \multicolumn{3}{l}{\nats^{\tup{0}} = \dsty\frac{1}{1-\X} \;\;\;\;= \ones}
\\[1em]
\nats^{\tup{1}} &=& 1 + X \times (\nats^{\tup{0}} + \nats^{\tup{1}})
\\
  &=&  1 + X \times (\frac{1}{1-\X}+\nats^{\tup{1}})  & \Ra &
 \multicolumn{3}{l}{\nats^{\tup{1}} = \dsty\frac{1}{(1-\X)^2} = \nats}
 \\[1em]
\nats^{\tup{2}} &=& \multicolumn{3}{l}{1 + X \times (\nats^{\tup{0}} + 2\nats^{\tup{1}} + \nats^{\tup{2}})}
\\
 &=&  \multicolumn{3}{l}{1 + X \times (\frac{1}{1-\X} + \frac{2}{(1-\X)^2} + \nats^{\tup{2}})}
  & \Ra  
  & \nats^{\tup{2}} = \dsty\frac{1+\X}{(1-\X)^3}
  \\[1em]
\nats^{\tup{3}} &=& \multicolumn{3}{l}{1 + X \times (\nats^{\tup{0}} + 3\nats^{\tup{1}} + 3\nats^{\tup{2}} + \nats^{\tup{3}})}
\\
 &=&  \multicolumn{3}{l}{1 + X \times (\frac{1}{1-\X} + \frac{3}{(1-\X)^2} + \frac{3(1+\X)}{(1-\X)^3} + \nats^{\tup{3}})}
  & \Ra 
  & \nats^{\tup{3}} = \dsty\frac{1+4\X +\X^2}{(1-\X)^4}
\end{array}\]
A recurrence relation for these rational expressions is given in
section 6.2 of \cite{NR:moessner}. In section 6.3 of \emph{loc.cit.}, it is also
noted that
\[ \nats^{\tup{n}} = \frac{A_n}{(1-\X)^n}
\]
where $A_n$ is the $n$th Eulerian polynomial\footnote{The $n$th Eulerian polynomial is $A_n(x) = \sum_{k=0}^m A(n,k)x^k$ where the $A(n,m)$ are the Eulerian numbers, see e.g.~\cite[Sec.~6.2]{GKP94} or the Wikipedia entry on Eulerian Numbers.}.
\end{exa}


\begin{rema}
In much of this section, we could have weakened our assumptions on $A$.
As mentioned already, 
the matrix solution method only requires $A$ to be a commutative ring.
For the notion of linear automata, we only need  $A$ to be a semiring, see the next section for a definition. A linear automaton would then be an automaton whose state space is a semimodule over $A$, rather than a vector space. Lemma~\ref{lem:final-lin} and Proposition~\ref{prop:linear-solutions} would still hold, i.e., coinduction for automata over semimodules can be used as a solution method. An analogue of Proposition~\ref{prop:class-linear} does not hold for arbitrary semirings, but we would have the following version of 1 $\Ra$ 2:
If $A$ is a so-called Noetherian semiring (cf.~\cite{EzikMal2011,BMS12})
and $\sigma$ is a solution to a finite linear equation system, then the sub-semimodule generated by $\sigma$ is finitely generated.
\end{rema}

\section{Context-free Specifications}
\label{sec:context-free-specs}

We recall equation  \eqref{eq:catalan} (on page \pageref{eq:catalan}):
\[
\sig(0)=1, \qquad \sigma' = \sig\times\sig
\]
which defines the stream of Catalan numbers.
It is neither simple nor linear, as the righthand side of
the equation uses the convolution product. 
In the present section, we will study the class of \emph{context-free} SDEs
to which this example belongs.

In this section, we assume that $A$ is a {\it commutative semiring}.
A \emph{semiring} is an algebraic structure $(A,+,\cdot,0,1)$ where $(A,+,0)$  
is a commutative monoid, $(A,\cdot,1)$ is a monoid, 
multiplication distributes over addition, and
$0$ annihilates. A semiring  $(A,+,\cdot,0,1)$ is \emph{commutative}, if also $(A,\cdot,1)$ is a commutative monoid.
The full axioms for commutative semirings are, for $a,b,c \in A$:
\begin{equation}\label{eq:semiring-axioms}
\begin{array}{lll}
(a+b)+c = a+(b+c)\quad
& 0+a = a
& a+b = b+a\\
(a\cdot b)\cdot c = a\cdot(b\cdot c)
& 1 \cdot a = a 
& a \cdot b = b \cdot a\\
a\cdot (b + c) = a \cdot b + a\cdot c
& (a+b) \cdot c = a\cdot c + b\cdot c \quad
& 0 \cdot a = a
\end{array}
\end{equation}

Examples of commutative semirings include the natural numbers $\bbN$ with
the usual operations, and more generally any commutative ring such as the integers $\bbZ$.
An important finite commutative semiring is the Boolean semiring
$(2, \lor,\land,\bot,\top)$. More exotic examples include the
tropical (min-plus) semiring
$(\bbR\cup\{\infty\}, \mathrm{min}, +, \infty, 0)$ and the max-plus semiring
$(\bbR\cup\{-\infty\}, \mathrm{max}, +, -\infty, 0)$.
The semiring of languages over an alphabet $K$
(with language concatenation as product)
$(\Pow(K^*), \cup,\cdot,\emptyset,\{\emptyword\})$
is an example of a non-commutative semiring, i.e., one in which the product 
is not commutative.

For any semiring $A$, we can define stream constants $[a]$ for $a\in A$,
elementwise addition $+$ and convolution product $\times$
on $\strms{A}$ using the SDEs in Section~\ref{ssec:stream-calc}.
The algebraic structure $(A^\omega,+,\times,\cns{0},\cns{1})$
is again a semiring (cf.~\cite[Thm.4.1]{Rut03:TCS-bde}) and the inclusion
$a \mapsto [a]$ is a homomorphism of semirings. We will therefore simply
write $a$ to denote the stream $[a]$.
Note that the convolution product is
commutative if and only if the underlying semiring multiplication $\cdot$
is commutative.
For notational convenience, we will write $\tau\sigma$ instead of
$\tau\times\sigma$ for all $\tau,\sigma\in\strms{A}$.

\subsection{Context-free equation systems}

Let
$\M(X^*) = \{a_0w_0 + \cdots + a_n w_n \mid a_i \in A, w_i \in X^*\}$
denote the set of formal linear combinations over the set $X^*$
of finite words over $X$, or equivalently,
the set of polynomials over (non-commuting) variables in $X$ 
with coefficients in $A$.
$\M(X^*)$ is again a semiring with the usual addition and multiplication
of polynomials.
If we take $A$ to be
the Boolean semiring then $\M(X^*)$ is the semiring of languages over
alphabet $X$.
We also note that $\M(X^*)$ contains $A$ as a subsemiring via the inclusion
$a \mapsto a \emptyword$, where $\emptyword$ denotes the empty word.
Since we assume $A$ is commutative, $\M(X^*)$ 
is a semiring generalisation of the notion of a unital associative algebra
over a commutative ring.

A \emph{context-free equation system over set $X = \{x_i \mid i \in I\}$}
is a collection of SDEs, one for each $x_i \in X$,
of the form
\[ \begin{array}{lclclcl}
x_i(0) &=& a_i, & \qquad x_i' &=& y_i;\\
\end{array}\]
where $a_i \in A$ and $y_i \in \M(X^*)$ for all $i \in I$.
In other words,
a context-free equation system is a map
$\es=\tup{o,d} \colon X \to A \times \M(X^*)$.

As in the linear case,
a {\it solution of $\es$} is an assignment $h\colon X \to A^\omega$
that preserves the equations, that is, for all $x \in X$,
if $d(x) = a_1w_1 + \ldots + a_nw_n$, then
\[
   h(x)(0) = o(x)  \quad\text{ and }\quad
   h(x)' = a_1 h^*(w_1) + \cdots + a_n h^*(w_n)
\]
where $h^*(x_1 \cdots x_n) = h(x_1) \times \cdots \times h(x_n)$.
We call a stream \emph{$\sig$ context-free} if $\sig$ is the solution of some finite
context-free equation system.

The name \emph{context-free} comes from the fact that
a finite context-free equation system
$\es=\tup{o,d}\colon X \to A \times \M(X^*)$
corresponds to an $A$-weighted context-free grammar in Greibach normal form
with non-terminals in $X$ for a one-letter alphabet $L=\{\lambda\}$
as follows:
\[\begin{array}{lcl}
\text{equation system} && \text{grammar rules}\\
\hline
o(x) = a & \quad\text{ iff }\quad & x \to_a \emptyword\\
d(x)(w)=a & \quad\text{ iff }\quad & x \to_a \lambda w, \quad w \in X^*
\end{array}\]
where $x \to_a \lambda w$ denotes that $x$ can produce $\lambda w$
with weight $a$.
By taking $A$ to be the Boolean semiring $\mathit{2}$
and allowing an arbitrary alphabet $L$, a
context-free grammar in Greibach normal form is a system
of type $X \to \mathit{2}\times \M(X^*)^L$.

\subsection{Solutions and characterisations}
\label{sec:context-free-solutions}

\begin{prop}\label{prop:context-free-solutions}
Every context-free equation system has a unique solution.
\end{prop}
\proof
Similar to the linear case, we can construct from
$\es\colon X \to A \times \M(X^*)$
a stream automaton 
$\es^\flat\colon\M(X^*) \to A \times\M(X^*)$,
and apply coinduction to obtain a solution
$X \sgoes{\eta_X} \M(X) \sgoes{g} \strms{A}$ as shown in this diagram:
\begin{equation}\label{eq:context-free-solutions}
\xymatrix@R=+10mm{
X \ar[d]_-{\es} \ar[r]^-{\eta_X}  
  & \M(X^*) \ar[dl]^<<<<<{\es^\flat} \ar[r]^-{g} 
  & A^\omega \ar[d]^-{\zeta}\\
A \times \M(X^*) \ar[rr]^-{\id_A\times g}
  &
  & A \times A^\omega
}
\end{equation}
where this time $\eta_X\colon X \to \M(X^*)$ denotes the inclusion
of variables as polynomials.
We refer to \cite{BRW:CF-pow,WBR:CF-LMCS} for details.
\qed

At present, there are no analogues of Propositions \ref{prop:class-linear} and \ref{prop:rational-linear-char} for context-free streams, but it follows from \cite[Theorem 23]{BRW:CF-coalg} that context-free streams over $A$ are exactly the constructively $A$-algebraic power series over a one-letter alphabet, since streams over $A$ can be viewed as formal power series over a one-letter alphabet with coefficients in $A$.

In Section~\ref{ssec:matrix-solution-method}, we saw that
solutions to linear equation systems
are definable in stream calculus as the
rational streams.
For context-free streams, no such closed form is known, in general.

We end this section with some more examples of context-free streams.
\begin{exa}[Catalan numbers]
Let $A = \bbN$ be the semiring of natural numbers.
The context-free SDE from equation \eqref{eq:catalan}
\[
\gamma(0)=1, \qquad \gamma' = \gamma\times\gamma
\]
defines the sequence $\gamma = (1, 1, 2, 5, 14, 42, 132, 429, 1430, \ldots)$
of \emph{Catalan numbers}, cf.~\cite{BRW:CF-pow}.
In \cite[p.~117-118]{Rut05:MSCS-stream-calc},
it is shown that the Catalan numbers satisfy
\[
\gamma = \; \frac{2}{1 + \sqrt{1-4\X}}
\]
where the square root of a stream $\sigma$
is defined by the following SDE (cf.~\cite[section 7]{Rut05:MSCS-stream-calc}):
\begin{equation}\label{eq:sde-sqrt}
\begin{array}{ll}
\sqrt\sig(0)  = \sqrt{\sig(0)} \qquad\quad
&
(\sqrt\sig)' = \frac{\sig'}{\sqrt{\sig(0)} + \sqrt{\sig}}
\end{array}
\end{equation}
\end{exa}

\begin{exa}[Schr\"{o}der numbers]
The solution for the stream differential equation
\[
\sig(0) = 1, \quad \sig' = \sig + (\sig\times\sig)
\]
is the sequence
$\sig = (1, 2, 6, 22, 90, 394, 1806, 8558, 41586, \ldots)$
of \emph{(large) Schr\"{o}der numbers}
(sequence A006318 in \cite{OEIS}), see also \cite{WBR:CF-LMCS}.
For $n\in \bbN$, $\sig(n)$ is the number of paths in the
$n\times n$ grid from $(0,0)$ to $(n,n)$
that use only single  steps going right, up or diagonally right-up,
and which do not go above the diagonal.
In contrast with the Catalan numbers, we do not know of any
stream calculus expression that defines the stream of Schr\"{o}der numbers.
\end{exa}

\begin{exa}[Thue-Morse]\label{exam:Thue-Morse}
This example is a variation on a similar example in~\cite{BRW:CF-pow}.
Let $A = \mathbb{F}_2$, the finite field
$\{0,1\}$ where $1+1=0$.
The following context-free system of equations
\[ \begin{array}{rclcrcl}
\tau(0) &=& 0, & \qquad & \tau' &=& (\mu \times \mu) +
(\X \times \sigma \times \sigma),\\
\sigma(0) &=& 1, & \qquad & \sigma' &=& (\sigma \times \sigma) +
(\X \times \nu \times \nu),\\
\mu(0) &=& 1, & \qquad & \mu' &=& (\tau \times \tau) +
(\X \times \nu \times \nu),\\
\nu(0) &=& 0, & \qquad & \nu' &=& (\nu \times \nu) +
(\X \times \sigma \times \sigma).\\
\end{array}\]
defines the so-called Thue-Morse sequence
\[
\tau = \, (0,1,1,0,1,0,0,1, \ldots)
\]
which, in the world of automatic sequences \cite{AS03},
is typically defined by means of a finite (Moore)
automaton. We return to automatic sequences
in Section~\ref{sec:non-standard}.
Note that we could include the definition of $\X$
in the system above by adding the equations:
\[ \begin{array}{rclcrcl}
\X(0) &=& 0, & \qquad & \X' &=& [1]\\
{[1]}(0) &=& 1, & \qquad & [1]' &=& [0]\\
{[0]}(0) &=& 0, & \qquad & [0]' &=& [0]\\
\end{array}\]
\end{exa}

\begin{exa}
The following example is taken from \cite{Rut:counting-waut}, and
is not actually context-free since it uses the shuffle product $\otimes$ -- rather than the convolution product --
which is defined by the following SDE:
\begin{equation}
\label{SDE for shuffle}
(\sigma \otimes \tau )(0) = \sigma(0) \cdot \tau(0), \;\;\;\;
(\sigma \otimes \tau )' = (\sigma' \otimes  \tau)  +  (\sigma \otimes \tau')
\end{equation}
But observing that 
$(\strms{A},+,\otimes,\cns{0},\cns{1})$ also forms a semiring,
it can be viewed as context-free with respect to this structure.
Let $A = \bbN$, and consider the SDE
\[\sigma' = 1 + (\sig\otimes\sig), \quad\sig(0)=1
\]
Its solution is the stream
\[ \sig = (1, 2, 4, 16, 80, 512, 3904, 34816, 354560, \ldots)
\]
which is the sequence A000831 in \cite{OEIS}.
The stream $\sig$ can be
described in stream calculus by a so-called
continued fraction (cf.~\cite[section 17]{Rut01:MFPS-stream-calc}), as follows:\enlargethispage{2\baselineskip}
\[
\frac{\X}{1 - \dsty\frac{1\cdot 2\cdot\X^2}{1-\dsty\frac{2\cdot 3\cdot\X^2}{1-\dsty\frac{3\cdot 4\cdot\X^2}{\ddots}}}}
\]
Again, we do not know of any closed stream calculus expression that defines this stream.
\end{exa}


\section{Non-standard Specifications}
\label{sec:non-standard}

All stream definitions that we discussed so far make use of the same concrete, ``canonical'' representation of streams: a stream
of elements of $A$ consists of a first element $\sigma(0) \in A$ (the ``head'') followed by another stream $\sigma' \in A^\omega$ (the ``tail'').
There are, however, many other possible stream representations and each of these different, ``non-standard'' representations yield new ways of defining streams and stream functions. We are now going to discuss a few of these alternative stream  representations and the resulting non-standard stream specifications.

\subsection{Stream representations}

Let us start by explaining what we mean by a stream representation:
A representation for streams over some set $A$  is a collection of functions that can be combined in order to turn
the set $A^\omega$ into a final stream automaton (possibly of a ``non-standard'' type; for example, we are going to encounter stream representations
that require automata in which states have two instead of one successor). This intuition has been made more precise in~\cite{KR:cocoop} where the corresponding, slightly more general notion is called a complete set of cooperations. Here we confine
ourselves to listing a few examples.

\begin{exa}\label{ex:representations}\hfill
\begin{enumerate}
\item\label{ex:h-delta} We can supply the set $A^\omega$ of streams over a field  $A$ with the following
structure.
For $\sigma \in A^\omega$ we define
\[
\Delta \sigma = \;
(\sigma(1) - \sigma(0) , \,
\sigma(2) - \sigma(1) , \,
\sigma(3) - \sigma(2) , \, \ldots )
\]
(cf. \cite{SloaneP95,PavlovicE98,Rut05:MSCS-stream-calc}).
The $\Delta$-operator plays a central role in the area of Finite Difference Calculus~\cite{Boole} and is often referred to as the forward
difference operator.  It can be seen as a discrete derivative operator
for integer functions and provides a tool for finding recurrence relations in integer sequences
(cf.~e.g.~\cite[Section~2.5]{SloaneP95}).
 It is not difficult to see that $A^\omega$ together with the map
\[
\tup{\zero{(\_)}, \, \Delta} : A^\omega \to A \times A^\omega
\;\;\;\;\;\;
\sigma \mapsto \, \tup{\sigma(0) , \, \Delta \sigma}
\]
is a final stream automaton.
\item\label{ex:h-diff} Another structure on $A^\omega$ is obtained
by defining
\[
\frac{d}{dX} \sigma = \,
( \sigma(1) , \, 2 \cdot \sigma(2) , \, 3 \cdot \sigma(3) , \, \ldots )
\]
for $\sigma \in A^\omega$. Again $(A^\omega, \langle \zero{(\_)}, \frac{d}{dX} \rangle)$
is a final stream automaton. 
The operator $\frac{d}{dX}$ computes the derivative of
a formal power series and has been used in~\cite{PavlovicE98} in order to establish a connection
between calculus and the theory of coalgebras.
\item In a similar fashion lots of examples could be designed: Given a set $A$ together with
some operation $o:A \times A \to A$, we define
\[
\Delta_o \sigma = \;
(o(\sigma(0),\sigma(1)) , \,
o(\sigma(1),\sigma(2)) , \,
o(\sigma(2),\sigma(3)) , \, \ldots )
\]
and we can see that $A^\omega$ together with the map $\tup{\zero{(\_)}, \, \Delta_o} : A^\omega \to A\times A^\omega$ is a final stream automaton provided that for any $a \in A$
the map $\lambda b. o(a,b)$ has an inverse.
\end{enumerate}
\end{exa}

\noindent The fact that $\tail$, $\Delta$ and $\frac{d}{dX}$ all give rise to a final stream automaton structure implies that there are unique stream isomorphisms between these three structures. These isomorphims can be viewed as transforms which leads to a fascinating coinductive approach to analytic calculus as first observed in \cite{PavlovicE98}. More recently, the Newton transform between the $\Delta$- and $\tail$-structures has been studied in \cite{BHPR:Newton-ICTAC}.

But non-standard stream representations are not limited to standard  stream  automata as the following two interesting examples show. In order to formulate them we need the notion of a $2$-stream automaton which 
generates an infinite binary tree {\em representing a stream} rather than a stream of symbols directly.
\begin{defi}
	A \emph{$2$-stream automaton} is a set $Q$ (of states) together with a function
	$\tup{o,d_0,d_1}:Q \to A \times Q \times Q$. A morphism
	between two $2$-stream automata $\tup{o,d_0,d_1}:Q \to A \times Q \times Q$ and $\tup{p,e_0,e_1}:P \to A \times P \times P$ is a function
	$f: Q \to P$ such that $p(f(q)) = o(q)$ and $e_i(f(q)) = f(d_i(q))$ for
	$i=0,1$ and for all $q \in Q$.
\end{defi}
The above definition has an obvious generalisation to $k$-stream automata.
Note that in this sense a stream automaton is just a $1$-stream automaton.

In Example \ref{ex:streams_as_trees} below,
we describe two ways of representing the set of streams
as a final $2$-stream automaton. These representations use the stream operations
$\even\colon\strms{A}\to\strms{A}$ and
$\odd\colon\strms{A}\to\strms{A}$:
\begin{eqnarray}
     \even(\sigma) & \coloneqq & (\sigma(0), \sigma(2), \sigma(4), \dots ) \label{eq:def-even}\\
     \odd (\sigma) & \coloneqq & (\sigma(1), \sigma(3), \sigma(5), \dots) \label{eq:def-odd}
\end{eqnarray}

\begin{exa}\label{ex:streams_as_trees}\hfill
Here are two examples of non-standard stream representations, based on $2$-stream automata.
  \begin{enumerate}
    \item\label{oddeven}
     The $2$-stream automaton with state set $\strms{A}$ and structure map
     $$ \sigma \mapsto \langle \zero{\sigma},\even (\sigma'),\odd (\sigma') \rangle:A^ \omega \to A \times A^\omega \times A^\omega$$
     is final among all $2$-stream automata (cf.~\cite{EGHKM,HKRW:k-regular}).
    \item\label{evenodd} The set $A^\omega$ together
    with the structure map
    $$\sigma \mapsto \tup{\zero{\sigma}, \even(\sigma),\odd(\sigma)}: A^\omega \to A \times A^\omega \times A^\omega$$
    is not final
    among all $2$-stream automata but among all {\it zero-consistent} $2$-stream automata (cf.~\cite{KR:autseq}), i.e.,
    among all $2$-stream automata
	$(Q,\tup{o,d_0,d_1})$ such that
	for all $q \in Q$ we have
	$o(d_0(q)) = o(q)$.
In Section~\ref{sec:automatic}, we will see that this slightly weaker finality property is sufficient for obtaining a syntactic stream definition format.
  \end{enumerate}
\end{exa}

\subsection{Simple non-standard specifications}
Next we discuss stream specifications that use the above non-standard stream representations.
 The first thing to note is that for the representations in Example~\ref{ex:representations} we can easily define
 non-standard variations of the simple, linear and context-free specifications discussed earlier.

 This can be done as follows: given any of the non-standard tail operations $\partial \in \{ \Delta, \frac{d}{dX} ,\Delta_o \}$
 and a simple, linear or context-free equation system over a set $X= \{x_i \mid i\in I\}$ of variables with
 \[ x_i(0) = a_i \qquad \mbox{ and } \qquad x_i' = y_i  \qquad \mbox{ for } i \in I,\]
we call the system of equations
\[ x_i(0) = a_i \qquad \mbox{ and } \qquad \partial(x_i) = y_i  \mbox{ for } i \in I,\]
obtained by replacing all derivatives $x_i'$ with the non-standard derivatives
$\partial (x_i)$,
a simple, linear or context-free $\partial$-specification, respectively.
As before, 
solutions for such systems of equations are functions $h:X \to A^\omega$ that preserve the equations.
As in the standard case, existence of unique solutions is guaranteed by the fact that each non-standard stream representation induces a final coalgebra on the set of streams.

\begin{exa}\label{ex:simple_nonstandard}
	Let $A=\mathbb{R}$ be the field of real numbers. The equations
	\[ x(0) = 1, \qquad \Delta(x) = x \]
	are an example of a simple $\Delta$-specification of the stream
	\[ (1,2,4,8, \dots).\]
	Similarly, the equations
	\[ x(0) = 1, \qquad  \frac{d }{dX} (x) = x \]
	are a simple $ \frac{d}{dX}$-specification of the stream
	\[ \left(\frac{1}{0!}, \frac{1}{1!},  \frac{1}{2!}, \frac{1}{3!}, \frac{1}{4!},\dots \right). \]
\end{exa}
The following proposition is folklore and provides a large class of examples of streams
that can be defined using simple $\Delta$-specifications.

\begin{prop}\label{prop:poly_Delta}
	Let $d \in \bbN$. For all streams $\sigma \in \bbR^\omega$ we have 	
	$\Delta^d (\sigma) = (0,0,0,0,\dots)$ 
	iff there exists a polynomial $\varphi(x)$ over $\bbR$  of degree $<d$ such that $\sigma(n) = \varphi(n)$ for
	all $n \geq 0$.
\end{prop}
\proof
	In order to simplify the notation in the proof, we write $\lambda n.(a_0 + a_1 n + \dots + a_d n^d)$ to denote
	the stream $\sigma$ defined, for all $n \geq 0$, by
	\[
	\sigma(n) = \, a_0 + a_1 n + \dots + a_d n^d
	\]
	Clearly, we have $\lambda n.a = (a,a,\dots)$, i.e., if the expression in the scope
	of $\lambda n$ does not contain a reference to $n$, the stream is constant.
	Furthermore we use the easily verifiable fact that $\Delta(\sigma + \tau) = \Delta(\sigma) + \Delta(\tau)$ for
	all streams $\sigma,\tau \in \bbR^\omega$.

	Suppose first that there exists some polynomial
	\[ \varphi(x) = a_0 + a_1 x + \dots + a_d x^d , \qquad a_0,\dots,a_d \in \bbR, a_d \not= 0\]
	of degree $d$ with $\sigma(n) = \varphi(n)$ for all $n \in \bbN$, i.e.,
	$\sigma = \lambda n.\varphi(n)$. 
	The following claim suffices to obtain $\Delta^{d+1}(\sigma) = 0$ as required:
		
	{\noindent \bf Claim} $\Delta^d(\sigma) = \lambda n. a_d d!$ \\
	The proof of the claim is by induction on $d$. 
	\begin{description}
		\item[Case] $d=0$. Then $\sigma = \lambda n.a_0$ and
		$\Delta^0(\sigma) = \sigma = \lambda n. a_0 0!$ 
		as required.
		\item[Case] $d = k + 1$. Then
		\begin{eqnarray*}
			\Delta^{k+1} (\sigma) & = & \Delta^{k+1}(\lambda n.(a_0 + a_1 n +\dots + a_{k+1} n^{k+1}))  \\	
				& = & \Delta^{k+1} (\lambda n.(a_0 + a_1 n + \dots + a_k n^k)) + \Delta^{k+1}(
				\lambda n.(a_{k+1} n^{k+1})) \\
				& \stackrel{\mbox{\tiny I.H.}}{=} & 0 + \Delta^k (\lambda n. (a_{k+1} (n+1)^{k+1} - a_{k+1}  n^{k+1})) \\
				& = &  \Delta^k (\lambda n.(a_{k+1} n^{k+1} + a_{k+1} \binom{k+1}{k} n^k + r(n) - a_{k+1} n^{k+1})) \\
				& & \mbox{ where } r(n) \mbox{ is a poly of degree } < k \\
				& = & \Delta^k(\lambda n.(a_{k+1} (k+1) n^k + r(n))) \\
				& & \mbox{ where } r(n) \mbox{ is a poly of degree } < k \\
				& \stackrel{\mbox{\tiny I.H.}}{=}  & \lambda n. a_{k+1} (k+1) k! + 0 = \lambda n. a_{k+1} (k+1)!
		\end{eqnarray*}
	\end{description}
	
	Conversely, consider  a stream $\sigma$ such that $\Delta^d(\sigma) = (0,0,0,\dots)$ and suppose
	that $d \in \bbN$ is the minimal such $d$. In case $d = 0$ there is nothing to prove.
	If $d > 0$ we have that $\Delta^{d-1}(\sigma) = (r,r,\dots)$ for some $r \not= 0$.
	Define $a := \frac{r}{(d-1)!}$. Then we put $\tau \mathrel{:=} \sigma - \lambda n. (a n^{d-1})$ such that $\sigma = \tau + \lambda n. (a n^{d-1})$.
	By the claim that we proved above this implies
	\[ \Delta^{d-1}(\sigma) = \Delta^{d-1}(\tau + \lambda n.(a n^{d-1})) = \Delta^{d-1}(\tau) + \cns{a (d-1)!} = \Delta^{d-1}(\tau) + \cns{r} .\]
	This clearly implies $\Delta^{d-1}(\tau) = (0,0,0,0,\dots)$ and hence we can apply the I.H. to $\tau$ in order to
	obtain a polynomial $\varphi(x) = a_0 + a_1 x + \dots + a_{d-2} x^{d-2}$ such that
	$\tau = \lambda n.(a_0 + a_1 n + \dots + a_{d-2} n^{d-2})$. This implies
	$\sigma = \lambda n.(a_0 + a_1 n + \dots + a_{d-2} n^{d-2} + a n^{d-1})$, i.e., for all $n \in \bbN$
	we have $\sigma(n) = \psi(n)$ for some polynomial $\psi(x)$ of degree $< d$.	
\qed

We are now going to compare finite simple/linear/context-free $\Delta$- and $ \frac{d}{dX} $-specifications
to the corresponding $\tail$-specifications of real-valued streams $\sigma \in \bbR^\omega$.
It is not too difficult to see that the set of streams that have a finite simple $\Delta$-specification
and the set of streams that have a finite simple $\tail$-specification are incomparable. This is demonstrated by the following examples:

\begin{exa}\hfill
	\begin{enumerate}
		\item Recall that a stream $\sigma$ has finite simple
			$\tail$-specification iff $\sigma$ is ultimately periodic.  Therefore the stream
			$$\sigma = (0, 1, 0, 1, 0, 1,\dots) \in \bbZ^\omega$$
			has a finite simple $\tail$-specification. One can prove by induction that
			$\sigma$ has infinitely many distinct
			$\Delta$-derivatives which implies that $\sigma$ does not have a finite simple $\Delta$-specification.
                        However, when $A = \bbZ/n\bbZ$ is a finite ring, $\sigma$ is definable by a 
                        finite simple $\Delta$-specification\footnote{This observation is thanks to Michael Keane and 
                        Henning Basold.}.
		\item It follows from Proposition~\ref{prop:poly_Delta} that the stream
		        $$ \sigma = (0,1,2^2,3^2,4^2, \,\dots)$$
		       has a finite simple $\Delta$-specification, but obviously no
		        finite simple $\tail$-specification.
	\end{enumerate}
\end{exa}

\noindent Finite linear  $\Delta$-specifications define the same class of streams as their standard linear counterparts.
This follows from the fact that
\[ \Delta(\sigma) = \sigma' - \sigma \qquad \mbox{ and } \qquad \sigma' = \Delta(\sigma) + \sigma .\]
Therefore any linear specification can be replaced by the equivalent $\Delta$-specification:
\[
\left.
\begin{array}{l}
	x_i(0) =  a_i \\
	x_i' = t
\end{array} \right\}
\quad \Rightarrow  \quad \left\{
\begin{array}{l}
	x_i(0) = a_i \\
	\Delta(x_i)= t - x_i
\end{array}
\right.
\]

Vice versa, any linear $\Delta$-specification can be easily transformed into
an equivalent standard one.  Analogously, context-free $\Delta$-specifications and
standard context-free specifications define precisely the same class of streams.
We summarise our observations in the following proposition.

\begin{prop}
	The set of streams $\sigma \in \bbR^\omega$ definable with finite simple $\tail$-speci\-fications and the set of streams
	definable with finite simple $\Delta$-specifications are
	incomparable. 
	 Furthermore we have the following equi\-valences:
	\begin{itemize}
		\item  Any stream $\sigma \in \bbR^\omega$ 
		is definable with
		a finite linear $\tail$-specification iff $\sigma$ is definable with a finite
		linear $\Delta$-specification.
		\item Any stream
	$\sigma \in \bbR^\omega$ 
	is definable with
	a finite context-free $\tail$-specification  iff $\sigma$ is definable with a finite
	context-free $\Delta$-specification.
	\end{itemize}
\end{prop}

\noindent When comparing $\frac{d}{dX}$-specifications with $\tail$-specifications, the following identities for
arbitrary streams $\sigma \in \bbR^\omega$ are
useful:
\begin{eqnarray}
	\frac{d}{dX}(\sigma) & = & \sigma' \odot \nats  \label{ident_ddX_1} \\
	\sigma' & = & \frac{d}{dX}(\sigma) \odot \nats^{-1}  \label{ident_ddX_2}
\end{eqnarray} 
where
\begin{eqnarray*}
	\nats  & = & (1,2,3,4,\dots) \\
	\nats^{-1} & = & (1, \frac{1}{2},\frac{1}{3},\frac{1}{4},\dots)
\end{eqnarray*}
and where $\odot$ denotes the so-called Hadamard-product (element-wise multiplication) given by
\[ \sigma \odot \tau = (\sigma(0)\tau(0),\sigma(1)\tau(1),\sigma(2)\tau(2),\dots).\]
This means that any simple $\tail$-specification can be replaced
by a simple $\frac{d}{dX}$-specification
in which we are also allowed to employ $\odot$ and $\nats$:
\[ \left. \begin{array}{l}
	x_i(0) =  a_i \\
	x_i' = x_j
\end{array} \right\}
\quad \Rightarrow  \quad \left\{
\begin{array}{l}
	x_i(0) = a_i \\
	\frac{d}{dX} (x_i)= x_j \odot \nats
\end{array}
\right.
\]
Similarly any simple $\frac{d}{dX}$-specification can be replaced by a simple
$\tail$-specification in which we are allowed to use $\odot$ and $\nats^{-1}$:
\[ \left. \begin{array}{l}
	x_i(0) =  a_i \\
	\frac{d}{dX} (x_i) = x_j
\end{array} \right\}
\quad \Rightarrow  \quad \left\{
\begin{array}{l}
	x_i(0) = a_i \\
	x_i'= x_j \odot \nats^{-1}
\end{array}
\right.
\]
We use the description \emph{simple $\tail$-$\nats^{-1}$-specification}
for a simple $\tail$-specification that may contain $\odot \nats^{-1}$ on the right hand side
of the equation for the derivative a. Similarly, we define
\emph{simple $\frac{d}{dX}$-$\nats$-specifications}. The above identities can be used to show
that simple $\tail$-$\nats^{-1}$-specifications and simple $\frac{d}{dX}$-$\nats$-specifications
are equally expressive.

Note that without the extension by $\odot$, $\nats$ and $\nats^{-1}$ the simple $\tail$- and
$\frac{d}{dX}$-specifications are incomparable, as the following
example shows:

\begin{exa}\hfill
	\begin{enumerate}
		\item The stream $\sigma = (1,1,1,1,\dots)$ has a simple $\tail$-specification but
		   no simple $\frac{d}{dX}$-specifica\-tion. In order to see the second statement, we use
		   (\ref{ident_ddX_1}) and (\ref{ident_ddX_2})
		   to compute:
		   \[\begin{array}{lcl}
		   	\frac{d}{dX}(\sigma) & = & \nats \\[.5em]
			\frac{d}{dX}(\nats)  & = & \nats + \nats \odot \nats
		   \end{array}\]
		   and from here onwards it is easy to see that all the derivatives
		   $(\frac{d}{dX})^n(\sigma)$ for $n \in \bbN$ will be distinct and thus that $\sigma$
		   has no finite simple $\frac{d}{dX}$-specification.
		\item The stream $\sigma = (1,1, \frac{1}{2!},\frac{1}{3!}, \dots)$ has a finite simple
		   $\frac{d}{dX}$-specification (cf.~Ex.~\ref{ex:simple_nonstandard}) but obviously no
		   finite simple $\tail$-specification.
	\end{enumerate}
\end{exa}
\subsection{Stream specifications for automatic sequences}
\label{sec:automatic}

We conclude this section by discussing stream specifications that make
use of the stream representation from
Example~\ref{ex:streams_as_trees}.\ref{evenodd}.
We refer to Remark~\ref{rem:even-odd-of-tail} below for a discussion on
how these results could be obtained for
the representation from Example~\ref{ex:streams_as_trees}.\ref{oddeven}.

\begin{defi}{}\label{def:oddEven}
	A \emph{simple $\even$-$\odd$-stream specification over a set $X=\{x_i \mid i \in I \}$} of variables
	contains for every $x_i \in X$ three equations:
	\[ x_i(0) = a, \qquad \even(x_i) = y^i_1, \qquad \odd(x_i) = y^i_2 \]
	where $a \in A$ and $y^i_1, y^i_2 \in X$ and where
	the equations entail that 
	\begin{equation}\label{equ:zero}
	x_i(0) = (\even(x_i))(0).
	\end{equation} 
	The notion of entailment
	can be formalised using conditional equational logic as demonstrated in~\cite{KR:cocoop}.	
	Solutions are again functions $h:X \to A^\omega$ preserving the equations.
\end{defi}
Simple $\even$-$\odd$-stream specifications are called $\mathsf{zip}$-specifications in \cite{EGHKM}.
Note that an $\even$-$\odd$-stream specification is not a stream differential equation as the stream derivative is nowhere used.
Nevertheless, as shown in \cite{KR:cocoop,EGHKM}, an $\even$-$\odd$-stream specification is a syntactic representation of a type of stream automaton, namely, of a
zero-consistent 2-stream automaton (cf.~Example~\ref{ex:streams_as_trees}.\ref{evenodd}).

\begin{lem}
There is a 1-1 correspondence between
simple $\even$-$\odd$-stream specifications over a set $X$
and zero-consistent 2-stream automata with state space $X$.
Consequently, every simple $\even$-$\odd$-stream specifications
has a unique solution.
\end{lem}
\begin{proof}
An $\even$-$\odd$-stream specification over a set $X$
defines a $2$-stream automaton with set of states $X$ in the obvious way:
$$ \gamma \coloneqq \tup{\zero{(\_)}, \even, \odd}: X \to A \times X \times X$$
As the equations of an $\even$-$\odd$-stream specification have to entail (\ref{equ:zero})
for all $x \in X$,
we have that $(X,\gamma)$ is {\em zero-consistent}.
Conversely, the output and transitions of a zero-consistent 2-stream automaton can be written in the form of a simple $\even$-$\odd$-stream specification.
Solutions are now easily seen to correspond to (the obvious notion of) homomorphism for (zero-consistent) 2-stream automata.
By finality of $(A^\omega,\tup{\zero{(\_)}, \even,\odd})$, (cf.\ Example~\ref{ex:streams_as_trees}.\ref{evenodd}), we obtain for every zero-consistent $2$-automaton with state space $X$, a unique homomorphism $h:X \to A^\omega$ which is the unique solution to the corresponding $\even$-$\odd$-stream specification.
\end{proof}

Our interest in $\even$-$\odd$-stream specifications is rooted in their close relationship to 
$k$-automatic sequences \cite{AS03}.
For simplicity, we only treat the case where $k=2$, but all definitions and results can be straightforwardly generalised for any natural number $k$. 

Let us first state the definition of the reverse binary encoding of natural numbers and of automatic sequences.

\begin{defi}{}\label{def:bbin_encoding}
	For $n \in \bbN$ we define the $\bbin$-encoding $\bbin(n)$ as the standard binary encoding read backwards, i.e., with
	the least significant bit first. For example: $\bbin(0)=\epsilon$, $\bbin(1)=1$, $\bbin(2) = 01$, $\bbin(5) = 101$, $\bbin(6) = 011$, etc.
\end{defi}
The following is one of several equivalent definitions of 2-automatic sequences.

\begin{defi}\label{def:automatic}
	A stream $\sigma \in A^\omega$ is called \emph{2-automatic} if
	it is generated by a finite zero-consistent $2$-automaton, i.e., if
	there exists a finite zero-consistent $2$-automaton
	$\mathcal{Q}_\sigma = (Q, \langle o,d_0,d_1 \rangle:Q \to A \times Q \times Q)$
	and a state $q_\sigma \in Q$	
	such that for all $n \in \bbN$ we have
	\[  \sigma(n) = o\left(d_{\bbin(n)}(q_\sigma)\right),\]
	where for $w \in 2^*$ the function $d_w:Q \to Q$ is inductively defined
	by $d_\epsilon(q) = q$ and $d_{wi}(q) = d_i(d_w(q))$.
	In other words, the $n$-th element of $\sigma$ is obtained as output
	from $\mathcal{Q}_\sigma$ by feeding the $\bbin$-encoding of $n$
	to the $2$-automaton $\mathcal{Q}_\sigma$ starting from position
	$q_\sigma$.
\end{defi}

The following characterisation result from \cite{KR:autseq} is now immediate.
\begin{thm}
Let $\sigma \in A^\omega$ be a stream over some alphabet $A$. The following are equivalent
\begin{enumerate}
\item $\sigma$ is 2-automatic.
\item $\sigma$ is the solution to a finite simple $\even$-$\odd$-stream specification.
\item The sub-automaton of $(A^\omega, \langle  \zero{(\_)},\even,\odd \rangle)$ generated by $\sigma \in A^\omega$ is finite.
\end{enumerate}
\end{thm}
The states of the sub-automaton mentioned in item 3 in the above theorem
are sometimes referred to as the $2$-kernel of $\sigma$. Hence another equivalent defintion of $2$-automaticity is to require that the $2$-kernel is finite, cf.~\cite{AS03}.

\begin{rem}\label{rem:even-odd-of-tail}
  The stream representation from Example~\ref{ex:streams_as_trees}.\ref{oddeven}
  gives rise to an automaton which is final among all 2-stream automata, and it corresponds to an \emph{$\even$-$\odd$-of-tail stream specification format} in which each equation specifies $\sigma(0), \even(\sigma')$ and $\odd(\sigma')$. Such specifications are equivalent to systems of stream differential equations of the form:
  \begin{equation}\label{eq:zip-tail-spec}
    x_i(0) = a, \qquad \sigma' = \zip(x_j,x_k)
    \end{equation}
  where the stream operation $\zip\colon \strms{A}\times\strms{A} \to \strms{A}$ is defined by:
  \[ \zip(x,y)(0) = x(0), \qquad \zip(x,y)' = \zip(y,x').\]
  It is easy to see that
  $\zip\colon \strms{A}\times\strms{A} \to \strms{A}$ and the pairing $\tup{\even,\odd} \colon\strms{A} \to \strms{A}\times\strms{A}$ are each others inverses. This is what yields the equivalence of the $\even$-$\odd$-of-tail format and \eqref{eq:zip-tail-spec}.
  One can show that every simple $\even$-$\odd$-stream specification can be transformed into one in the format given in \eqref{eq:zip-tail-spec}, and one obtains again a characterisation of 2-automatic streams, but with a different encoding of the natural numbers. For more details, we refer to \cite{HKRW:k-regular}, where also $k$-regular sequences are characterised in terms of solutions to a linear generalisation of the format in \eqref{eq:zip-tail-spec}.
\end{rem}

\begin{exa}
	As one example of a $\even$-$\odd$-specification consider the following simple
	specification of the Thue-Morse sequence from Example~\ref{exam:Thue-Morse}:
	\[
	\begin{array}{rclcrcl}
		\mathrm{TM}(0) & = & 0 & \qquad & N(0) & = & 1 \\
		\even(\mathrm{TM}) & = & \mathrm{TM} & & \even(N) & = & N \\
		\odd(\mathrm{TM})  & = & N & & \odd(N) & = & \mathrm{TM} \\
	\end{array}
	\]
	Clearly the given equations entail that $(\even(\mathrm{TM}))(0) = \mathrm{TM}(0)$
	and $(\even(N))(0) = N(0)$ as required by the definition of an $\even$-$\odd$-specification.
	The unique solution for this specification maps $\mathrm{TM}$ to the Thue-Morse sequence.
        Much more on this way of looking at automatic sequences can be found in~\cite{EGHKM,KR:autseq}. 
\end{exa}

\section{The Syntactic Method}
\label{sec:syntactic-method}

The examples of the previous sections illustrate the general approach to defining
streams and stream operations by systems of SDEs.
In this section, we discuss a general method for showing that 
many such systems of SDEs have a unique solution.
Because the method associates with each such system of SDEs
a set of terms, we call it \emph{syntactic}.
As we shall see, the method will work for all systems of SDEs
that satisfy a rather general condition on their (syntactic) shape.
Furthermore we will show that the various specific families of SDEs that we discussed in Sections \ref{sec:simple-specs},
\ref{sec:linear-specs}  and \ref{sec:context-free-specs} can be seen as instances of the syntactic method.
An earlier version of the material in this section
is found in \cite{KNR:SDE}.

The basic idea of the syntactic method is as follows.
Given a signature $\Sigma$ with operation symbols $\syn{f}$,
let $\term_\Sig(\strms{A})$ denote the set of
all $\Sigma$-terms over $\strms{A}$.
Any system of SDEs that
for each $k$-ary symbol $\syn{f}$ in $\Sig$ and any streams
$\sigma_1, \ldots, \sigma_k$
contains an SDE that defines
$\syn{f}(\sigma_1, \ldots, \sigma_k)$,
yields an inductive definition of a stream automaton
$\tup{o,d} \colon \term_\Sig(\strms{A}) \to A \times \term_\Sig(\strms{A})$
which has terms as states.
The stream solutions are obtained via coinduction:
\[\xymatrix{
\term_\Sig(\strms{A}) \ar@{.>}[r]^-{\beh{-}} \ar[d]_-{\tup{o,d}} & \strms{A} \ar[d]^-{\zeta} \\
A \times \term_\Sig(\strms{A})\ar@{.>}[r] &  A \times \strms{A}
}\]
The behaviour map $\beh{-}$ thus yields for each term $t \in\term_\Sig(\strms{A})$ a stream $\beh{t} \in \strms{A}$, in other words,
it defines an algebra (of signature $\Sigma$)
on the set of streams.
In particular, the stream defined by $\syn{f}(\sigma_1, \ldots, \sigma_k)$
is obtained as $\sem{\syn{f}(\sigma_1, \ldots, \sigma_k)}$.

\begin{exa}
\label{ex:gsos-spec}
In order to define the sequence of natural numbers as in
\eqref{eq:sde-nats}, we take
$A=\bbN$ and $\Sigma = \{\ones,\nats,+\}$
where $\ones$ and $\nats$ are 0-ary operations (constants),
and $+$ is binary. The associated (infinite) system of SDEs 
consists of all defining SDEs put together:
\[
\begin{array}{lclclcl}
\ones(0) &=& 1, & \quad & \ones' &=& \ones,\\
\nats(0) &=& 0, & & \nats' &=& \nats+\ones,\\
(\sigma+\tau)(0) &=& \sigma(0)+\tau(0), & & (\sigma+\tau)' &=& \sigma'+\tau', \qquad\qquad \text{for all } \sigma, \tau \in  \strms{A}.\\
\end{array}
\]
\end{exa}
\bigskip
The shapes of the SDEs seen so far are all instances of the
general format called \emph{stream GSOS}, cf.~\cite{Klin11}.
Informally stated, a system of SDEs is in the stream GSOS format if
for all $k$-ary operations $\syn{f}$ in $\Sigma$,
the SDE defining $\syn{f}$ has the shape:
\[
\syn{f}(\sig_1,\ldots, \sig_k)(0) = a, \quad
\syn{f}(\sig_1,\ldots, \sig_k)' = t 
\]
where $a\in A$ depends only on $\sig_1(0), \ldots, \sig_k(0)$,
and $t$
is a $\Sigma$-term over $\sig_1,\ldots, \sig_k,\sig'_1,\ldots, \sig'_k$
that depends only on $\sig_1(0), \ldots, \sig_k(0)$.

To see how things can go wrong when straying from the GSOS format,
consider the following SDE (for the signature $\Sig$ containing a single constant $\nm{c}$):
\begin{equation}\label{eq:non-GSOS-example}
 \nm{c}(0)=1, \quad \nm{c}' = \nm{c}'
\end{equation}
This SDE does not have a unique solution,
since any stream starting with a 1 is a solution, and
indeed \eqref{eq:non-GSOS-example} is not in the GSOS format.
The reason is that the derivative of $\nm{c}$ should be defined as a
term $t \in \term_\Sig(\emptyset)$, and
$\nm{c}' \notin  \term_\Sig(\emptyset) = \{\nm{c}\}$
(since the derivative operation is not part of the signature).
Moreover, note that it is not possible to extend the signature with the derivative operation.
This follows from the fact that all stream operations defined in the GSOS format are {\em causal} (as we will see in Proposition~\ref{prop:gsos-implies-causal}), a property which the derivative operation lacks.
We return to causal operations in Section~\ref{ssec:causal-ops}.

In the remainder of this section we present and prove the
correctness of the syntactic method for SDEs in the stream GSOS format.
This result follows from more general insights in the theory of
bialgebras and abstract GSOS, cf.~\cite{Bartels:PhD,Klin11,TuriPlotkin:LICS-GSOS},
and we give a brief summary of this more abstract, categorical
presentation in Section~\ref{sec:stream-gsos}. In the current section,
we wish to present a self-contained, elementary proof of this fact.


\subsection{Terms and algebras}
\label{ssec:terms-algebras}

A signature $\Sigma$ is a collection of operation symbols $\syn{f}$,
each of which has an arity $k$. Nullary  operations (with arity $0$)
are called constants, and unary operations are called functions.
We write $\Sig_k$ for the set of $k$-ary operations in $\Sig$.
The set of $\Sig$-terms over a set $X$ (of generators)
is denoted by $\term_\Sigma(X)$, and defined inductively
as the least set $T$ that contains $X$ and is closed under the following
formation rule: if $t_1, \ldots, t_k$ are in $T$ and 
$\syn{f}$ is in $\Sig_k$, $k \in \bbN$,
then $\syn{f}(t_1, \ldots, t_k)$ is in $T$.

A \emph{$\Sigma$-algebra} $\tup{X,\alpha}$ consists of a carrier set $X$
and a collection of maps
$\alpha =\{ f_\alpha \colon X^k \to X \mid \syn{f}\in\Sig_k, k \in \bbN \}$
containing for each $k$-ary operation $\syn{f} \in \Sig$,
a map $f_\alpha \colon X^k \to X$ interpreting $\syn{f}$.
A \emph{homomorphism of $\Sig$-algebras} from
$\tup{X,\alpha}$ to $\tup{Y,\beta}$ is a function
$h\colon X \to Y$ that respects the algebra structure, i.e.,
for all $\syn{f} \in \Sig_k$, $k \in \bbN$, and all $x_1, \ldots, x_k \in X$:
$h(f_\alpha(x_1,\ldots, x_k)) = f_\beta(h(x_1),\ldots, h(x_k))$.

For any $X$, the set $\term_\Sigma(X)$ of $\Sig$-terms over $X$
is a $\Sig$-algebra $\tup{\TSig(X),\gamma_\Sig}$ where $\gamma_\Sig$ is given by construction of terms.
In fact, it is the so-called \emph{free $\Sig$-algebra over $X$}
which means 
that  if $\tup{Y,\alpha}$ is a $\Sig$-algebra
and $h \colon X \to Y$ is a function mapping generators to elements in $Y$,
then there is a unique homomorphism
$\ext{h}\colon \tup{\term_\Sig(X),\gamma_\Sig} \to \tup{Y,\alpha}$ extending $h$
which is defined inductively by:
\[\begin{array}{rcll}
\ext{h}(x) &=& h(x)  & \text{ for all } x \in X,\\
\ext{h}(\syn{f}(t_1,\ldots, t_k)) &=& f_\alpha(\ext{h}(t_1), \ldots,\ext{h}(t_k)) & \text{ for all } \syn{f} \in \Sig_k, k\in \bbN.
\end{array}
\]
Note that every homomorphism $g\colon \tup{\term_\Sig(X),\gamma_\Sig} \to \tup{Y,\alpha}$
is determined by its action on the generators $X$.
In other words, there is a 1-1 correspondence
between homomorphisms $\tup{\term_\Sig(X),\gamma_\Sig} \to \tup{Y,\alpha}$
and maps $X \to Y$.
In particular, a $\Sig$-algebra
$\tup{X,\alpha}$ corresponds uniquely to a homomorphism
$\interp{\alpha}\colon\tup{\TSig(X),\gamma_\Sig} \to \tup{X,\alpha}$
(by taking $\interp{\alpha}$ to be the homomorphic extension $\ext{\id_X}$).
We call the homomorphism $\interp{\alpha}\colon\tup{\TSig(X),\gamma_\Sig} \to \tup{X,\alpha}$
the \emph{interpretation of $\Sig$-terms induced by $\alpha$}. 

Terms come equipped with the standard notion of substitution.
A \emph{substitution} is a homomorphism
$s\colon \term_\Sig(X)\to\term_\Sig(Y)$.
For a term $t \in \term_\Sig(X)$ over variables
$x_1, \ldots, x_k \in X$
and a substitution
$s$ for which $s(x_i) = s_i$ for $i=1,\ldots, k$,
we write $t[s_i/x_i]_{i \leq k}$ for
the result of applying the substitution $s$ to $t$.

\subsection{Stream GSOS definitions}
\label{ssec:stream-GSOS-defs}

In the rest of this section, let $\Sig$ be an arbitrary, but fixed signature.

\begin{defi}[Stream GSOS definition]
A {\em stream GSOS definition for $\syn{f} \in \Sig_k$, $k\in\bbN$},
is a pair $\tup{o_{\syn{f}},d_{\syn{f}}}$
(defining ``initial value'' $o_{\syn{f}}$ and ``derivative'' $d_{\syn{f}}$ 
of $\syn{f}$) where
  \begin{eqnarray*}
    && o_{\syn{f}}:  A^k \to A \\
    && d_{\syn{f}}:  A^k
  \to \term_\Sigma(\{x_1,\ldots,x_k,y_1,\ldots,y_k\})
  \end{eqnarray*}
If $d_{\syn{f}}(a_1,\ldots,a_k)$ does not contain any of the $x_i$ variables,
then we say that $\tup{o_{\syn{f}},d_{\syn{f}}}$ is a \emph{stream SOS definition}
of $\syn{f}$.

A {\em stream GSOS (respectively, SOS) definition for $\Sig$}
is a set $\sdefi$ of stream GSOS (respectively, SOS) definitions $\tup{o_{\syn{f}},d_{\syn{f}}}$,
one for each $\syn{f} \in \Sigma$.
\end{defi}

Note that in the above definition,
each pair $\tup{o_{\syn{f}},d_{\syn{f}}}$ corresponds to a stream differential equation:
\begin{equation}\label{eq:gsos-sde}
\begin{array}{lcl}
\syn{f}(\sig_1, \ldots, \sig_k)(0) &=& o_{\syn{f}}(\sig(0),\ldots, \sig_k(0))\\
\syn{f}(\sig_1, \ldots, \sig_k)'   &=& d_{\syn{f}}(\sig(0),\ldots, \sig_k(0))[\sig_i/x_i,\sig_i'/y_i]_{i\leq k}
\end{array}
\end{equation}

\begin{exa}[GSOS definition of arithmetic operations]
\label{ex:arith-gsos-def}

Let $A = \bbR$. 
The SDEs defining addition and convolution product on streams of real numbers in 
\eqref{eq:sde-plus} and \eqref{eq:sde-convprod} are equivalent to the following
stream GSOS definition.
Take as signature $\Sigma_\mathsf{ar}=\{\syn{+},\syn{\times}\} \cup \{ \syn{[a]} \mid a \in \bbR \}$, 
where $\syn{+}$ and $\syn{\times}$ are binary operation symbols 
and $\syn{[a]}$ is a constant symbol, for all $a \in \bbR$.
(We use the underline to indicate the difference between an operation symbol and its interpretation.)
Let 
$\tup{o_{\syn{[a]}},d_{\syn{[a]}}}$, $\tup{o_{\syn{+}},d_{\syn{+}}}$, $\tup{o_{\syn{\times}},d_{\syn{\times}}}$ 
be defined as follows:
\[
\begin{array}{lclclcl}
o_{\syn{[a]}} &=& a & \; & 
d_{\syn{[a]}} &=& \syn{[0]} \qquad \text{ for all } a \in \bbR,
\\[1em]
o_{\syn{+}}(a,b) &=& a+b, & \; &
d_{\syn{+}}(a,b) &=& y_1 \;{\syn{+}}\; y_2, 
\\[.8em]
o_{\syn{\times}}(a,b) &=& a \cdot b, & \; &
d_{\syn{\times}}(a,b) &=& 
(y_1 \;{\syn{\times}}\; x_2) \;{\syn{+}}\; (\syn{[a]}\;\syn{\times}\; y_2) 
\end{array}
\]
where $+$ and $\cdot$ on the right-hand sides of $o$-definitions 
denote addition and multiplication of real numbers.
Note that, in fact,
$\tup{o_{\syn{[a]}},d_{\syn{[a]}}}$ and $\tup{o_{\syn{+}},d_{\syn{+}}}$
are stream SOS definitions whereas
$\tup{o_{\syn{\times}},d_{\syn{\times}}}$ is a stream GSOS definition,
since it uses $x_2$ in $d_{\syn{\times}}(a,b)$.
\end{exa}

A \emph{solution} of a stream GSOS definition $\sdefi$ for $\Sig$
is a $\Sig$-algebra $\tup{\strms{A},\alpha}$ on the set of streams
which respects  $\sdefi$, that is, for all $\syn{f}\in\Sig$, $k\in \bbN$,
\begin{equation}\label{eq:D-solution}
\begin{array}{lcl}
f_\alpha(\sig_1, \ldots, \sig_k)(0) &=& o_{\syn{f}}(\sig_1(0),\ldots, \sig_k(0))\\
f_\alpha(\sig_1, \ldots, \sig_k)' &=& \interp{\alpha}(d_{\syn{f}}(\sig_1(0),\ldots, \sig_k(0))[\sig_i/x_i,\sig'_i/y_i]_{i\leq k})
\end{array}
\end{equation}
This definition, 
in fact, says that $\alpha$ is a solution if
the induced interpretation $\interp{\alpha}$
is a homomorphism not only of algebras,
but also of stream automata. We will make this precise below.

We will now prove that every stream GSOS definition $\sdefi$
has a unique solution.
Using the correspondence between $\Sig$-algebras on $\strms{A}$
and interpretations $\term_\Sig(\strms{A})\to \strms{A}$,
we obtain a candidate solution by coinduction
by observing that a stream GSOS definition $\sdefi$
yields a stream automaton structure on $\term_\Sig(\strms{A})$.

\begin{defi}[Syntactic stream automaton]\label{defi:sdefi-automaton}
Let $\sdefi$ be a stream GSOS definition for a signature $\Sig$.
The \emph{syntactic stream automaton for $\sdefi$}
is the map
$\tup{\osdefi, \dsdefi}\colon \term_\Sig(\strms{A}) \to  A \times\term_\Sig(\strms{A})$ defined inductively as follows:
For all $\sig \in \strms{A}$,
\[ \osdefi(\sig) =\sig(0), \qquad
\dsdefi(\sig) =\sig'
\]
and for all $k \in \bbN$, $\syn{f} \in \Sig_k$,  and $t_1,\ldots,t_k \in \term_\Sig(\strms{A})$,
\[\begin{array}{lcllcl}
\osdefi(\syn{f}(t_1, \ldots, t_k)) &=& o_{\syn{f}}(\osdefi(t_1), \ldots, \osdefi(t_k))\\[.4em]
\dsdefi(\syn{f}(t_1, \ldots, t_k)) &=& d_{\syn{f}}(\osdefi(t_1), \ldots, \osdefi(t_k))[t_i/x_i,\dsdefi(t_i)/y_i]_{i \leq k}
\end{array}
\]
The final homomorphism of stream automata from
$\tup{\term_\Sig(\strms{A}),\tup{\osdefi,\dsdefi}}$ is denoted by
$\sem{-}_{\sdefi}$, i.e.,
\begin{equation}\label{eq:syntactic-solution}
\xymatrix@C=6em{
\term_\Sig(\strms{A}) \ar[r]^-{\sem{-}_{\sdefi}} \ar[d]_-{\tup{\osdefi,\dsdefi}} & \strms{A} \ar[d]^-{\zeta}\\
A\times \term_\Sig(\strms{A}) \ar[r]^-{\id_A \times \sem{-}_{\sdefi}}  &
A \times \strms{A}
}\end{equation}
and we let $\alpha_\sdefi$ be the $\Sig$-algebra on $\strms{A}$
obtained by restricting $\sem{-}_{\sdefi}$ to terms of depth 1.
That is,
$\alpha_\sdefi = \{f_\sdefi \colon (\strms{A})^k \to \strms{A} \mid \syn{f} \in \Sig_k, k\in \bbN\}$ where
\begin{equation}\label{eq:alpha-D}
f_\sdefi(\sig_1,\ldots,\sig_k) = \sem{\syn{f}(\sig_1,\ldots,\sig_k)}_\sdefi
\end{equation}
\end{defi}

\begin{exa}
\label{ex:syn-aut}
Let $\sdefi$ be the stream GSOS definition
from Example~\ref{ex:arith-gsos-def}. We briefly describe some of the transitions in the syntactic stream automaton of $\sdefi$.
We use again the notation introduced 
in Subsection~\ref{ssec:stream-aut} by writing
$x \sgoes{a} y$ when $\osdefi(x)=a$ and $\dsdefi(x)=y$.
Let 
$\sigma = (2,0,0,\ldots)$,
$\tau = (1,1,1,\ldots)$,
$\delta = (1,0,0,\ldots)$ and
$\rho = (0,0,0,\ldots)$.
Then here are two examples of states and transitions:
\[\begin{array}{rcc}
\sig \times (\tau + \delta) & \sgoes{4} &
(\sig' \times (\tau + \delta)) + ([2] \times(\tau' + \delta') )
\\
 & & = 
\\
&& 
(\rho\times (\tau + \delta)) + ([2] \times (\tau + \rho))
\\[.5em]
[5]\times\sig & \sgoes{10} & ([0]\times\sigma)+ ([5]\times[0])\qquad\qquad\qquad
\end{array}  
\]
\end{exa}

The definition of the syntactic stream automaton ensures that the
following fundamental result holds.

\begin{lem}[Bisimilarity is a congruence]\label{lem:bisim-is-congruence}
On the syntactic stream automaton given by
$\tup{\TSig(\strms{A}),\tup{\osdefi,\dsdefi}}$,
bisimilarity is a congruence, that is,
for all terms $g \in \TSig(Z)$ over some set of variables
$Z = \{z_1,\ldots,z_n\}$,
and all terms $s_1,\ldots,s_n, u_1,\ldots,u_n \in \TSig(\strms{A})$,
\[
\forall j=1,\ldots, n: s_j \sim u_j \quad\Ra\quad
 g[s_j/z_j]_{j \leq n} \sim g[u_j/z_j]_{j \leq n}
\]
\end{lem}
\proof
We define relations $\{R_m\}_{m\in \bbN}$ on $\TSig(\strms{A})$
inductively by
$R_0 := \;\sim$ (the bisimilarity relation on $\TSig(\strms{A})$)
and for $m \geq 1$, $R_{m+1}$ is defined by the following
congruence rule:
\begin{equation}\label{eq:cong-relation}
\begin{array}{c}
s_1 \,R_m\, u_1 \quad\cdots\quad s_n \,R_m\, u_n
\\
\hline
 g[s_j/z_j]_{j \leq n} \,R_{m+1}\, g[u_j/z_j]_{j \leq n}
\end{array}
(g \in \TSig(Z))
\end{equation}
where $Z = \{z_1,\ldots,z_n\}$.
Note that $R_m \sse R_{m'}$ for all $m \leq m'$.
Let $R = \bigcup_{m\in\bbN}R_m$.
We show that $R$ is a bisimulation.
More precisely, we show by induction on $m$ that
\begin{equation}\label{eq:main-induction}
\forall t, v \in \TSig(\strms{A}):
t \,R_m\, v \;\Ra\;
\left[o_\sdefi(t) = o_\sdefi(v) \text{ and }
d_\sdefi(t) \,R\, d_\sdefi(v) \right]
\end{equation}
For convenience,
we use the shorthand notation
$t[s] := t[s_j/z_j]_{j \leq n}$
and
$t[u] := t[u_j/z_j]_{j \leq n}$
for any term $t \in \TSig(Z)$.

The base case ($m=0$) is immediate since $R_0 = \,\sim$ and $\sim\,\sse R$.
For the induction step ($m+1$),
suppose 
$s_1 \,R_m\, u_1 \, , \, \ldots \, , \, s_n \,R_m\, u_n$
and $g \in \TSig(Z)$.
We show by subinduction on the term structure of $g$
that
\begin{equation}\label{eq:subinduction-claim}
o_\sdefi(g[s]) = o_\sdefi(g[u]) \quad\text{ and }\quad
d_\sdefi(g[s]) \,R\, d_\sdefi(g[u])
\end{equation}
For $g=z_j \in Z$,
it follows that $\tup{g[s],g[u]} = \tup{s_j,u_j} \in R_m$
and hence  \eqref{eq:subinduction-claim}
holds by the main induction hypothesis (for $m$).

For $g = \syn{f}(t_1, \ldots, t_k)$,
by subinduction hypothesis, we have for all $i=1,\ldots,k$:
\[ \begin{array}{c}
o_\sdefi(t_i[s]) = o_\sdefi(t_i[u])
\quad \text{ and } \quad
d_\sdefi(t_i[s]) \;R\; d_\sdefi(t_i[u]).
\end{array}
\]
We now check the subinduction claim \eqref{eq:subinduction-claim} for $g$. \\
\textit{Outputs are equal:}
\[\begin{array}{rcll}
 & & o_\sdefi(\syn{f}(t_1,\ldots,t_k)[s])
\\
 &=& o_{\syn{f}}(o_\sdefi(t_1[s]), \ldots ,o_\sdefi(t_k[s]))
 & (\text{def.\ } o_\sdefi)
\\
 &=&
 o_{\syn{f}}(o_\sdefi(t_1[u]), \ldots ,o_\sdefi(t_k[u]))
 & (\text{sub-I.H.})
\\
 &=&
 o_\sdefi(\syn{f}(t_1,\ldots,t_k)[u])
 & (\text{def.\ } o_\sdefi)
\end{array}\]
\textit{Next states are related:}
First, for notational convenience, let
$w$ denote the term that specifies the next state for $\syn{f}$, i.e.,
\[\begin{array}{rcl}
 w &:=&
  d_{\syn{f}}(o_\sdefi(t_1[s]), \ldots ,o_\sdefi(t_k[s]))
\\
 &\stackrel{\text{sub-I.H.}}{=}& d_{\syn{f}}(o_\sdefi(t_1[u]), \ldots ,o_\sdefi(t_k[u]))
\end{array}\]
From the definition of $R$  it follows that
\[t_i[s] \,R_{m+1}\,t_i[u] \qquad
\text{ for all $i = 1,\ldots,k$}
\]
and from the sub-induction hypothesis, it follows that
\[d_\sdefi(t_i[s]) \,R\, d_\sdefi(t_i[u]) \qquad
\text{ for all $i = 1,\ldots,k$}.
\]
hence there is some $M \in \bbN$ such that
\[
t_i[s] \,R_{M}\,t_i[u], \quad
d_\sdefi(t_i[s]) \,R_{M}\, d_\sdefi(t_i[u]) \qquad
\text{ for all $i = 1,\ldots,k$}.
\]
By the definition of $R$, we then have
\begin{equation}\label{eq:IH-help}
\begin{array}{c}
w[t_i[s]/x_i,d_\sdefi(t_i[s])/y_i]_{i \leq k}
\;R_{M+1}\;
w[t_i[u]/x_i, d_\sdefi(t_i[u])/y_i]_{i \leq k}
\end{array}
\end{equation}
and hence
\[\begin{array}{rcll}
 & & d_\sdefi(\syn{f}(t_1,\ldots,t_k)[s])
\\
 &=& w[t_i[s]/x_i,d_\sdefi(t_i[s])/y_i]_{i \leq k}
 & (\text{def.\ } d_\sdefi)
\\
 & R &
 w[t_i[u]]/x_i, d_\sdefi(t_i[u]])/y_i]_{i \leq k}
 & \text{(by \eqref{eq:IH-help})}
\\
 &=&
 d_\sdefi(\syn{f}(t_1,\ldots,t_k)[u])
 & (\text{def.\ } d_\sdefi).
\end{array}\]
This concludes the subinduction on $g$, and hence also the main induction
for $m$.
\qed

The map $\sem{-}_{\sdefi}$ is by definition a stream homomorphism.
We now show that it is also an algebra homomorphism.

\begin{lem}[$\sem{-}_\sdefi$ is algebra homomorphism]
\label{lem:final-map-is-algebra-hom}
Let $\sdefi$ be a stream GSOS definition for a signature $\Sig$,
and
$\alpha_\sdefi$
be the $\Sig$-algebra on $\strms{A}$ defined in \eqref{eq:alpha-D} of Definition~\ref{defi:sdefi-automaton}.
The term interpretation $\interp{\alpha_\sdefi}$ induced
by $\alpha_\sdefi$ is precisely $\sem{-}_\sdefi$. Consequently, $\sem{-}_\sdefi$ is a morphism of
$\Sig$-algebras.
\end{lem}
\proof
Let $\alpha_\sdefi = \{f_\sdefi \colon (\strms{A})^k \to \strms{A} \mid \syn{f} \in \Sig_k, k\in \bbN\}$ 
be defined as in \eqref{eq:alpha-D}.
We show by induction on the term structure that
for all  $t \in \TSig(\strms{A})$:
\begin{equation}\label{eq:final-map-is-algebra-hom}
\interp{\alpha_\sdefi}(t) =  \sem{t}_\sdefi
\end{equation}
For $t = \sig \in \strms{A}$, we clearly have that
$\interp{\alpha_\sdefi}(\sig) = \sig = \sem{\sigma}_\sdefi$.
For $k \in \bbN$, $\syn{f} \in \Sig_k$,  and
$t_1, \ldots ,t_k \in \TSig(\strms{A})$,
we have
\[\begin{array}{rcll}
\interp{\alpha_\sdefi}(\syn{f}(t_1,\ldots,t_k)) &=&
  f_\sdefi(\interp{\alpha_\sdefi}(t_1), \ldots,\interp{\alpha_\sdefi}(t_k))
 & (\text{def. }\interp{\alpha_\sdefi})
\\
 &=& f_\sdefi(\sem{t_1}_\sdefi, \ldots, \sem{t_k}_\sdefi)
  & (\text{I.H.})
\\
 &=& \sem{\syn{f}(\sem{t_1}_\sdefi, \ldots, \sem{t_k}_\sdefi)}_\sdefi
 &   (\text{def. }{f}_\sdefi)
\\
 &=& \sem{\syn{f}(t_1, \ldots, t_k)}_\sdefi
\end{array}\]
where the last equality holds because $\sem{-}_\sdefi$
identifies bisimilar states, and bisimilarity
of $\syn{f}(\sem{t_1}_\sdefi, \ldots, \sem{t_k}_\sdefi)$
and $\syn{f}(t_1, \ldots, t_k)$ follows from 
Lemma~\ref{lem:bisim-is-congruence} (bisimilarity is a congruence),
and the fact that
for all $t \in \TSig(\strms{A})$,
\begin{equation}\label{eq:sem-is-bis}
 t \sim \sem{t}_\sdefi
\end{equation}
since $\sem{-}_\sdefi$ is a stream homomorphism.
\qed
We now characterise the solutions to $\sdefi$ as being those
maps $\alpha$ whose induced interpretation is a stream homomorphism.

\begin{prop}{}\label{prop:D-solutions}
Let $\sdefi$ be a stream GSOS definition for a signature $\Sig$.
For all $\Sig$-algebras $\tup{\strms{A},\alpha}$, 
$\alpha$ is a solution of $\sdefi$ if and only if
$\interp{\alpha}\colon \term_\Sig(\strms{A}) \to \strms{A}$
is a stream automaton homomorphism
from $\tup{\term_\Sig(\strms{A}),\tup{\osdefi,\dsdefi}}$
to the final stream automaton
$\tup{\strms{A},\zeta}$.
\end{prop}
\proof
Let $\alpha$ be a solution of $\sdefi$.
We show that $\interp{\alpha}$ is a homomorphism of stream automata
by induction on the term structure.
The base case is immediate, since by definition
$\interp{\alpha}(\sig)(0) = \sig(0) =\osdefi(\sig)$ and
$\interp{\alpha}(\sig)' = \sig' = \dsdefi(\sig)$.
For the inductive step,
let  $k\in\bbN$, $\syn{f}\in \Sig_k$. We have
\[\begin{array}{rcll}
\interp{\alpha}(\syn{f}(t_1,\ldots, t_k))(0) &=&
  f_\alpha(\interp{\alpha}(t_1), \ldots,\interp{\alpha}(t_k))(0)
 & \text{(def. $\interp{\alpha}$)}
\\[.2em]
 &=& o_{\syn{f}}(\interp{\alpha}(t_1)(0), \ldots,\interp{\alpha}(t_k)(0))
  & \text{($\alpha$ is solution)}
\\[.2em]
 &=& o_{\syn{f}}(\osdefi(t_1), \ldots, \osdefi(t_k))
  & \text{(by I.H.)}
\\[.2em]
 &=& \osdefi(\syn{f}(t_1,\ldots, t_k))
 & \text{(def. $\osdefi$)}
\end{array}
\]
and
\[\begin{array}{rcll}
&& \interp{\alpha}(\syn{f}(t_1,\ldots, t_k))'
\\[.2em]
  &=&  f_\alpha(\interp{\alpha}(t_1), \ldots,\interp{\alpha}(t_k))'
 & \text{(def. $\interp{\alpha}$)}
\\[.2em]
 &=& \interp{\alpha}(d_{\syn{f}}(\interp{\alpha}(t_1)(0), \ldots,\interp{\alpha}(t_k)(0))[\interp{\alpha}(t_i)/x_i,\interp{\alpha}(t_i)'/y_i]_{i\leq k})
  & \text{($\alpha$ is solution)}
\\[.2em]
 &=& \interp{\alpha}(d_{\syn{f}}(\osdefi(t_1), \ldots, \osdefi(t_k))[\interp{\alpha}(t_i)/x_i,\interp{\alpha}(\dsdefi(t_i))/y_i]_{i\leq k})
  & \text{(I.H.)}
\\[.2em]
 &=& \interp{\alpha}(d_{\syn{f}}(\osdefi(t_1), \ldots, \osdefi(t_k))[t_i/x_i,\dsdefi(t_i)/y_i]_{i\leq k})
  & (*)
\\[.2em]
 &=& \interp{\alpha}(\dsdefi(\syn{f}(t_1,\ldots, t_k)))
 & \text{(def. $\dsdefi$)}
\end{array}\]
where $(*)$ holds since nested applications of $\interp{\alpha}$
are ``flattened'' into one outermost application
which interprets the entire term.

For the converse, assume that $\interp{\alpha}$ is a homomorphism of
stream automata. Then in particular, for all $k\in\bbN$, $\syn{f}\in\Sig_k$, 
and all $\sig_1, \ldots,\sig_k \in \strms{A}$,
\begin{eqnarray}
\interp{\alpha}(\syn{f}(\sig_1,\ldots, \sig_k))(0) &=& \osdefi(\syn{f}(\sig_1,\ldots, \sig_k)) \label{eq:alpha-o}
\\
\interp{\alpha}(\syn{f}(\sig_1,\ldots, \sig_k))' &=& \interp{\alpha}(\dsdefi(\syn{f}(\sig_1,\ldots, \sig_k))) \label{eq:alpha-d}
\end{eqnarray}
It follows that
\[\begin{array}{rcll}
f_\alpha(\sig_1,\ldots,\sig_k)(0) &=&
  \interp{\alpha}(\syn{f}(\sig_1,\ldots, \sig_k))(0)
 & \text{(def. $\interp{\alpha}$)}
\\[.2em]
 &=&  \osdefi(\syn{f}(\sig_1,\ldots, \sig_k))
  & \text{(by \eqref{eq:alpha-o})}
\\[.2em]
 &=& o_{\syn{f}}(\sig_1(0),\ldots,\sig_k(0))
 & \text{(def. $\osdefi$)}
\end{array}
\]
and
\[\begin{array}{rcll}
&& f_\alpha(\sig_1,\ldots,\sig_k)'
\\[.2em]
&=&
  \interp{\alpha}(\syn{f}(\sig_1,\ldots, \sig_k))'
 & \text{(def. $\interp{\alpha}$)}
\\[.2em]
 &=& \interp{\alpha}(\dsdefi(\syn{f}(\sig_1,\ldots, \sig_k)))
 & \text{(by \eqref{eq:alpha-d})}
\\[.2em]
 &=&  \interp{\alpha}(d_{\syn{f}}(\osdefi(\sig_1), \ldots, \osdefi(\sig_k))[\sig_i/x_i,\dsdefi(\sig_i)/y_i]_{i\leq k})
 & \text{(def. $\dsdefi$)}
\\[.2em]
 &=&  \interp{\alpha}(d_{\syn{f}}(\sig_1(0), \ldots, \sig_k(0))[\sig_i/x_i,\sig_i'/y_i]_{i\leq k})
 & \text{(def. $\osdefi$ and $\dsdefi$)}
\end{array}
\]
which proves that  $\alpha$ is indeed a solution of $\sdefi$.
\qed
Finally, we can put everything together.

\begin{thm}{}\label{thm:syntactic-solution}
Let $\sdefi$ be a stream GSOS definition for a signature $\Sig$.
The unique solution of $\sdefi$ is the
$\Sig$-algebra $\tup{\strms{A},\alpha_\sdefi}$
that corresponds to the term interpretation
given by the final stream homomorphism 
$\sem{-}_\sdefi \colon \term_\Sig(\strms{A}) \to \strms{A}$
of the syntactic stream automaton.
\end{thm}

\proof
By Lemma~\ref{lem:final-map-is-algebra-hom},
$\interp{\alpha_\sdefi} = \sem{-}_\sdefi$,
hence by Proposition~\ref{prop:D-solutions},
$\alpha_\sdefi$ is a solution to $\sdefi$.
The uniqueness of $\alpha_\sdefi$ follows from the uniqueness
of $\sem{-}_\sdefi$ and the 1-1 correspondence between
$\Sig$-algebras $\tup{X,\alpha}$ and 
term interpretations $\interp{\alpha}\colon \TSig(\strms{A})\to\strms{A}$.
\qed

\begin{exa}
\label{ex:solution-arith}
Consider the final map $\sem{-} = \sem{-}_\sdefi$ for the 
GSOS definition $\sdefi$ of the arithmetic operations
from Example~\ref{ex:arith-gsos-def} 
(taking again $A = \bbR$, and
$\sigma = (2,0,0,\ldots)$,
$\tau = (1,1,1,\ldots)$,
$\delta = (1,0,0,\ldots)$,
$\rho = (0,0,0,\ldots)$).
We find that 
\[\begin{array}{rcl}
\sem{\sig \times (\tau + \delta)} &=& 
\sem{\sig} \times (\sem{\tau} + \sem{\delta})
\\[.4em]
 &=& (4,2,2,2,\ldots)
\\[.8em]
\sem{(\rho\times (\tau + \delta)) + ([2] \times (\tau + \rho))} &=&
(\sem{\rho}\times (\sem{\tau} + \sem{\delta})) + (\sem{[2]} \times (\sem{\tau} + \sem{\rho}))
\\[.4em]
 & = & (2,2,2,\ldots)
\end{array}  
\]
which confirms that $\sem{-}$ respects the transition from $\sig \times (\tau + \delta)$.
Similarly, we find that the following transition in the syntactic automaton
\[\begin{array}{rcl}
[5]\times\sig & \sgoes{10} & ([0]\times\sigma)+ ([5]\times[0])
\end{array}  
\]
is mapped by $\sem{-}$ to the following transition in $\tup{\bbR^\omega,\zeta}$
\[ (10,0,0,0,\ldots) \sgoes{10} (0,0,0,\ldots).
\]
\end{exa}


\subsection{Causal stream operations}
\label{ssec:causal-ops}

Next we will show
that stream GSOS definitions
exactly define the so-called \emph{causal} stream operations, that is,
operations such that for all $n \in \bbN$, the $n$-th value of the result stream
depends only on the first $n$ values of the argument stream(s).
For a formal definition,
we use the following notation.
For $\sig, \tau \in \strms{A}$ and $n \in \bbN$,
we write $\sig \equiv_n \tau$ if for all $j \leq n$, $\sig(j) = \tau(j)$.
A $k$-ary stream operation
$f \colon (\strms{A})^k \to \strms{A}$ is \emph{causal}
if for all $\sig_i, \tau_i \in \strms{A}, i=1,\ldots,k$,
\[
\forall i \leq k: \sig_i \equiv_n \tau_i
\qquad \Ra  \qquad
f(\sig_1,\ldots,\sig_k) \equiv_n f(\tau_1,\ldots,\tau_k)
\]
Let $\Gamma_k$ denote the set of all causal $k$-ary stream operations
$f \colon (\strms{A})^k \to \strms{A}$.
The elements of $\Gamma_k$
are exactly the behaviours of \emph{($k$-ary) Mealy machines}
which are maps of type $m \colon X \to (A \times X)^{A^k}$.
Mealy machines and causal stream functions are treated in detail in
\cite{HR:mealy-SACS,Rut05:FACS-short}. We give a brief recap here.
For all $f \in \Gamma_k$ and
all $(a_1, \ldots, a_k) \in A^k$,
we define the notion of \emph{Mealy output} $f[(a_1, \ldots, a_k)]$ and
\emph{Mealy derivative} $f_{(a_1, \ldots, a_k)}$ of $f$
as follows. For all $\sig_1,\ldots, \sig_k \in \strms{A}$,

\begin{equation}\label{eq:Mealy-derivatives}
\begin{array}{rcl}
f[(a_1, \ldots, a_k)] &=&
  f(a_1\!:\!\sig_1, \ldots, a_k\!:\!\sig_k)(0)
\\
f_{(a_1, \ldots, a_k)}(\sig_1,\ldots,\sig_k) &=&
  f(a_1\!:\!\sig_1, \ldots, a_k\!:\!\sig_k)'
\end{array}
\end{equation}
Note that since $f$ is causal, it follows that $f_{(a_1, \ldots, a_k)} \in \Gamma_k$
and that $f[(a_1, \ldots, a_k)]$ is well-defined, as it does not depend on
$\sig_1,\ldots, \sig_k$.
We define a Mealy machine structure $\gamma\colon \Gamma_k \to (A\times\Gamma^k)^{A^k}$ by
\begin{equation}\label{eq:final-Mealy-structure}
\gamma(f)({a_1,\ldots,a_k}) = \tup{f[(a_1, \ldots, a_k)],f_{(a_1, \ldots, a_k)}},
\end{equation}
In fact, $\tup{\Gamma_k,\gamma}$ is a final Mealy machine,
cf.~\cite{HR:mealy-SACS,Rut05:FACS-short}.

\begin{prop}\label{prop:gsos-implies-causal}
If 
$f \colon (\strms{A})^k \to \strms{A}$
is stream GSOS definable, then $f$ is causal.
\end{prop}
\proof
Suppose that $f$ is stream GSOS definable, that is, $f$ is one of the
operations in the solution $\alpha_\sdefi$ for some stream GSOS definition
$\sdefi$.
The proof follows essentially from the fact that for all $n \in \bbN$,
$\equiv_n$ is a congruence,
that is, for all $t \in \TSig(Z)$, $Z = \{z_1, \ldots, z_l\}$,
and all $\sig_i,\tau_i \in \strms{A}$, $i=1,\ldots,l$:
\begin{equation}\label{eq:n-equiv-congruence}
\text{for all } i \leq l: \sig_i \equiv_n \tau_i
\qquad\Ra\qquad
\sem{t[\sig_i/z_i]_{i\leq l}}_\sdefi \equiv_n \sem{t[\tau_i/z_i]_{i \leq l}}_\sdefi
\end{equation}
which can be shown by double induction on $n$ and the structure of $t$.
We refer to \cite{KNR:SDE} for details.
\qed

Conversely, any causal stream operation can be defined by a
(potentially very large) stream definition.

\begin{prop}
If 
$f \colon (\strms{A})^k \to \strms{A}$ is causal,
then $f$ is stream GSOS definable.
\end{prop}
\proof
We define a stream definition $\sdefi_\causal$ for the signature
$\Sig_\causal = \{ \syn{f} \mid f \colon (\strms{A})^k \to \strms{A}
                                                 \text{ causal},$ $ k \in \bbN \}$,
by including,
for each $k$-ary function symbol $\syn{f} \in \Sig_\causal$, the equation
\begin{equation}\label{eq:sdefi-causal}
\begin{array}{rcll}
o_{\syn{f}}(a_1, \ldots, a_k) &=& f[(a_1, \ldots, a_k)] & \in A\\
d_{\syn{f}}(a_1, \ldots, a_k) &=& \syn{f_{(a_1, \ldots, a_k)}}(y_1, \ldots, y_k)
  & \in \TSig(\{y_1, \ldots,y_k\})\\
\end{array}
\end{equation}
Let $\alpha$ be the $\Sig_\causal$-algebra on $\strms{A}$ in which each
symbol $\syn{f}$ is interpreted as $f$.
We show that $\alpha$ is a solution to $\sdefi_\causal$.
For $f \in \Gamma_k$, we have
\[\begin{array}{lclcl}
f(\sig_1,\ldots,\sig_k)(0) &=& f[(\sig_1(0),\ldots, \sig_k(0))]
\\
&=&
  o_{\syn{f}}(\sig_1(0),\ldots, \sig_k(0))
\\[.5em]
f(\sig_1,\ldots,\sig_k)' &=&
  {f_{(\sig_1(0), \ldots, \sig_k(0))}}(\sig_1', \ldots, \sig_k')
\\ &=&
  \interp{\alpha}(\syn{f_{(\sig_1(0), \ldots, \sig_k(0))}}(y_1, \ldots, y_k)[\sig'_i/y_i]_{i\leq k})
\end{array}\]
which shows that $\alpha$ is a solution to $\sdefi_\causal$.
\qed

\begin{rema}
Note that in \eqref{eq:sdefi-causal}
the derivative term $d_{\syn{f}}(a_1,\ldots,a_k)$
only uses the $y_i$-variables,
i.e.\ the derivatives of the arguments,
and not the arguments themselves (i.e.\ the $x_i$-variables).
This means that all causal stream operations are definable
by a (possibly infinite) SOS specification.
\end{rema}

\begin{thm}
Let $f \colon (\strms{A})^k \to \strms{A}$ be a stream operation.
We have: $f$ is causal if and only if $f$ is stream GSOS definable.
\end{thm}


\subsection{Causality and productivity}
\label{ssec:causality-productivity}

Every stream GSOS defined operation is \emph{productive}, meaning
that by successively computing output and derivative using the
SDEs we can construct the entire stream in the limit.

A well known example of a stream operation that is not causal
is the operation
\[\begin{array}{lcl}
\even(\sig) &=&  (\sig(0), \sig(2),\sig(4),\ldots)\\
\end{array}\]
which we encountered
already in Section~\ref{sec:non-standard} (cf.~equation \eqref{eq:def-even}).
The operation $\even\colon\strms{A}\to\strms{A}$
can be defined by the following SDE:
\[\begin{array}{lcllcl}
\even(\sig)(0) &=& \sig(0), \qquad \even(\sig)' &=&  \even(\sig'')\\ 
\end{array}\]
If $\sig$ is given by a productive definition, then
also $\even(\sig)$ 
is productive.
However, it is easy to give a SDE using $\even$ 
which is not productive:
\begin{equation}\label{eq:sde-unproductive}
\begin{array}{c}
\sig(0) = 0, \qquad \sig' = \even(\sig)
\end{array}
\end{equation}
One sees the problem when we try to compute initial value and derivatives.
The first two steps are fine:
\[\begin{array}{lll}
1. & \sig(0) = 0, \qquad & \sig' = \even(\sig)\\
2. & \even(\sig)(0) = \sig(0) = 0, \qquad & \even(\sig)' = \even(\sig'') 
\end{array}
\]
But when we try to compute the initial value of $\even(\sig'')$, we get:
\[
\even(\sig'')(0) = \sig''(0) = \; ?
\]
which does not yield a value.
The SDE \eqref{eq:sde-unproductive} has several solutions,
e.g. $\sig=[0] = (0,0,0,\ldots)$ or
$\sig=0:0:\ones = (0,0,1,1,1,\ldots)$, but it does not have a unique one.

Productivity of stream definitions in a term rewriting context
have been closely studied in
\cite{EH:lazy-productivity,EGHIK:productivity-TCS}.

\subsection{Simple/linear/context-free stream specifications revisited}

In conclusion of this section, we will demonstrate how the syntactic method can be applied to prove the existence of unique solutions to the simple, linear and context-free specifications
from Sections \ref{sec:simple-specs}-\ref{sec:context-free-specs}.

\subsubsection*{Simple Specifications}\enlargethispage{\baselineskip}

A simple equation system (i.e.\ a stream automaton) $\tup{X,\es}$
can be seen as a stream definition over the signature $\Sig$ which contains
a constant for each $x \in X$, and no further operation symbols.
Note that since a stream definition consists of one equation for each operation symbol, we must treat the elements of $X$ as constants (rather than variables)) in order to view $\tup{X,\es}$ as a stream definition.
Hence $\TSig(\strms{A}) = X$, and it follows that the syntactic solution
from Theorem~\ref{thm:syntactic-solution} coincides with the direct solution
by coinduction.\enlargethispage{\baselineskip}

\subsubsection*{Linear Specifications}

A linear equation system over $X$ can be viewed as a stream definition
for a signature which contains a constant for each $x \in X$,
and operation symbols for scalar multiplication and sum, as we explain now.
Consider the \emph{linear signature $\Sig$} which contains
a unary scalar multiplication operation for each $a \in A$
and a binary sum operation.
The set of $\Sig$-terms over a set $X$ is generated by the
following grammar:
\begin{equation}\label{eq:linear-grammar}
t ::= x \in X \mid a \cdot t \mid t + t, \qquad a \in A.
\end{equation}
A linear equation system over a set $X$ can now be seen as a map
$\tup{o,d}\colon X \to A \times \TSig(X)$.
In order to get a stream definition,
we can again view elements of $X$ as constants and consider the
larger signature $\bar{\Sig} = \Sig \cup X$.
So in particular, $X \sse \term_{\bar\Sig}(Y)$ for any set $Y$.
By putting together the equations from $\tup{o,d}$ and the SDEs that define
scalar multiplication and sum, we obtain a big
stream definition $\sdefi$ for $\bar{\Sig}$.
From the syntactic method (Theorem~\ref{thm:syntactic-solution}),
we then obtain a map $X \to \strms{A}$ via inclusion and the term interpretation
$X \injR \term_{\bar{\Sig}}(\strms{A}) \sgoes{\interp{\alpha}} \strms{A}$.
We repeat here the relevant diagram for convenience:
\begin{equation}\label{eq:synt-method}
\xymatrix@C=6em{
  \term_\Sig(\strms{A}) \ar[r]^-{\interp{\alpha}} \ar[d]_-{\tup{\osdefi,\dsdefi}}
  & \strms{A} \ar[d]^-{\zeta}\\
A\times \term_\Sig(\strms{A}) \ar[r]^-{\id_A \times \interp{\alpha}}  &
A \times \strms{A}
}\end{equation}
This map $X \to \strms{A}$ preserves the equations in $\tup{o,d}$,
since $\interp{\alpha}$ is a homomorphism
of both $\bar{\Sig}$-algebras and stream automata,
hence it is a solution to $\tup{o,d}$,
and by uniqueness of solutions it must coincide with the solution obtained
as the composition $X \sgoes{\eta_X} \V(X) \sgoes{g} \strms{A}$ in
\eqref{eq:linear-solutions} on page~\pageref{eq:linear-solutions} in Section~\ref{sec:linear-specs}.

A more detailed argument of why the syntactic method yields a solution
in the sense of Section~\ref{sec:linear-specs} goes as follows.
We prove that the two methods lead to the same solution map
$X \to \strms{A}$ by showing that the following relation on streams
\begin{equation}\label{eq:bis-up-to-rel-linear}
 R = \{\tup{\interp{\alpha}(x),g(\eta_X(x))} \mid x \in X\} \;\sse\; \strms{A}\times\strms{A}
\end{equation}
is a bisimulation-up-to scalar multiplication and sum,
cf.~Theorem~\ref{thm:up-to}.
To this end, let $x \in X$ be arbitrary, and suppose that
$d(x) = a_1x_1 + \ldots + a_kx_k$.\\
\noindent\emph{Initial value:}
\[
 \interp\alpha(x)(0) = \osdefi(x) = o(x) = o^\sharp(\eta_X(x)) = g(\eta_X(x))(0).
\]
\noindent\emph{Derivative:} We have
\[\begin{array}{r}
\interp\alpha(x)' \stackrel{\eqref{eq:synt-method}}{=}
\interp\alpha(d(x)) =
\interp\alpha(a_1x_1 + \ldots a_kx_k) =
a_1\interp\alpha(x_1) + \ldots + a_k\interp\alpha(x_k)
\\[.7em]
g(\eta_X(x))'  \stackrel{\eqref{eq:linear-solutions}}{=}
g(d(x)) = g(a_1x_1 + \ldots a_kx_k) =
a_1g(x_1) + \ldots + a_kg(x_k)
\end{array}\]
where the last equalities in each line follow from $\interp\alpha$
being a $\bar{\Sig}$-algebra homomorphism, and $g$ being linear, respectively.
We have now shown that $R$ is a
bisimulation-up-to scalar multiplication and sum.
It follows that for all $x \in X$, $\interp{\alpha}(x)$ and $g(\eta_X(x))$
are bisimilar, and hence by coinduction they are equal.

The equivalence between the two solution methods also follows from a
more general result in \cite{BHKR:pres-DL} which relates specifications
that use pure syntax (such as $\TSig(X)$) and specifica\-tions that use
an algebraic structure viewed as syntax modulo axioms
(such as $\V(X)$ viewed as $\TSig(X)$ modulo vector space axioms).
We describe this is more detail in section~\ref{ssec:non-free-monads}.

\subsubsection*{Context-free Specifications}

As in the linear case 
we obtain unique solutions to context-free equation systems
by combining the equations with the SDEs that define the operations
used on the right-hand side of the equations.
In this case, we consider the \emph{polynomial signature} $\Sig$,
which contains a stream constant for each $a\in A$,
and binary symbols $+$ and $\times$.
The set $\TSig(X)$ of all $\Sig$-terms over a set $X$ is generated by
the following grammar:
\begin{equation}\label{eq:polynomial-grammar}
t ::= x \in X \mid a \in A \mid t+t \mid t \times t
\end{equation}
A context-free equation system over $X$ can now be seen as a map
$\tup{o,d} \colon X \to A \times \TSig(X)$.
Putting the equations from $\tup{o,d}$
together with the SDEs defining the polynomial $\Sig$-operations
$a \in A, +, \times$ we obtain
one big stream definition $\sdefi$ for the extended signature
$\bar{\Sig} = \Sig \cup X$ where elements from $X$ are viewed as constants.
From the syntactic method (Theorem~\ref{thm:syntactic-solution}),
we obtain a solution map
$X \injR T_{\bar{\Sig}}(\strms{A}) \sgoes{\interp\alpha} \strms{A}$.

As in the linear case,
one can show that this solution coincides with the solutions obtained
via stream automata (cf.~Section~\ref{sec:context-free-solutions})
using bisimulation-up-to polynomial operations.

\begin{rema}
	Unique solutions of simple, linear and context-free equation systems for the
	non-standard tail operations $\partial \in \{ \Delta, \frac{d}{dX} ,\Delta_o \}$
	can be obtained via the syntactic method in essentially the same way as
	the method only relies on the finality of $\zeta:\strms{A} \to A \times \strms{A}$.	
\end{rema}




\section{A General Perspective}\enlargethispage{\baselineskip}
\label{sec:stream-gsos}

In this section, we describe how the stream GSOS definitions of the
prevous section relate to the categorical framework
known as abstract GSOS.
\emph{Abstract GSOS} was developed in \cite{TuriPlotkin:LICS-GSOS} as a
general framework in structural operational semantics~\cite{Aceto:SOS-HB}
for studying rule formats that guarantee a compositional semantics.
The more recent survey paper \cite{Klin11} gives an excellent introduction
to abstract GSOS, and includes many examples on streams.
We present here a brief account
of the categorical underpinnings of stream GSOS, and relate
the general constructions to the concrete ones we have seen in
earlier sections.
The material in this section is based mainly on
\cite{Bartels:PhD,Klin11,LPW2004:cat-sos}.

For this section, we assume some familiarity with
basic categorical notions such as functor and natural transformation,
cf.~e.g.~\cite{MacLane}.
Throughout, let $\Set$ be the category of sets and functions.

The generality of abstract GSOS is obtained by generalising stream automata
to $F$-coalgebras, and
observing that a GSOS definition (for $F$-coalgebras)
corresponds to a so-called distributive law which links algebraic structure
with coalgebraic behaviour.

\subsection{Coalgebras for a functor}

In previous sections, we focused on stream automata
which are maps of the type $X \to A\times X$.
We will now look at them from a more abstract point of view,
namely as coalgebras \cite{Rut:univ-coalg,JR:tut-2011}.
Coalgebra is a framework for studying
state-based systems in a uniform setting.
This is achieved by describing the system type by
a functor $F$ which defines the kind of transitions and observations
the system can make.
By varying $F$ we obtain many known structures
such as $A$-labelled binary trees ($FX = X \times A\times X$),
deterministic automata ($FX = 2\times X^A$), and
labelled transition systems ($FX = \Pow(X)^A$), to mention just a few.
The advantage of viewing systems as $F$-coalgebras is that
we obtain generic definitions of morphisms and bisimulation,
and we can often prove results uniformly for many system types.

The general definition is as follows.
Given a functor $F\colon \Set \to \Set$, 
an \emph{$F$-coalgebra} is a pair $\tup{X,c}$ where $X$ is a set and
$c \colon X \to FX$ is a function. An \emph{$F$-coalgebra morphism}
from $\tup{X,c}$ to $\tup{Y,d}$ is a map $f\colon X \to Y$
such that $d\circ f = Tf \circ c$.
An $F$-coalgebra $\tup{Z,\zeta}$ is \emph{final} if
for any $F$-coalgebra $\tup{X,c}$ there is a unique
$F$-coalgebra morphism $h\colon \tup{X,c} \to \tup{Z,\zeta}$.
An $F$-coalgebra bisimulation between $\tup{X,c}$ and $\tup{Y,d}$ 
is a relation $R \sse X \times Y$ which carries itself an $F$-coalgebra
structure $r\colon R \to FR$ such that the projections $R \to X$ and $R \to Y$
are $F$-coalgebra morphisms.
It is straightforward to check that stream automata are
coalgebras for the functor $F = A \times (-)$
which maps a set $X$ to $A \times X$
and a function $f\colon X \to Y$ to $\id_A \times f$.
In particular, 
$A \times (-)$-coalgebra morphisms and
$A \times (-)$-coalgebra bisimulations are 
stream homomorphisms and stream bisimulations, respectively, and
the final $A\times(-)$-coalgebra is indeed the final stream automaton
described in Section~\ref{ssec:stream-basics}.

\subsection{Algebras for a monad}\label{ssec:algebras-for-monad}

Where coalgebra gives us an abstract view on systems and behaviour,
algebras for a monad give us an abstract view on algebraic theories,
and compositionality.

We start by explaining how the usual notion of an algebra for a signature
(described in Section~\ref{ssec:terms-algebras}) can be
understood categorically.
An algebra for a signature $\Sig$ of operations $\syn{f}_i, i \in I,$
with arities $k_i$, $i\in I$, is a map
$\coprod_{i\in I}X^{k_i} \to X$ where $X$ is the carrier and $\coprod$
denotes coproduct (or disjoint union).
For example, if $\Sig$ contains a constant $\syn{c}$, a unary $\syn{f}$ and a
binary $\syn{g}$, then an algebra for $\Sig$ with carrier $X$ is a map
$[c,f,g]\colon 1 + X + (X \times X) \to X$
given by case distinction with components $c \colon 1 \to X$,
$f\colon X \to X$ and $g \colon X \times X \to X$.
A signature $\Sig$ corresponds in this way to a $\Set$-functor (which we also denote by $\Sigma$), defined by
$\Sig X = \coprod_{i\in I}X^{k_i}$, and an algebra for the signature $\Sig$
with carrier $X$ is thus a pair $\tup{X, \alpha\colon\Sig X \to X}$.
More generally, for any functor $G\colon \Set\to\Set$,
a \emph{$G$-algebra} is a pair
$\tup{X,\alpha\colon GX \to X}$,
and a \emph{$G$-algebra homomorphism} from $\tup{X,\alpha}$ to $\tup{Y,\beta}$
is a map $h\colon X \to Y$ such that $h \circ \alpha = \beta\circ Gh$.
A $G$-algebra $\tup{X,\alpha}$ is \emph{initial} if for any $G$-algebra
$\tup{Y,\beta}$ there is a unique $G$-algebra homomorphism
$h\colon \tup{X,\alpha} \to \tup{Y,\beta}$.
Note that a $\Sig$-algebra (where $\Sig$ is viewed as a functor) is the same as an algebra for $\Sig$ (where $\Sig$ is viewed as a signature).

Monads are functors with extra ``monoid'' structure.
Formally, a \emph{monad} is a triple $\T=\tup{T, \eta, \mu}$ consisting of
a $\Set$-functor $T$, together with natural transformations
$\eta \colon \Id \Rightarrow T$ (the \emph{unit}),
and $\mu \colon TT \Rightarrow T$ (the \emph{multiplication}) such
that $\mu \circ T \eta = \id = \mu \circ \eta_T$ and
$\mu \circ \mu_T = \mu \circ T \mu$.

An \emph{Eilenberg-Moore algebra for the monad}
$\T=\tup{T,\eta,\mu}$ (or just $\T$-algebra for short) is a $T$-algebra
$\tup{X, \alpha}$ that respects the monad structure meaning that
$\alpha \circ \eta_X = \id$ and
$\alpha \circ \mu_X = \alpha \circ T\alpha$.
Note that the latter condition says that $\alpha$ is itself a homomorphism.
A homomorphism of $\T$-algebras is just a homomorphism of $T$-algebras.
An important role is played by $\tup{TX, \mu_X}$ which is the
\emph{free $\T$-algebra}.
Given any $\T$-algebra $\tup{Y, \alpha}$
and any function $f \colon X \rightarrow Y$,
there is a unique $T$-algebra homomorphism
$\ext{f} \colon TX \rightarrow A$  such that
$\ext{f}(\eta(x)) = f(x)$ for all $x \in X$, given by $\alpha \circ Tf$.

We have already encountered several examples of monads.
For a signature $\Sig$,
the mapping $\TSig$ that assigns to a set $X$ the set $\TSig(X)$ of $\Sig$-terms
over $X$ is the (functor part of the) \emph{free monad generated by the functor $\Sig$}.
The unit $\eta_X\colon X \to \TSig(X)$ is inclusion
of variables as terms, and the multiplication
$\mu_X\colon \TSig\TSig(X) \to \TSig(X)$ is the flattening of nested terms
into terms.

Another example of a monad is the construction
$\V$ from Section~\ref{ssec:linear-aut} where $A$ is assumed to be a field.
Recall that $\V(X)$
is the set of all formal linear combinations over $X$,
i.e.,
\[
\V(X) =
\{a_1x_1 + \ldots + a_nx_n \mid a_i \in A, x_i\in X, \;\forall i: 1\leq i\leq n\}\]
First, $\V$ is a functor by defining $\V(f)\colon \V{X} \to \V{Y}$
by $\V(f)(\sum{a_i x_i}) = \sum{b_j y_j}$
where $b_j = \sum_{f(x_i)=y_j}{a_i}$.
The unit $\eta_X\colon X \to \V{X}$
includes variables as the linear combinations: $x \mapsto 1 x$,
and the multiplication $\mu_X \colon\V^2{X} \to \V{X}$ flattens
by distributing scalars over sums as illustrated here for $a,b,c,d,e,f \in A$
and $x,y,z \in X$:
\[ \mu_X(a(cx + dy) + b(ex + fz)) = (ac +be)x + ady + bfz.
\]
The free $\V$-algebra $\tup{\V{X},\mu_X}$ is the vector space
with basis $X$.

Finally, also the construction $\M(X^*)$ of polynomials over $X$
with coefficients in a commutative semiring $A$
from Section~\ref{sec:context-free-specs}
is a monad with unit and multiplication defined in the expected way.
For $A = \bbN$, this was shown in \cite[sec.~3.4]{Jacobs:bialg-dfa-regex}, 
and the proof generalises in a straightforward manner.
As noted already in Section~\ref{sec:context-free-specs},
the free algebra $\M(X^*)$ is again a semiring.

The vector space monad $\V$ and the polynomials monad $\M((-)^*)$
are examples of monads that capture equational theories.
Namely, a variety of algebras defined by a signature $\Sig$ and
equations $E$ is (isomorphic to) the class of Eilenberg-Moore algebras
for the quotient monad $\TSig/\!\equiv_E$ that maps a set $X$ to
$\TSig(X)/\!\equiv_E$ where $\equiv_E$ is the congruence generated by $E$ on $\Sig$-terms.
For example, $\V(X)$ can be viewed as the set of ``linear terms'' defined in
\eqref{eq:linear-grammar} quotiented with the axioms of vector spaces.
Similarly, $\M(X^*)$ is the set of ``polynomial terms'' defined in
\eqref{eq:polynomial-grammar} quotiented with the axioms of unital associative algebras
over a semiring.

\subsection{Bialgebras for a distributive law}
\label{ssec:bialgebras}

The notion of a bialgebra combines coalgebraic and algebraic structure.
The interaction between the two structures should be specified
by a so-called \emph{distributive law}. This definition is rather abstract
at first sight, but we will later see that for a free monad $\monTSig=\tup{\TSig,\eta,\mu}$,
distributive laws involving $\monTSig$ are essentially systems of SDEs.

In the rest of this subsection, we let $\T=\tup{T,\eta,\mu}$ be a monad
and $F$ be a functor, both on $\Set$.
A \emph{distributive law of $\T$ over $F$}
is a natural transformation $\lambda\colon TF \To FT$
that is compatible with the monad structure,
i.e., for all $X$ the following diagrams commute:
\[
\xymatrix@R=1.5em@C=1.2em{
FX \ar[rr]^-{\eta_{FX}} \ar[ddrr]_-{F\eta_{X}}
&& TFX \ar[dd]^-{\lambda_X}
\\
&  \ar@{}[r]^(.4){(\text{unit-}\lambda)} & &&&
\\
&& FTX &&
}
\xymatrix@C=3em@R=3em{
T^2FX \ar[d]_-{\mu_{FX}} \ar[r]^-{T\lambda_X} \ar@{}[drr]|{(\text{mult-}\lambda)}
& TFTX \ar[r]^-{\lambda_{TX}}
& FT^2 X \ar[d]^-{F\mu_X}
\\
TFX \ar[rr]^-{\lambda_X} && FTX
}
\]

A \emph{$\lambda$-bialgebra} is a triple $\tup{X,\alpha,\beta}$
where $\alpha\colon TX\to X$ is a $\T$-algebra and $\beta\colon X \to FX$
is an $F$-coalgebra, and the two are compatible via $\lambda$, i.e.,
the following diagram commutes:
\begin{equation}\label{eq:pentagon}
\xymatrix{
TX \ar[d]_-{T\beta} \ar[r]^-{\alpha} & X \ar[r]^-{\beta} & FX \\
TFX \ar[rr]^-{\lambda_X} && FTX \ar[u]_-{F\alpha}
}\end{equation}
A \emph{morphism of $\lambda$-bialgebras} from $\tup{X_1,\alpha_1,\beta_1}$
to $\tup{X_2,\alpha_2,\beta_2}$ is a function $f\colon X_1 \to X_2$
which is both a $T$-algebra morphism and an $F$-coalgebra morphism.

At present we are mainly interested in the case where $F = A \times (-)$
is the functor of stream automata, and 
we find that a distributive law $\lambda$ of $\T$ over $A \times (-)$
is a natural transformation whose $X$-component has the type
$\lambda_X \colon T(A\times X) \to A \times TX$, and a $\lambda$-bialgebra
has the type $TX \sgoes{\alpha} X \sgoes{\beta} A \times X$.

An important reason why distributive laws yield solutions to systems of
SDEs is that they induce
$\T$-algebraic structure on the final $F$-coalgebra, as we explain now.

Given a distributive law $\lambda$ of $\T$ over $F$,
the functor $F$ lifts to a functor $F_\lambda$ on the category of $\T$-algebras;
and dually the monad $\T$ lifts to a monad
${\T}_\lambda$ on the category of $F$-coalgebras
(cf.~\cite[Lem.~3.4.21]{Bartels:PhD}).
In particular, the functor ${\T}_\lambda$ maps an $F$-coalgebra $\xi\colon X \to FX$ to
the $F$-coalgebra $\lambda_X\circ T\xi\colon TX \to FTX$.
Applying ${\T}_\lambda$ to the final $F$-coalgebra $\tup{Z,\zeta}$,
we obtain an $F$-coalgebra on $TZ$,
and hence by the finality of $\tup{Z,\zeta}$
there is a unique $F$-coalgebra morphism
$\alpha\colon TZ \to Z$.
For the case of stream automata, this is shown in the following diagram:
\begin{equation}\label{eq:monadic-sos}
\xymatrix@C=+15mm{
T(\strms{A}) \ar[r]^-{T\zeta} \ar@{.>}[d]_-{\alpha} &
T(A \times \strms{A}) \ar[r]^-{\lambda_{\strms{A}}} &
A \times T(\strms{A}) \ar@{.>}[d]^-{\id_A\times\alpha}
\\
\strms{A} \ar[rr]^-{\zeta} &&
A \times \strms{A}
}\end{equation}
Furthermore, it can be shown that $\alpha$ is a
$\T$-algebra on $Z$,
and that $\tup{Z,\alpha,\zeta}$
is a final $\lambda$-bialgebra,
see e.g.,~\cite{Bartels:PhD,Klin11} for details.
In short, a distributive law $\lambda$ of $\T$ over $F$
induces a canonical $\T$-algebra on $\tup{Z,\zeta}$.

This leads us to yet another reason why distributive laws and bialgebras are useful. Namely, since $\alpha$ is also a $\T$-algebra homomorphism, the coalgebraic semantics is compositional with respect to $\T$-algebraic structure.
In particular, $F$-bisimilarity is a $\T$-algebra congruence
(cf.~\cite[Thm.~3.2.6]{Bartels:PhD}),
and Lemmas \ref{lem:bisim-is-congruence} and \ref{lem:final-map-is-algebra-hom}, Proposition~\ref{prop:D-solutions} and Theorem~\ref{thm:syntactic-solution}
are thus special instances of
more general results on bialgebras.
Moreover, the presence of a distributive law ensures the
soundness of the enhanced coinduction principle \emph{coinduction-up-to context} (cf.~\cite{Bartels:PhD,RBR:up-to-SOFSEM}) of which Theorem~\ref{thm:up-to} is an instance.


\subsection{The Syntactic Method via Abstract GSOS}
\label{ssec:stream-GSOS}

We now show how SDEs and the syntactic method can be understood 
in terms of abstract GSOS.
The relationship between SDEs and operational rules is described very well in
\cite{Klin11}, and we focus here on a more direct translation fom SDEs to the
abstract GSOS framework in which formats correspond to certain types of
natural transformations.

\subsubsection{Stream differential equations as natural transformations}
\label{ssec:sde-to-nt}

To illustrate, we use the SDEs from Section~\ref{ssec:stream-ops}
that define the constant streams $\cns{a}$, addition $+$ and
convolution product $\times$.
We repeat them here together for convenience:
\begin{equation}\label{eq:sde-again}
  \begin{array}{lclclcl}
    \cns{a}(0) &=& a, & \qquad& \cns{a}' &=& \cns{0} \qquad \text{ for all } a \in \bbR,\\[.5em]
    (\sig+\tau)(0) &=& \sig(0)+\tau(0), &\quad &
    (\sig+\tau)' &=& \sig'+\tau',\\[.5em]
    (\sig\times\tau)(0) &=& \sig(0)\cdot\tau(0), &\quad &
    (\sig\times\tau)' &=& (\sig'\times\tau) + (\cns{\sig(0)}\times\tau')
  \end{array}
\end{equation}
They correspond to stream GSOS definitions for the
signature $\Sigma_\mathsf{ar}=\{\syn{+},\syn{\times}\} \cup \{ \syn{[a]} \mid a \in \bbR \}$ as shown in Example~\ref{ex:arith-gsos-def}, and we repeat them here:
\begin{equation}\label{eq:gsos-arith-again}
\begin{array}{lclclcl}
o_{\syn{[a]}} &=& a, & \; & 
d_{\syn{[a]}} &=& \syn{[0]} \qquad \text{ for all } a \in \bbR,
\\[.8em]
o_{\syn{+}}(a,b) &=& a+b, & \; &
d_{\syn{+}}(a,b) &=& y_1 \;{\syn{+}}\; y_2, 
\\[.8em]
o_{\syn{\times}}(a,b) &=& a \cdot b, & \; &
d_{\syn{\times}}(a,b) &=& 
(y_1 \;{\syn{\times}}\; x_2) \;{\syn{+}}\; (\syn{[a]}\;\syn{\times}\; y_2) 
\end{array}
\end{equation}
where $a,b \in \bbR$ correspond to $\sig(0), \tau(0)$ and
$x_1,y_1,x_2,y_2$ are stream variables
that correspond to $\sig,\sig',\tau,\tau'$.

The connection with abstract GSOS is made by observing that the
definitions in \eqref{eq:gsos-arith-again} correspond to families of
functions:
\begin{equation}\label{eq:rho-arith}
  \begin{array}{rcl}
    (X \times \bbR \times X) & \sgoes{\rho^{\cns{a}}_X} & \bbR \times \TSig(X)\\
    \tup{x_1,a,y_1} & \mapsto & \tup{a, \syn{\cns{0}}}
  \\[1em]
  (X \times \bbR \times X) \times (\bbR \times X) & \sgoes{\rho^+_X} & \bbR \times \TSig(X)  \\
  \tup{\tup{x_1,a,y_1},\tup{x_2,b,y_2}} & \mapsto & \tup{a+b, y_1 \;{\syn{+}}\; y_2}
  \\[1em]
  (X \times \bbR \times X) \times (\bbR \times X) & \sgoes{\rho^\times_X} & \bbR \times \TSig(X)  \\
  \tup{\tup{x_1,a,y_1},\tup{x_2,b,y_2}} & \mapsto & \tup{a \cdot b,(y_1 \;{\syn{\times}}\; x_2) \;{\syn{+}}\; (\syn{[a]}\;\syn{\times}\; y_2)}
\end{array}
\end{equation}
The functor corresponding to the arithmetic signature is
$\Sig_\mathsf{ar}(X) = {X}^A + (X \times X)+ (X \times X)$, and we can combine the above three maps into one $\rho_X = [(\rho^{\cns{a}}_X)^A,\rho^+_X, \rho^\times_X]$ (which applies the relevant component by case distinction on its argument):
\[
\rho_X \colon \Sig_\mathsf{ar}( X \times \bbR \times X)  \;\sgoes{} \; \bbR \times T_{\Sig_\mathsf{ar}}(X)
\]
In general, a stream GSOS definition for a
signature $\Sig$ corresponds to a family of maps $\rho_X$:
\[ \rho_X \colon \Sig( X \times A \times X) \; \sgoes{} \; A \times \TSig(X)\\
\]
which has a component for each $k$-ary $\syn{f} \in \Sig$:
\begin{equation}\label{eq:gsos-stream-def}
\rho_X^{\syn{f}} \colon  \tup{x_1,a_1,y_1}, \ldots, \tup{x_k,a_k,y_k} \; \mapsto \;
\tup{  o_{\syn{f}}(a_1,\ldots,a_k),d_{\syn{f}}(a_1,\ldots,a_k)}
\end{equation}
Notably, $\rho_X$ is defined uniformly in $X$, and in fact,
$\rho$ is a natural transformation of type
\begin{equation}\label{eq:stream-gsos-rule}
\rho \colon \Sig( - \times A \times -) \To A \times \TSig(-)
\end{equation}
This is an instance (with $F = A \times (-)$) of the more general
type of natural transformation $\rho\colon \Sig(\Id\times F) \To F\TSig$.

The reader may have noticed that in the above, $\rho^{\syn{\cns{a}}}_X$ and $\rho^+_X$ do not use the $x$-components of their arguments. In fact, any collection of stream definitions $\sdefi$ for a signature $\Sig$ that do not use the $x$-variable on the right-hand side (i.e., $\sdefi$ is in the SOS-format)
can be expressed by a natural transformation of the simpler type
\begin{equation}\label{eq:stream-sos-rule}
\rho \colon \Sig(A \times -)  \;\To \; A \times T_{\Sig}(-)
\end{equation}
For an arbitrary functor $F$, this would be
a natural transformation $\rho\colon \Sig F \To F\TSig$.

We have thus seen how stream definitions correspond to
natural transformations, and that the types of these natural
transformations correspond to various definition formats such as stream SOS and
stream GSOS.

\subsubsection{From natural transformations to distributive laws}

The following central results in abstract GSOS show that
natural transformations $\rho$ involving a signature $\Sig$
as in the previous subsection,
in fact, determine distributive laws for the free monad $\T_\Sig$.
We start with the relatively simple result for natural transformations
in the SOS-format.

\begin{lem}\label{lem:SOS-format}
Let $\Sig$ be a signature functor, and ${\T}_\Sig = \tup{\TSig,\eta,\mu}$ the free monad over $\Sig$. For any functor $F$, there is a 1-1 correspondence:
\[\begin{array}{ll}
\lambda\colon \TSig F \To F\TSig \quad &\text{distributive law of $\T_\Sig$ over $F$}\\
\hline\hline
\rho\colon \;\;\Sig F \To F\TSig & \text{plain natural transformation}
\end{array}\]
\end{lem}
\proof
This is Lemma 3.4.24(i) of \cite{Bartels:PhD}.
\qed

This correspondence extends to $\rho$ of the type
in \eqref{eq:stream-gsos-rule}
with one small modification,
namely that a natural transformation $\rho\colon \Sig(\Id\times F) \To F\TSig$
induces a distributive law $\lambda$
of $\T_\Sig$ over the
functor $\Id\times F$ such that $\pi_1\circ\lambda=\TSig\pi_1$,
where 
$\pi_1\colon \Id\times F \To \Id$ is the left projection.
We call such a distributive law $\lambda$ a
\emph{GSOS law for $\T_\Sig$ and $F$}.
A GSOS law is also known as a distributive law of the monad $\T_\Sig$
over the \emph{cofree copointed functor over $F$ (given by $\tup{\Id\times F,\pi_1}$)}.
We refer to \cite{LPW2004:cat-sos} or \cite[sec.~3.5.2]{Rot:PhD} for 
further details.

\begin{lem}\label{lem:GSOS-format}
Let $\Sig$ be a signature functor, and $\T_\Sig = \tup{\TSig,\eta,\mu}$ the free monad over $\Sig$. For any functor $F$, there is a 1-1 correspondence,
\[\begin{array}{ll}
\lambda\colon \TSig (\Id\times F) \To (\Id\times F)\TSig \quad &\text{distributive law of $\T_\Sig$ over $\tup{\Id\times F,\pi_1}$}\\
\hline\hline
\rho\colon \;\;\Sig(\Id\times F) \To F\TSig & \text{plain natural transformation}
\end{array}\]
\end{lem}

\proof
See Lemma~3.5.3 of \cite{Rot:PhD} or Lemma 3.5.2(i) of \cite{Bartels:PhD}.
\qed

If $\sdefi$ is a stream GSOS definition with corresponding $\rho$,
we obtain by Lemma~\ref{lem:GSOS-format} a stream GSOS law
$\lambda$,
and for any stream automaton $\beta\colon X \to A\times X$
this $\lambda$ yields a stream automaton structure on
$\TSig(X)$ by
\[\xymatrix@C=+15mm{
\TSig(X) \ar[r]^-{\TSig\tup{\id_X,\beta}} &
\TSig(X \times (A \times X)) \ar[r]^-{\pi_{2}\circ \lambda_X} &
A \times \TSig(X)
}
\]
If we apply this construction to the final stream automaton
$\zeta = \tup{\hd,\tl}\colon \strms{A} \to A \times \strms{A}$, we obtain
precisely the syntactic stream automaton $\tup{\osdefi,\dsdefi}$
from Definition~\ref{defi:sdefi-automaton},
and hence the unique stream automaton homomorphism
$\sem{-}_\sdefi =\ol{\alpha} \colon \TSig(\strms{A}) \to \strms{A}$
by coinduction, as shown in the following diagram:
\begin{equation}\label{eq:monadic-gsos}
\xymatrix@C=+15mm{
\TSig(\strms{A}) \ar[r]^-{\TSig\tup{\id,\zeta}} \ar@{.>}[d]_-{\ol{\alpha}} &
\TSig(\strms{A} \times (A \times \strms{A})) \ar[r]^-{\pi_{2}\circ\lambda_{\strms{A}}} &
A \times \TSig(\strms{A}) \ar@{.>}[d]^-{\id_A\times\ol{\alpha}}
\\
\strms{A} \ar[rr]^-{\zeta} &&
A \times \strms{A}
}
\end{equation}
As in \eqref{eq:monadic-sos}, it can be shown that
$\ol{\alpha}$ is a $\TSig$-algebra homomorphism, and
hence essentially a solution to $\sdefi$.

To summarise, a collection of SDEs that together form a stream GSOS definition
$\sdefi$ 
corresponds to a stream GSOS law $\lambda$,
which yields a stream automaton structure on $\TSig(\strms{A})$, and
hence, by coinduction, a unique interpretation of stream operations in $\Sig$.


\subsection{Solving systems of equations}

We have now seen how the syntactic method is essentially an instance
of the abstract GSOS framework.
We now show that also the solution methods based on coinduction
for stream automata in
Sections~\ref{ssec:linear-aut} and \ref{sec:context-free-solutions}
can be placed in the bialgebraic framework. 
They are, in fact, instances of \emph{$\lambda$-coinduction} as defined
in \cite{Bartels:PhD}.

Recall that linear equation systems are maps of the form 
$\es \colon X \to A \times\V{X}$,
and context-free equation systems are maps of the form
$\es \colon X \to A \times\M({X}^*)$.
More generally, a system of equations for a monad $\T=\tup{T,\eta,\mu}$
and a functor $F$ is a map
$\es \colon X \to FTX$.
If we have a distributive law $\lambda$ of $\T$ over $F$, then
for every equation system $\es \colon X \to FTX$, we can construct
a $\lambda$-bialgebra $\tup{TX,\mu_X,\es_\lambda}$ with
free $\T$-algebra component by taking
$\es_\lambda = F\mu_{X} \circ \lambda_{TX}\circ T\es$,
cf.~\cite[Lemma~4.3.3]{Bartels:PhD}.
We now obtain a unique $\lambda$-bialgebra morphism
$g$ into the final $\lambda$-bialgebra,
as shown here for the stream functor $F = A \times (-)$:
\begin{equation}\label{eq:solutions-bialgebra}
\xymatrix@R=+10mm@C=1em{
 & T^2(X) \ar[d]_-{\mu_X} \ar[rrr]^-{Tg} &&& T(\strms{A}) \ar[d]^-{\alpha}
\\
X \ar[d]_-{\es} \ar[r]^-{\eta_X}  
  & T(X) \ar[dl]^<<<<<{\es_\lambda} \ar[rrr]^-{g} 
  &&& A^\omega \ar[d]^-{\zeta}\\
A \times T(X) \ar[rrrr]^-{\id_A\times g}
  &
  &&& A \times A^\omega
}
\end{equation}
Note that diagrams 
\eqref{eq:linear-solutions} for linear solutions
and \eqref{eq:context-free-solutions} for context-free solutions
are both instances of \eqref{eq:solutions-bialgebra}; except that
the algebra part was left implicit.

In the terminology of \cite{Bartels:PhD}, $e\colon X \to FTX$ is a guarded recursive specification, and the map $g\circ\eta_X\colon X \to \strms{A}$
is a $\lambda$-coiterative arrow, which implies that $g\circ\eta_X$
is the unique solution to $e$, cf.~\cite[Lemma~4.3.4]{Bartels:PhD}.

The above generalises to distributive laws
of monads over cofree copointed functors,
i.e., in particular to GSOS laws,
but the argument is a bit more involved.
Detailed arguments are found in Corollary~4.3.6 and Lemma~4.3.9 from \cite{Bartels:PhD}; 
see also~\cite{Jacobs:bialg-dfa-regex,LPW2004:cat-sos}.


\subsubsection{Distributive laws for non-free monads}
\label{ssec:non-free-monads}

If $\T = \T_\Sig$ is a free monad for a signature $\Sig$,
then $\lambda$ is essentially given by a collection of SDEs
that define $\Sig$-operations.
However, the two monads $\V$ and $\M({-}^*)$ relevant for linear
and context-free systems are not free.
If $\T$ is not free, then we cannot immediately claim the
existence of a $\lambda$ (and hence unique solutions)
by giving a system of SDEs.
However, when $\T$ encodes a variety
of algebras in terms of operations $\Sig$ and equations $E$,
(such as, for example, $\V$ or $\M((-)^*)$), then
$\lambda$ can often be described as a quotient of a law $\lambda_\Sig$ 
that does correspond to a system of SDEs.
In this case, the solution obtained from $\lambda$ coincides
with the solution obtained from $\lambda_\Sig$.
The existence of such a $\lambda$
can be proved by showing that the SDEs defining the operations in $\Sig$
respect the equations in $E$ in a certain sense. 
These results are described in detail in \cite{BHKR:pres-DL}.


\subsubsection{Linear equation systems, revisited}
\label{ssec:linear-revisited}

Let $A$ be a field.
The behaviour functor is the stream automaton functor $F=A\times (-)$,
$\T$ is the vector space monad $\V$ described in
Section~\ref{ssec:algebras-for-monad}.
Let $\lambda\colon \V (A \times (-)) \To A \times \V(-)$
be given by
(cf.~\cite[Thm.~10]{Jacobs06}) :
\begin{equation}\label{eq:linear-pointwise-dl}
\xymatrix@C=+50pt{
\lambda_X : \V(A \times X) \ar[r]^-{\tup{\V\pi_1,\V\pi_2}} &
 \V{A} \times \V {X} \ar[r]^-{\beta\times\id_{\V {X}}} &
 A \times \V {X}
}\end{equation}
where $\beta\colon \V {A} \to A$ is the vector space structure
on the field $A$, and $\pi_1, \pi_2$ denote left and right projection,
respectively. It is straightforward to verify that $\lambda$ is indeed
a distributive law.
Moreover, by working out the details one sees that
the $\V$-algebra (i.e.~vector space structure) induced on $\strms{A}$
by $\lambda$
coincides with the element-wise operations of scalar multiplication
and addition that are also defined by the SDEs. This way of obtaining
a distributive law easily generalises to any stream operation that is defined
element-wise from an operation on $A$.

A linear equation system is a map $e\colon X \to A \times \V{X}$,
and by \eqref{eq:solutions-bialgebra}, a unique solution always exists.
This solution method is essentially the same as linear coinduction.
Namely, for the $\lambda$ in \eqref{eq:linear-pointwise-dl},
$\lambda$-bialgebras are the same as linear automata.

\subsubsection{Context-free equation systems, revisited}

Now we assume that $A$ is a commutative semiring.
The behaviour functor is again the stream functor $F=A\times (-)$
and $\T$ is the polynomial monad $\M((-)^*)$
described at the end of Section~\ref{ssec:algebras-for-monad}.

In order to solve context-free systems using \eqref{eq:solutions-bialgebra},
we need a distributive law for $\M((-)^*)$.
Note, however, that we cannot simply replace $\V$ by $\M((-)^*)$ in
\eqref{eq:linear-pointwise-dl} above, since the desired algebraic structure
on $\strms{A}$ is not an element-wise extension, as in the linear case.
In particular, the convolution product of streams $\strms{A}$ is \emph{not} the
element-wise extension of the semiring product on $A$.
We therefore need a distributive law $\lambda$ of $\M((-)^*)$ over the
cofree copointed functor over $F$.
The existence of such a $\lambda$ is shown in \cite[Example~4.11]{BHKR:pres-DL}
by showing that the SDEs in \eqref{eq:sde-again} respect the semiring axioms,
as explained in Section~\ref{ssec:non-free-monads}.
It follows that
every context-free equation system $e\colon X \to A \times \M(X^*)$
has a unique stream solution.


\section{Discussion and Related Work}
\label{sec:rel-work}

\subsection{Other Specification Methods}

There exist many ways of representing streams, other than by 
stream differential equations. Among the classical methods in mathematics
are recurrence relations, generating functions and continued fractions.
In computer science, weighted automata are also often used 
(cf.~Section~\ref{ssec:linear-aut}).
As a basic and instructive example, we use the stream of  Fibonacci numbers
\[
\phi = (0,1,1,2,3,5,8,13, \ldots)
\]
to quickly illustrate a number of different stream representations.

We already saw a definition of $\phi$ by means of a stream differential equation  (cf. (\ref{eq:fib})):
\begin{equation}
\label{Fibonacci BDE}
\phi(0) = 0 \;\;\;\; \phi(1) = 1 \;\;\;\;\;\; \phi '' = \phi + \phi'
\end{equation}
A definition of $\phi$ by means of a \emph{recurrence relation} is the following:
\begin{equation}
\label{Fibonacci recurrence}
\phi(0) = 0 \;\;\;\; \phi(1) = 1 \;\;\;\;\;\; \phi(n+2) = \phi(n) + \phi(n+1)
\end{equation}
The following representation is called in mathematics a closed form \emph{generating function}:
\begin{equation}
\label{Fibonacci GF}
f(x) = \frac{x}{1-x-x^2}
\end{equation}
It corresponds to the rational expression $\frac{\X}{1-\X-\X^2}$,
which we already saw in (\ref{rational expression for Fibonacci}).
The expansion of $f(x)$  into $f(x) = x+ x^2 + 2x^3 + 3x^4 + 5x^5 + \cdots$ gives us the
Fibonacci numbers.
Finally, the value of the $n$th Fibonacci number can be read from this
weighted automaton (where a state is underlined if its output is 1, otherwise the output is 0)
\begin{equation}
\label{Fibonacci automaton}
\xymatrix
@C+1.5ex
@R-3ex
{
&&\\
\underline{s} \ar@/^1.5ex/[r]^-{1} 
\ar@(l,u)^-{1}
&
t \ar @/^1.5ex/[l]^-{1} 
}
\end{equation}
by counting the number of finite paths of length $n$ leading from the state $s$
back to the state $s$ again.

All is well with this basic example.
All four representations above (and still others)  are well-understood,
including the way to obtain one from the other (as we have seen in 
Section~\ref{sec:linear-specs}). 
But things
get much less clear very quickly.
Consider for instance the stream of factorial numbers $\psi = (0!,1!,2!,3!, \ldots)$.
A recurrence relation is again easily given:
\begin{equation}
\label{factorial recurrence}
\psi(0) =  1 \;\;\;\;\;\; \psi(n+1) = (n+1) \cdot \psi(n)
\end{equation}
but now look at the following stream differential equation, also defining $\psi$:
\begin{equation}
\label{factorial BDE}
\psi(0) =  1 \;\;\;\;\;\; \psi ' = \psi \otimes \psi
\end{equation}
where the righthand side uses the shuffle product
(defined in (\ref{SDE for shuffle})).
It is unclear how (\ref{factorial recurrence})
and (\ref{factorial BDE}) are related.
Furthermore, we know of no closed form generating function
for $\psi$ but then again, there is the following continued fraction:
\begin{equation}
\label{factorial continued fraction}
\psi = \,
\cfrac{1}{1-x - \cfrac{1^2x^2}{1 - 3x - \cfrac{2^2x^2}{1 - 5x - \cfrac{3^2x^2}{\ddots}}}}
\end{equation}
as well as the following representation of $\psi$ by means of an (infinite) weighted automaton
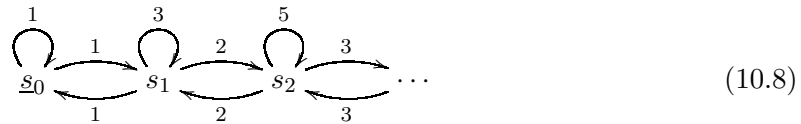
\begin{equation}
\label{factorial automaton}
\xymatrix
@C+1.5ex
@R-2.0ex
{
\underline{s}_0 \ar @/^1.5ex/[r]^-{1}
\ar@(ul,ur)^-1
&
s_1
\ar @/^1.5ex/[r]^-{2}
\ar @/^1.5ex/[l]^-{1}
\ar@(ul,ur)^-3
&
s_2
\ar @/^1.5ex/[r]^-{3}
\ar @/^1.5ex/[l]^-{2}
\ar@(ul,ur)^-5
&
\ar @/^1.5ex/[l]^-{3}
\cdots
}
\end{equation}
For this example, the relation between (\ref{factorial continued fraction})
and (\ref{factorial automaton}) is fairly direct but, more generally,
the relation between all  four different representations
(\ref{factorial recurrence})-(\ref{factorial automaton}) of the factorial numbers
is by no means well-understood, and serves as an illustration of an interesting class of
problems that need further study.


\subsection{Related Work}
We have given an overview of recent results on stream differential equations obtained via
a coalgebraic perspective.
In this subsection we will give pointers to the surveyed literature,
and a brief overview of some related work, which is bound to be incomplete.

\paragraph{\it Formal power series and automata theory}
Streams are formal power series in only one variable and as a consequence,
many of the properties of streams and  stream differential equations
presented here are ultimately special instances of more general facts about
formal power series. We mention \cite{BR11} as a fundamental reference on
formal power series in multiple noncommutative variables,
and refer to \cite{Winter:PhD} for an extensive discussion of the relationship between
the coalgebraic and the classical approaches to streams and formal power series.\enlargethispage{2\baselineskip}

\paragraph{\it Streams and coalgebra}
The coalgebraic treatment of streams, stream differential equations 
and stream calculus started with \cite{Rut01:MFPS-stream-calc,Rut03:TCS-bde}.
Section~\ref{sec:linear-specs} on linear specifications is based on work found 
in the just mentioned papers, as well as further investigations into 
rational streams and linear systems in \cite{Rut:rational,BBBRS:lwa}.
Section~\ref{sec:context-free-specs} on 
context-free specifications is based on
\cite{BRW:CF-pow,BRW:CF-coalg,WBR:CF-LMCS}.
Previously, context-free languages were studied coalgebraically 
in \cite{HJ:CFL}, but using a different approach, 
see \cite[sec.~1.1]{WBR:CF-LMCS} for a discussion.
Section~\ref{sec:non-standard} on non-standard specifications is based on work in \cite{KR:cocoop} and for automatic sequences on \cite{EGHKM,KR:autseq}.
Further work in this direction includes \cite{HKRW:k-regular} on $k$-regular sequences.

Other coalgebraic investigations into streams and stream functions 
include the following.
Specification formats and coalgebraic semantics (as Mealy machines) for
stream functions in 2-adic arithmetic have been 
studied in \cite{Rut05:FACS-short,HR:mealy-SACS}.
Causal stream functions generalise to continuous stream functions,
which have been characterised categorically in \cite{GHP:eating}.

\paragraph{\it Stream circuits}
Linear circuits (or signal flow graphs) are another representation of streams
(which we did not include in our survey). In \cite{Rut05:signal-flow}
it was shown that rational streams are exactly the streams that can be defined by closed linear circuits. An axiomatisation of rational streams in a fixed point calculus was given in \cite{Milius:-lin-streams-LICS}. Recently, the semantics of open linear circuits was given a coalgebraic and algebraic characterisation in \cite{BBHR:sfg}, which leads also to a complete axiomatisation in a calculus of commutative rings and modules.

\paragraph{\it Morphic and automatic sequences}
Yet another way of specifying streams which comes from the field of 
combinatorics on words is as a limit of a (monoid) morphism, see e.g.~\cite[Ch.~10]{Lothaire:AppCoW} and \cite[Ch.~7]{AS03}. A translation between morphic definitions and coinductive definitions was given in \cite[Sec.~2]{endr:hend:bodi:2013}. Coalgebraic characterisations of automatic and regular streams were given in \cite{EGHKM,KR:autseq,HKRW:k-regular}

\paragraph{\it Abstract GSOS}
Abstract GSOS originated as a categorical approach to
structural operational semantics \cite{Aceto:SOS-HB}.
The seminal paper on the topic is 
\cite{TuriPlotkin:LICS-GSOS}, and \cite{Klin11} provides 
an introductory overview, which also contains many examples for streams.
Other rich sources of general results on bialgebras and distributive laws
are 
\cite{Bartels:2003,Bartels:PhD,Klin09,LPW00:distr,LPW2004:cat-sos,Watanabe:cmcs2002}.
See also \cite{Jacobs:bialg-dfa-regex,Jacobs06} for
a bialgebraic treatment of formal languages and regular expressions,
and several other examples.
In \cite{HK:GSOS-buf}, it is shown that a stream GSOS definition $D$
can be transformed into a GSOS defintion $C$ for causal stream functions
that defines the pointwise extensions of the stream operations defined by $D$.

\paragraph{\it Functional programming}
Lazy functional programming languages, such as Haskell, allow programming 
on streams, and leads to many interesting examples and applications 
\cite{DoetsEijck:HR,Hinze11}.
Here it is also of interest to find methods of ensuring that a 
program operating on streams (or, more generally, on codata) is well-defined. 
Specification formats for codata in functional languages have been studied in,
e.g., \cite{Abel:ICFP13,Atcon13}.
Functional programming on non-wellfounded structures such as stream automata 
was studied in \cite{JKS:cocaml}.

\paragraph{\it Term rewriting}
Closely related to functional programming is the work on streams in term rewriting. In particular, the productivity of stream specifications given as term rewrite systems is studied in \cite{EGHIK:productivity-TCS,EH:lazy-productivity,Zantema10}.

\paragraph{\it Tools}
Several tools exist for specifying and reasoning about streams using 
stream differential equations. We mention just a few.
The rewriting-based tool \emph{CIRC~\cite{LucanuR07, LucanuGCR09, CIRC:wiki-webpage}} can check equivalence
of stream specifications (i.e., whether they define the same stream)
using circular coinduction.
The tool \emph{Streambox~\cite{ZantemaE11}} uses more general equational reasoning combined with circular coinduction to prove equivalence of stream specifications.
The Haskell-based tool
\emph{QStream~\cite{Winter13}} provides facilities for entering stream differential equations, and exploring streams together with interfacing with the OEIS~\cite{OEIS}.


\bibliographystyle{plain}
\bibliography{sde}


\end{document}